\documentclass[11pt]{article}
\usepackage[hmargin={2.8cm,2.8cm},vmargin={2.8cm,2.8cm}]{geometry}
\usepackage{float}
\usepackage{amsmath,amssymb}
\usepackage{amsfonts}
\usepackage{mathrsfs}
\usepackage{mathtools}
\usepackage{dsfont}

\usepackage{pgf}
\usepackage{tikz}
\usetikzlibrary{arrows.meta}
\usetikzlibrary{decorations.markings}
\usetikzlibrary{calc}
\usepackage{fp} 

\usepackage{bm}
\usepackage{empheq}
\usepackage{multirow}

\usepackage[skip=2pt,font=small]{caption}
\usepackage{hyperref}
\hypersetup{
colorlinks,
linkcolor={blue!60!black},
citecolor={blue!90!black},
urlcolor={blue!90!black}}

\usepackage{footnote}

\usepackage{cite}
\usepackage{enumerate}
\usepackage[curve]{xy}

 \usepackage{amsthm}

\usepackage{titlesec}

\titleformat{\paragraph}
{\normalfont\normalsize\bfseries}{\theparagraph}{1em}{}
\titlespacing*{\paragraph}
{0pt}{3.25ex plus 1ex minus .2ex}{1.5ex plus .2ex}

\usepackage[titletoc,title]{appendix}

\newtheorem{thm}{Theorem}

\theoremstyle{remark}

\theoremstyle{definition}

\numberwithin{equation}{section}

\tikzset{->-/.style={decoration={
  markings,
  mark=at position #1 with {\arrow{Latex}}},postaction={decorate}}}
\tikzset{->>-/.style={decoration={
  markings,
  mark=at position .5 with {\arrow{Latex[sep=10pt]Latex}}},postaction={decorate}}}
\tikzset{-<-/.style={decoration={
  markings,
  mark=at position #1 with {\arrow{Latex[reversed]}}},postaction={decorate}}}
\tikzset{-<<-/.style={decoration={
  markings,
  mark=at position .5 with {\arrow{Latex[reversed,sep=-10pt]Latex[reversed]}}},postaction={decorate}}}
  
\tikzset{-|-/.style={decoration={
  markings,
  mark=at position .51 with {\arrow{Bar}}},postaction={decorate}}}
\tikzset{-||-/.style={decoration={
  markings,
  mark=at position .49 with {\arrow{Bar[sep=-5pt] Bar}}},postaction={decorate}}}
\tikzset{-!-/.style={decoration={
  markings,
  mark=at position #1 with {\arrow{Rays[n=6,length=6pt]}}},postaction={decorate}}}
\tikzset{-!!-/.style={decoration={
  markings,
  mark=at position .51 with {\arrow{Bar[sep=1pt,length=4pt] Rays[n=6,length=8pt]}}},postaction={decorate}}}

\newcommand{\ds}{\displaystyle}

\renewcommand{\author}[1]{\large\rm #1\\ \bigskip}
\newcommand{\address}[1]{{\normalsize\it #1\\}\bigskip}
\renewcommand{\title}[1]{\bigskip\bigskip\Large\bf #1\bigskip\bigskip\\}

\def\EXP{\textrm{{\large e}}}

\newcommand{\x}{{\boldsymbol{x}}}
\newcommand{\y}{{\boldsymbol{y}}}
\newcommand{\bal}{{\bm{\alpha}}}
\newcommand{\bbt}{{\bm{\beta}}}

\newcommand{\al}{{{\alpha}}}
\newcommand{\bt}{{{\beta}}}
\newcommand{\gm}{{{\gamma}}}

\newcommand{\bbu}{{\bm{u}}}
\newcommand{\bbv}{{\bm{v}}}
\newcommand{\bbw}{{\bm{w}}}
\newcommand{\bbz}{{\bm{z}}}
\newcommand{\bbx}{{\bm{x}}}
\newcommand{\bby}{{\bm{y}}}

\newcommand{\oW}{\overline{W}}
\newcommand{\oV}{\overline{V}}

\newcommand{\lag}{{\mathcal \L}}
\newcommand{\ol}{\overline{\lag}}
\newcommand{\lam}{{\Lambda}}
\newcommand{\olam}{\overline{\lam}}

\newcommand{\olamh}{\hat{\olamh}}

\newcommand{\xv}{{x}}
\newcommand{\yv}{{y}}
\newcommand{\yva}{{y_1}}
\newcommand{\yvb}{{y_2}}
\newcommand{\dxv}{\dot{\xv}}

\newcommand{\pdz}{p_0}
\newcommand{\pda}{p_1}
\newcommand{\pdb}{p_2}
\newcommand{\pdc}{p_3}
\newcommand{\pdd}{p_4}

\newcommand{\pqaa}{\al}
\newcommand{\pqab}{\al-\gm}
\newcommand{\pqba}{\bt}
\newcommand{\pqbb}{\bt-\gm}

\newcommand{\ta}{q_1}
\newcommand{\tb}{q_2}
\newcommand{\tc}{q_3}
\newcommand{\td}{q_4}

\newcommand{\pa}{\al_1}
\newcommand{\pb}{\al_2}
\newcommand{\pc}{\al_3}
\newcommand{\pd}{\al_4}

\newcommand{\dwp}{{\wp'}}
\newcommand{\wes}{q}

\newcommand{\dwepa}{\dwp(\pa)}
\newcommand{\dwepb}{\dwp(\pb)}
\newcommand{\dwepc}{\dwp(\pc)}
\newcommand{\dwepd}{\dwp(\pd)}

\newcommand{\spn}{\sigma}
\newcommand{\bspn}{\bm{\sigma}}
\newcommand{\bu}{\bm{u}}
\newcommand{\bv}{\bm{v}}

\newcommand{\ccx}{x}

\newcommand{\xa}{x_a}
\newcommand{\xb}{x_b}
\newcommand{\xc}{x_e}
\newcommand{\xd}{x_d}
\newcommand{\xe}{x_f}
\newcommand{\xf}{x_c}
\newcommand{\ya}{y_a}
\newcommand{\yb}{y_b}
\newcommand{\yc}{y_e}
\newcommand{\yd}{y_d}
\newcommand{\ye}{y_f}
\newcommand{\yf}{y_c}
\newcommand{\za}{z_a}
\newcommand{\zb}{z_b}
\newcommand{\zc}{z_e}
\newcommand{\zd}{z_d}
\newcommand{\ze}{z_f}
\newcommand{\zf}{z_c}
\newcommand{\wa}{w_a}
\newcommand{\wb}{w_b}
\newcommand{\wc}{w_e}
\newcommand{\wdd}{w_d}
\newcommand{\we}{w_f}
\newcommand{\wf}{w_c}
\newcommand{\ua}{u_a}
\newcommand{\ub}{u_b}
\newcommand{\uc}{u_e}
\newcommand{\ud}{u_d}
\newcommand{\ue}{u_f}
\newcommand{\uf}{u_c}
\newcommand{\va}{v_a}
\newcommand{\vb}{v_b}
\newcommand{\vc}{v_e}
\newcommand{\vd}{v_d}
\newcommand{\ve}{v_f}
\newcommand{\vf}{v_c}

\newcommand{\rhob}{{\overline{\rho}}}
\newcommand{\rhot}{{\hat{\rho}}}

\renewcommand{\L}{L}

\newcommand{\A}[8]{A(#1;#2,#3,#4,#5;#6,#7,#8)}
\newcommand{\C}[8]{C(#1;#2,#3,#4,#5;#6,#7,#8)}
\newcommand{\Cbar}[8]{\overline{C}(#1;#2,#3,#4,#5;#6,#7,#8)}

\newcommand{\At}[7]{A(#1;#2,#3,#4,#5;#6,#7)}

\newcommand{\Ct}[7]{C(#1;#2,#3,#4,#5;#6,#7)}
\newcommand{\Cft}[7]{\overline{C}(#1;#2,#3,#4,#5;#6,#7)}

\newcommand{\Aa}{A.a}
\newcommand{\Ab}{A.b}
\newcommand{\Af}{A.c}
\newcommand{\Ad}{A.d}
\newcommand{\Ac}{A.e}
\newcommand{\Ae}{A.f}

\newcommand{\Ca}{C.a}
\newcommand{\Cb}{C.b}
\newcommand{\Cf}{C.c}
\newcommand{\Cd}{C.d}
\newcommand{\Cc}{C.e}
\newcommand{\Ce}{C.f}

\newcommand{\Qa}{Q.a}
\newcommand{\Qb}{Q.b}
\newcommand{\Qc}{Q.c}
\newcommand{\Qd}{Q.d}
\newcommand{\Qi}{Q.i}

\newcommand{\Ha}{H.a}
\newcommand{\Hb}{H.b}
\newcommand{\Hc}{H.c}
\newcommand{\Hd}{H.d}
\newcommand{\Hi}{H.i}

\newcommand{\sh}[1]{s(#1)}
\newcommand{\zz}[1]{z(#1)}

\newcommand{\hexAneq}[9]{\bm{A}(#1,#2,#3,#4,#5,#6;#7,#8,#9)}
\newcommand{\hexCneq}[9]{\bm{C}(#1,#2,#3,#4,#5,#6;#7,#8,#9)}
\newcommand{\hexA}[9]{\bm{A}(#1,#2,#3,#4,#5,#6;#7,#8,#9)=0}
\newcommand{\hexC}[9]{\bm{C}(#1,#2,#3,#4,#5,#6;#7,#8,#9)=0}

\newcommand{\prsm}[6]{\bm{P}(#1,#2;#3,#4,#5,#6)=0}

\newcommand{\quadQ}[6]{Q(#1,#2,#3,#4;#5,#6)}
\newcommand{\quadQs}[6]{Q^\ast(#1,#2,#3,#4;#5,#6)}

\newcommand{\quadH}[6]{H(#1,#2,#3,#4;#5,#6)}
\newcommand{\quadHu}[6]{\mathop{H}\limits^\ast(#1,#2,#3,#4;#5,#6)}

\newcommand{\quadHun}{\mathop{H}\limits^\ast}

\definecolor{cof}{RGB}{219,144,71}
\definecolor{pur}{RGB}{186,146,162}
\definecolor{greeo}{RGB}{91,173,69}
\definecolor{greet}{RGB}{52,111,72}

\begin{document}

\vglue 2cm

\begin{center}

\title{Integrable systems on hexagonal lattices and consistency on polytopes with quadrilateral and hexagonal faces}
\author{Andrew P.~Kels}
\address{Scuola Internazionale Superiore di Studi Avanzati,\\ Via Bonomea 265, 34136 Trieste, Italy}

\end{center}

\begin{abstract}

The new concept of a system of hex equations is introduced as an overdetermined system of six five-point face-centered quad equations defined on six vertices of a hexagon.   For a consistent system of hex equations,  two variables on neighbouring vertices of the hexagon can be solved for uniquely in terms of the other four.  A consistent system of hex equations has a well-defined unique evolution in the hexagonal lattice under suitable initial value problems defined on a single connected staircase of points.  Multidimensional consistency for systems of hex equations is proposed in terms of their consistency on certain polytopes which have both hexagonal and quadrilateral faces, and specific examples are presented for the hexagonal prism, the elongated dodecahedron, the truncated octahedron, and the 6-6-duoprism.  Consistent systems of hex equations on such polytopes may be constructed from face-centered quad equations which satisfy consistency-around-a-face-centered-cube, in combination with regular quad equations that satisfy consistency-around-a-cube.

\end{abstract}





\section{Introduction}

The Yang-Baxter equation is an important equation for integrable models and mathematical physics.  
One of its main applications is for integrablity of lattice models of statistical mechanics, where it implies that the transfer matrices of the model commute which can be used to solve the model in the thermodynamic limit \cite{Baxter:1982zz}.  For Ising-type lattice models with nearest-neighbour interactions involving discrete spin variables, the Yang-Baxter equation takes a special form known as the star-triangle relation \cite{Onsager,Zamolodchikov:1980mb,Fateev:1982wi,Kashiwara:1986tu,AuYang:1987zc,Baxter:1987eq,AuYangPerkSTRreview,Bazhanov:2016ajm}.  

The quasi-classical limit of such star-triangle relations also contains defining equations for an entirely different class of integrable models, known as discrete integrable systems.  Namely, the saddle-point equations of the star-triangle relations in the quasi-classical limit have been shown \cite{Bazhanov:2007mh,Bazhanov:2010kz,Bazhanov:2016ajm,Kels:2018xge} to correspond to equations that arise in the Adler, Bobenko, and Suris (ABS) classification of integrable quad equations \cite{ABS,ABS2}.  These are integrable two-dimensional partial difference equations that are defined on vertices and edges of a square (see the diagram on the left of Figure \ref{fig-intro}) and evolve in the square lattice.   Such difference equations provide integrable discrete counterparts of well-known integrable soliton equations, such as the Korteweg-de Vries equation, where the latter may be obtained through some continuous limits.  The quad equations in the ABS list satisfy a condition of integrability known as consistency-around-a-cube (CAC) \cite{nijhoffwalker,BobSurQuadGraphs}, which requires the consistency of an overdetermined system of six quad equations for four unknowns on the cube, and this property in turn implies Lax pairs and B\"acklund transformations for the equations.  For quad equations there were also found classical counterparts of the star-triangle relations, independently of the quasi-classical limit, in the form of a closure relation that was introduced for the concept of Lagrangian multiforms \cite{LobbNijhoff,BS09}.

\begin{figure}[htb!]
\centering
\begin{tikzpicture}[scale=1.15]

\begin{scope}[scale=0.6]

\fill[white!] (1,1) circle (0.01pt)
node[left=1pt]{\color{black} $\al$};
\fill[white!] (5,1) circle (0.01pt)
node[right=1pt]{\color{black} $\al$};

\fill[white!] (3,3) circle (0.01pt)
node[above=1pt]{\color{black} $\bt$};
\fill[white!] (3,-1) circle (0.01pt)
node[below=1pt]{\color{black} $\bt$};

\draw[-] (5,-1)--(5,3)--(1,3)--(1,-1)--(5,-1);

\fill (1,-1) circle (3.5pt)
node[left=1.5pt]{\color{black} $x_c$};
\filldraw[fill=black,draw=black] (1,3) circle (3.5pt)
node[left=1.5pt]{\color{black} $x_a$};
\fill (5,3) circle (3.5pt)
node[right=1.5pt]{\color{black} $x_b$};
\filldraw[fill=black,draw=black] (5,-1) circle (3.5pt)
node[right=1.5pt]{\color{black} $x_d$};

\end{scope}

\begin{scope}[scale=0.6,xshift=240pt]

\fill[white!] (1,1) circle (0.01pt)
node[left=1pt]{\color{black}\small $(\al_1,\al_2)$};
\fill[white!] (5,1) circle (0.01pt)
node[right=1pt]{\color{black}\small $(\al_1,\al_2)$};

\fill[white!] (3,3) circle (0.01pt)
node[above=1pt]{\color{black}\small $(\bt_1,\bt_2)$};
\fill[white!] (3,-1) circle (0.01pt)
node[below=1pt]{\color{black}\small $(\bt_1,\bt_2)$};

\draw[-,gray,dashed] (5,-1)--(5,3)--(1,3)--(1,-1)--(5,-1);
\draw[-] (5,-1)--(1,3); \draw[-] (5,3)--(1,-1);

\fill (3,1) circle (3.5pt)
node[left=2.5pt]{\color{black} $\ccx$};
\fill (1,-1) circle (3.5pt)
node[left=1.5pt]{\color{black} $x_c$};
\filldraw[fill=black,draw=black] (1,3) circle (3.5pt)
node[left=1.5pt]{\color{black} $x_a$};
\fill (5,3) circle (3.5pt)
node[right=1.5pt]{\color{black} $x_b$};
\filldraw[fill=black,draw=black] (5,-1) circle (3.5pt)
node[right=1.5pt]{\color{black} $x_d$};

\end{scope}

\begin{scope}[scale=1.4,xshift=240pt,yshift=13pt]

\pgfmathsetmacro\tr{sqrt(3)/2}

\coordinate (A1) at (-1,0);
\coordinate (A2) at (-1/2,{\tr});
\coordinate (A3) at (1/2,{\tr});
\coordinate (A4) at (1,0);
\coordinate (A5) at (1/2,{-\tr});
\coordinate (A6) at (-1/2,{-\tr});

\draw (A1)--(A2) node[midway,above left]{$\gamma$};
\draw (A2)--(A3) node[midway,above]{$\alpha$};
\draw (A3)--(A4) node[midway,above right]{$\beta$};
\draw (A4)--(A5) node[midway,below right]{$\gamma$};
\draw (A5)--(A6) node[midway,below]{$\alpha$};
\draw (A6)--(A1) node[midway,below left]{$\beta$};


\fill (A1) circle (1.4pt) node[left]{$\xa$};
\fill (A2) circle (1.4pt) node[above left]{$\xb$};
\fill (A3) circle (1.4pt) node[above right]{$\xf$};
\fill (A4) circle (1.4pt) node[right]{$\xd$};
\fill (A5) circle (1.4pt) node[below right]{$\xc$};
\fill (A6) circle (1.4pt) node[below left]{$\xe$};

\end{scope}

 \end{tikzpicture}
 
 \caption{Variables and parameters of a quad equation (left), a face-centered quad equation (center), and a system of hex equations (right).}
 \label{fig-intro}

\end{figure}
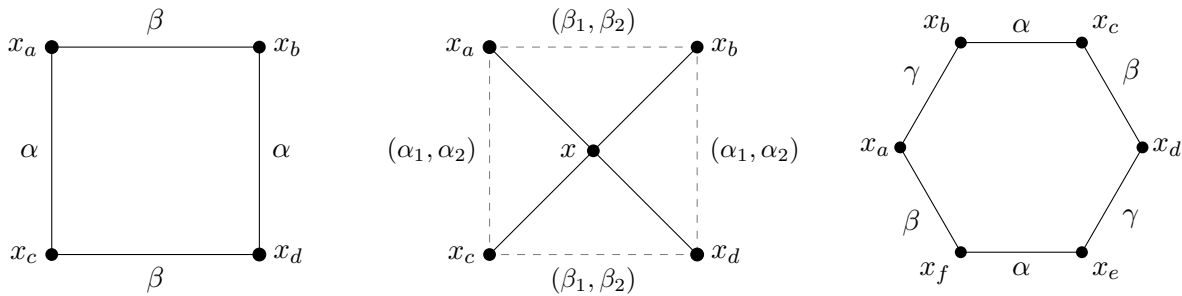

Another type of Yang-Baxter equation for interaction-round-a-face (IRF) models constructed from the star-star relation \cite{Bazhanov:1992jqa,Baxter:1997tn,Bazhanov:2011mz} was used to develop the concept of face-centered quad equations \cite{Kels:2020zjn}.  A face-centered quad equation may be regarded as an extension of a regular quad equation, having dependence on five variables assigned to vertices on the face and corners of a square, and two-component parameters associated to edges (see the diagram in the center of Figure \ref{fig-intro}).  The face-centered quad equations also evolve in the square lattice, but in a more complicated manner, requiring initial data that is defined on double staircases \cite{GubbiottiKels} in comparison to the single staircases required for regular quad equations \cite{AdlerVeselovIVPs,Viallet2006,KampIVPlattice}.  An analogue of CAC was proposed for face-centered quad equations called consistency-around-a-face-centered-cube (CAFCC), requiring the consistency of an overdetermined system of fourteen equations for eight unknowns on the face-centered cubic unit cell, and this property has also been shown to imply Lax pairs of the equations \cite{KelsLax,KelsLax2}.  

The aim of this paper is to further build on the above ideas to develop the analogues of integrable quad equations and their properties for systems of equations that evolve in hexagonal lattices, rather than square lattices.  In comparison to the theory of integrable quad equations, there appears to be only a handful of results for discrete integrable systems in hexagonal lattices or in hexagonal configurations \cite{DateHexagon,NijhoffHexagon,NovikovDynnikov,AgafonovBobenkoCircles,BHSHexagonalCircles,BHHexagonalCircles,DoliwaHoneycomb,DoliwaHoneycomb2,ABSoctahedron,KampIVPquad,JKMNHexagonal}, and to the best of the author's knowledge this paper offers a different approach and new results for this topic.  To begin with, instead of a square, the equations are to be defined on a hexagonal unit cell and have dependence on six variables assigned to vertices and three parameters assigned to edges, where the same parameter is assigned to opposite edges (see the diagram on the right of Figure \ref{fig-intro}), and under suitable initial value problems (including for initial data on a single connected path) these equations should have a unique evolution in the hexagonal lattice.  Also, as an analogue of CAC and CAFCC the equations are to form a consistent system of equations when assigned to polytopes that contain hexagonal faces.

Such consistent systems of equations on hexagons will be constructed in this paper using face-centered quad equations as the basic building block.  A new form of consistency condition for face-centered quad equations will be introduced in Section \ref{sec:quads}, called consistency-around-a-hexagon (CAH).  CAH is a simpler consistency condition than CAFCC, requiring the consistency of an overdetermined system of six equations for two unknown variables on the hexagon.  
Since they are defined on a hexagonal unit cell rather than a face-centered cubic structure, the system of six face-centered quad equations on the hexagon will be referred to in this paper as {\it a system of hex equations}.  For hex equations that satisfy CAH, two unknown variables on the hexagon can be solved for consistently in terms of four known variables, and this will be enough to uniquely evolve systems of hex equations in the hexagonal lattice under initial conditions defined on a single connected staircase of points.  

The previously found face-centered quad equations which satisfy CAFCC (specifically type-A and type-C equations) \cite{Kels:2020zjn} also satisfy CAH, and thus may be used to define consistent systems of hex equations.  Some general guidelines will be proposed in Section \ref{sec:polytope} for constructing consistent systems of hex equations on polytopes in combination with regular quad equations.  These guidelines are based on matching the edges associated to the three-leg and four-leg equations associated to quad and face-centered quad equations respectively.  
In particular, the consistency of equations on the hexagonal prism provides a natural extension of a two-dimensional system of hex equations in the hexagonal lattice onto the hexagonal prismatic honeycomb, which is the three-dimensional tessellation obtained by translations of the hexagonal prism.  This may be regarded as the analogue of extending systems of regular quad equations from the lattice $\mathbb{Z}^2$ to the lattice $\mathbb{Z}^3$ using their consistency on the cube (a uniform square prism).   Besides the hexagonal prism, explicit examples will be presented for consistent systems of equations in three dimensions on the elongated dodecahedron and the truncated octahedron, and in four dimensions on the 6-6-duoprism, which is a polytope whose cells are hexagonal prisms.

\section{Integrable quad and face-centered quad equations}\label{sec:quads}

This section will provide an introduction to the regular quad equations which satisfy consistency-around-a-cube (CAC), and the face-centered quad equations which satisfy consistency-around-a-face-centered-cube (CAFCC), as well as the relevant properties of the two types of equations and their connections to star-triangle and star-star relations which has motivated this paper.  The face-centered quad equations will play a central role in this paper for the construction of consistent systems of equations that evolve in hexagonal lattices (see Section \ref{sec:hexagon}), while both types of quad equations will be required to form consistent systems of equations on polytopes that have both quadrilateral and hexagonal faces (see Section \ref{sec:polytope}).

\subsection{Quad equations and star-triangle relations} 

A quad equation may be written as
\begin{equation}\label{afflinquad}
\quadQ{x_a}{x_b}{x_c}{x_d}{\alpha}{\beta}=0,
\end{equation}
where $Q$ is a multilinear polynomial in four variables $x_a,x_b,x_c,x_d$, that takes the general form
\begin{equation}\label{afflinquadpoly}
\begin{split}
\quadQ{x_a}{x_b}{x_c}{x_d}{\alpha}{\beta}
=\kappa_1x_ax_bx_cx_d+\kappa_2x_ax_bx_c+\kappa_3x_ax_bx_d+\kappa_4x_ax_cx_d+\kappa_5x_bx_cx_d\phantom{.}&
\\
+\kappa_6x_ax_b+\kappa_7x_ax_c+\kappa_8x_ax_d+\kappa_9x_bx_c+\kappa_{10}x_bx_d+\kappa_{11}x_cx_d\phantom{.}& \\
+\kappa_{12}x_a+\kappa_{13}x_b+\kappa_{14}x_c+\kappa_{15}x_d+\kappa_{16}.&
\end{split}
\end{equation}
The coefficients $\kappa_i=\kappa_i(\alpha,\beta)$ ($i=1,\ldots,16$) depend only on the parameters $\alpha$ and $\beta$.  The variables and parameters of the quad equation \eqref{afflinquad} may be assigned to the square of Figure \ref{fig:quadQ}.  Under appropriate initial conditions the quad equations \eqref{afflinquad} can be uniquely evolved in the square lattice due to the multilinearity of the polynomial $Q$.

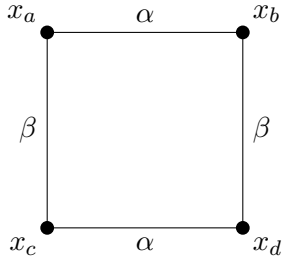
\begin{figure}[htb!]
\centering
\begin{tikzpicture}[scale=1.3]

\coordinate (A1) at (-1,-1);
\coordinate (A2) at (-1,1);
\coordinate (A3) at (1,1);
\coordinate (A4) at (1,-1);

\draw (A1)--(A2) node[midway,left]{$\beta$};
\draw (A2)--(A3) node[midway,above]{$\alpha$};
\draw (A3)--(A4) node[midway,right]{$\beta$};
\draw (A4)--(A1) node[midway,below]{$\alpha$}node[midway,below=20pt]{$\quadQ{x_a}{x_b}{x_c}{x_d}{\alpha}{\beta}=0$};

\fill (A1) circle (2pt) node[below left]{$x_c$};
\fill (A2) circle (2pt) node[above left]{$x_a$};
\fill (A3) circle (2pt) node[above right]{$x_b$};
\fill (A4) circle (2pt) node[below right]{$x_d$};

\end{tikzpicture}
\caption{The quad equation \eqref{afflinquad}.}
\label{fig:quadQ}
\end{figure}

Quad equations of the form \eqref{afflinquad} have been shown \cite{Bazhanov:2007vg,Bazhanov:2010kz,Bazhanov:2016ajm,Kels:2017vbc,Kels:2018xge} to arise from the quasi-classical limit of more general equations for hypergeometric integrals which take the form of an important equation for integrability of lattice models of statistical mechanics called the star-triangle relation.  The quad equations that are obtained this way may be identified with equations in the Adler, Bobenko, and Suris (ABS) classification \cite{ABS,ABS2}. The connection to these types of star-triangle relations is summarised as follows.

In certain cases, hypergeometric integrals satisfy formulas that may be written in one of the following two typical forms \cite{Kels:2018xge}
\begin{align}\label{STRQ}
\int_Cd\spn_0 \oW_{\theta_1}(\spn_1,\spn_0)W_{\theta_1+\theta_3}(\spn_2,\spn_0)\oW_{\theta_3}(\spn_0,\spn_3)
&\!=\!R_{\theta}W_{\theta_1}(\spn_2,\spn_3)\oW_{\theta_1+\theta_3}(\spn_1,\spn_3)W_{\theta_3}(\spn_2,\spn_1), \\
\label{STRH}
\int_Cd\spn_0 \oV_{\theta_1}(\spn_1,\spn_0)V_{\theta_1+\theta_3}(\spn_2,\spn_0)\oW_{\theta_3}(\spn_0,\spn_3)
&\!=\!R_{\theta}V_{\theta_1}(\spn_2,\spn_3)\oV_{\theta_1+\theta_3}(\spn_1,\spn_3)W_{\theta_3}(\spn_2,\spn_1).
\end{align}
Each of $V_\theta(\spn_i,\spn_j),\oV_\theta(\spn_i,\spn_j),W_\theta(\spn_i,\spn_j),\oW_\theta(\spn_i,\spn_j)$, are complex-valued functions of $\spn_i,\spn_j,\theta$, and the formulas \eqref{STRQ} and \eqref{STRH} depend on the independent complex-valued variables $\spn_1,\spn_2,\spn_3$, and parameters $\theta_1,\theta_3$.  The  $R_{\theta}$ is a factor that is independent of the variables $\spn_1,\spn_2,\spn_3$.  The functions $V_\theta(\spn_i,\spn_j),\oV_\theta(\spn_i,\spn_j),W_\theta(\spn_i,\spn_j),\oW_\theta(\spn_i,\spn_j)$, are typically written in terms of the gamma function or its generalisations \cite{Ruijsenaars:1997:FOA}, with the contours $C$ chosen to separate infinite sequences of poles of the integrand that go to zero and/or infinity. Such formulas may be regarded as generalisations of the Euler beta function in the context of hypergeometric integrals, and generalisations of the Ising model in the context of integrable lattice models of statistical mechanics, where both of the latter may be obtained in certain limits \cite{Bazhanov:2010kz}.  For models of statistical mechanics, the functions $V_\theta(\spn_i,\spn_j),\oV_\theta(\spn_i,\spn_j),W_\theta(\spn_i,\spn_j),\oW_\theta(\spn_i,\spn_j)$, are identified as the Boltzmann weights of the associated model(s) \cite{Bazhanov:2007mh,Bazhanov:2010kz,Bazhanov:2016ajm}.

In a quasi-classical limit, the integrals of \eqref{STRQ} and \eqref{STRH} are written in the respective forms
\begin{equation}\label{STRqcl}
\begin{gathered}
\int_Cdx_0\, \EXP^{\hbar^{-1}\bigl(\ol_{\alpha_1}(x_1,x_0)+\lag_{\alpha_1+\alpha_3}(x_2,x_0)+\ol_{\alpha_3}(x_0,x_3)\bigr)+O(1)}, \\
\int_Cdx_0\, \EXP^{\hbar^{-1}\bigl(\olam_{\alpha_1}(x_1,x_0)+\lam_{\alpha_1+\alpha_3}(x_2,x_0)+\ol_{\alpha_3}(x_0,x_3)\bigr)+O(1)},
\end{gathered}
\end{equation}
for some parameter $\hbar\to0$, where each of $\lam_\alpha(x_i,x_j),\olam_\alpha(x_i,x_j),\lag_\alpha(x_i,x_j),\ol_\alpha(x_i,x_j)$, are complex-valued functions of $x_i,x_j,\alpha$, which arise as the leading $O(\hbar^{-1})$ asymptotics of the Boltzmann weights $V_\theta(x_i,x_j),\oV_\theta(x_i,x_j),W_\theta(x_i,x_j),\oW_\theta(x_i,x_j)$, respectively.  The variables and parameters of \eqref{STRQ}, \eqref{STRH}, and \eqref{STRqcl}, are related by
\begin{equation}\label{YBEcov}
\spn_i=F_i(x_i),\quad i=0,1,2,3,\qquad \theta_j=G(\alpha_j),\quad j=1,3,
\end{equation}
where $F_i(z)$ and $G(z)$ are some M\"obius transformations that depend on $\hbar$.  The saddle-point equations of these integrals are given by
\begin{equation}\label{STRsaddlepoint}
\begin{split}
\frac{\partial}{\partial x_0}\bigl(\ol_{\alpha_1}(x_1,x_0)+\lag_{\alpha_1+\alpha_3}(x_2,x_0)+\ol_{\alpha_3}(x_0,x_3)\bigr)=0, \\
\frac{\partial}{\partial x_0}\bigl(\olam_{\alpha_1}(x_1,x_0)+\lam_{\alpha_1+\alpha_3}(x_2,x_0)+\ol_{\alpha_3}(x_0,x_3)\bigr)=0.
\end{split}
\end{equation}
A discrete integrable equation is typically then obtained directly from these saddle-point equations, which are identified as the three-leg equations \cite{BobSurQuadGraphs,ABS} associated to quad equations.  A point transformation
\begin{equation}\label{quadcov}
x_a=f_0(x_0),\quad x_b=f_1(x_1),\quad x_c=f_2(x_2),\quad x_d=f_3(x_3),\qquad
\alpha=g(\alpha_1),\quad \beta=g(\alpha_3),
\end{equation}
of the variables and parameters should be used that will turn the saddle point equations \eqref{STRsaddlepoint} into equations of the form
\begin{equation}
\label{3legmult}
\begin{split}
\frac{a(x_a;x_c;\beta)}{a(x_a;x_b;\alpha)a(x_a;x_d;{\beta-\alpha})}=1, \\
\frac{c(x_a;x_c;\beta)}{c(x_a;x_b;\alpha)a(x_a;x_d;{\beta-\alpha})}=1,
\end{split}
\end{equation}
or the form
\begin{equation}
\label{3legadd}
\begin{split}
a(x_a;x_c;\beta)-a(x_a;x_b;\alpha)-a(x_a;x_d;{\beta-\alpha})=0, \\
c(x_a;x_c;\beta)-c(x_a;x_b;\alpha)-a(x_a;x_d;{\beta-\alpha})=0,
\end{split}
\end{equation}
where the functions $a(x;y;\alpha)$ and $c(x;y;\alpha)$ are both linear fractional functions of the variable $y$. Typical choices of the functions $f_i$ and $g$ for \eqref{quadcov} are given in Table \ref{tab:cov}, which depend on whether the equations \eqref{STRQ}--\eqref{STRsaddlepoint} are elliptic, hyperbolic, rational, or algebraic (also known as classical). The algebraic cases give equations of the form \eqref{3legadd}, and all other cases give equations of the form \eqref{3legmult}.

\begin{table}[htb!]
\centering
\begin{tabular}{c|c|c}
 & $f_i(x)$ & $g(x)$ \\[0.0cm]
\hline \\[-0.4cm]
Elliptic & $\wp(x)$ & $\wp(x)$ \\[0.05cm]
Hyperbolic & $\EXP^{x}$ or $\cosh(x)$ & $\EXP^{x}$  \\[0.05cm]
Rational & $x$ or $x^2$ & $x$ \\[0.05cm]
Algebraic & $x$ & $x$ \\[0.05cm]
\hline 
\end{tabular}
\caption{Typical choices for the changes of variables \eqref{quadcov}, depending on whether the equations \eqref{STRQ}--\eqref{STRsaddlepoint} are elliptic, hyperbolic, rational, or algebraic.  The function $\wp(z)$ is the Weierstrass elliptic function.}
\label{tab:cov}
\end{table}

For the equations of the form \eqref{3legadd}, if after combining each term under a common denominator the numerator of the resulting expression has the form \eqref{afflinquadpoly}, then this numerator provides the desired quad equation.

For the equations of the form \eqref{3legmult}, since $a(x;y;\alpha)$ and $c(x;y;\alpha)$ are linear fractional in $y$, these equations have the form
\begin{equation}\label{tlg}
\frac{N(x_a;x_b,x_c,x_d;\alpha,\beta)}{D(x_a;x_b,x_c,x_d;\alpha,\beta)}=1,
\end{equation}
where $N(x_a;x_b,x_c,x_d;\alpha,\beta)$ and $D(x_a;x_b,x_c,x_d;\alpha,\beta)$ are multilinear polynomials of the three variables $x_b$, $x_c$, $x_d$.  Rewriting the equation \eqref{tlg} as
\begin{equation}\label{tlgg}
N(x_a;x_b,x_c,x_d;\alpha,\beta)-D(x_a;x_b,x_c,x_d;\alpha,\beta)=0,
\end{equation}
the desired multilinear quad equation is then obtained if a factorisation of \eqref{tlgg} can be found of the form
\begin{equation}\label{quadfactor}
f(x_a;\alpha,\beta)\quadQ{x_a}{x_b}{x_c}{x_d}{\alpha}{\beta}=0,
\end{equation}
where $f$ is some function that is independent of $x_b,x_c,x_d$, and $Q$ is a polynomial of the form \eqref{afflinquadpoly}.  From \eqref{quadfactor} it is seen that the equations \eqref{3legmult} 
are satisfied on solutions of the corresponding multilinear quad equation $\quadQ{x_a}{x_b}{x_c}{x_d}{\alpha}{\beta}=0$, but there may be some solutions for \eqref{3legmult} corresponding to $f(x_a;\alpha,\beta)=0$ that are not solutions of $\quadQ{x_a}{x_b}{x_c}{x_d}{\alpha}{\beta}=0$.

The above procedure for obtaining quad equations from quasi-classical limits has been investigated \cite{Bazhanov:2016ajm,Kels:2018xge} for different solutions of the star-triangle relations of the forms \eqref{STRQ} and \eqref{STRH} that arise as hyperbolic/rational/algebraic degenerations of Bazhanov and Sergeev's solution of the star-triangle relation \cite{Bazhanov:2010kz}.  In terms of hypergeometric integrals, the latter star-triangle relation is equivalent to Spiridonov's elliptic beta integral formula \cite{Spiridonov:2010em}.  The quad equations that were obtained from the quasi-classical limit of these integrals correspond to equations in the ABS classification \cite{ABS,ABS2}, where the latter classification was made using completely different methods that are independent of the star-triangle relations.  Thus, the quasi-classical limit offers a different approach to the study of discrete integrable equations.

The equations in the ABS list may be grouped into two main types known as type-Q and type-H. These equations are listed in Table \ref{tab:ABSlist} and are given explicitly in Appendix \ref{app:equations}.  The type-Q equations may be derived from the quasi-classical expansion of star-triangle relations of the form \eqref{STRQ} and the type-H equations may be derived from the quasi-classical expansion of star-triangle relations of the form \eqref{STRH}.

\begin{table}[htb!]
\centering
\begin{tabular}{c|c|c}

& Type-Q  & Type-H  
 
 \\
 
 \hline
 
 




Elliptic & $Q4$ & - \\

Hyperbolic & $Q3_{(1)}$ $Q3_{(0)}$ & $H3_{(1;\,1)}$, $H3_{(1,0)}$, $H3_{(0,0)}$ \\[0.0cm]

Rational & $Q2$, $Q1_{(1)}$ & $H2_{(1)}$, $H2_{(0)}$ \\[0.0cm]

Algebraic & $Q1_{(0)}$ & $H1_{(1)}$, $H1_{(0)}$

\\[0.0cm]

\hline 
\end{tabular}
\caption{List of type-Q and type-H ABS quad equations grouped according to whether they come from a quasi-classical expansion that involves elliptic, hyperbolic, rational, or algebraic equations \cite{Bazhanov:2016ajm,Kels:2018xge}.  The equations are given in Appendix \ref{app:equations}.}%
\label{tab:ABSlist}
\end{table}

Each of the type-Q quad equations in Table \ref{tab:ABSlist} satisfy the following symmetries of the square 
\begin{equation}\label{symmetriesquad}
\begin{split}
&\quadQ{x_a}{x_b}{x_c}{x_d}{\al}{\bt}
=\quadQ{x_b}{x_a}{x_d}{x_c}{\al}{\bt}, \\
&\quadQ{x_a}{x_b}{x_c}{x_d}{\al}{\bt}
=-\quadQ{x_d}{x_b}{x_c}{x_a}{\bt}{\al}. 
\end{split}
\end{equation}
If $\varepsilon=0$, the type-H equations $H3_{(\delta;\,\varepsilon)}$, $H2_{(\varepsilon)}$, $H1_{(\varepsilon)}$, also satisfy both of the symmetries of \eqref{symmetriesquad}, while if $\varepsilon=1$ they only satisfy the second symmetry of \eqref{symmetriesquad}.

\subsubsection{Three-leg equations associated to quad equations}

The three-leg equations \eqref{3legmult} and \eqref{3legadd} are each written in terms of three functions, where in the arguments of the functions the variable $x_a$ is paired one of the other three variables $x_b,x_c,x_d$, assigned to vertices of the square.  These are also known as three-leg equations centered at $x_a$.  As was observed in the original classification result of ABS \cite{ABS}, each of the equations in Table \ref{tab:ABSlist} also have associated three-leg equations that are centered at the other three vertices.  The system of four different type-Q three-leg equations may be expressed in terms of a single function $a(x_i;x_j;\al)$, but the system of four different type-H three-leg equations is more complicated and expressed in terms of up to four different functions \cite{BollSuris11}.  
\begin{thm}
\label{thm:3leg}
If $\quadQ{x_a}{x_b}{x_c}{x_d}{\alpha}{\beta}=0$ is one of the elliptic, hyperbolic, or rational type-Q quad equations from Table \ref{tab:ABSlist}, then there is a system of three-leg equations of the form
\begin{equation}
\begin{alignedat}{2}\label{43legsQ}
(\Qa)&&\qquad\ds \frac{a(x_a;x_c;\beta)}{a(x_a;x_b;\alpha)a(x_a;x_d;{\beta-\alpha})}=1, \\
(\Qb)&&\qquad\ds \frac{a(x_b;x_a;\alpha)}{a(x_b;x_d;\beta)a(x_b;x_c;{\alpha-\beta})}=1, \\
(\Qc)&&\qquad\ds \frac{a(x_c;x_d;\alpha)}{a(x_c;x_a;\beta)a(x_c;x_b;{\alpha-\beta})}=1, \\
(\Qd)&&\qquad\ds \frac{a(x_d;x_b;\beta)}{a(x_d;x_c;\alpha)a(x_d;x_a;{\beta-\alpha})}=1,
\end{alignedat}
\end{equation}
which are satisfied on solutions of $\quadQ{x_a}{x_b}{x_c}{x_d}{\alpha}{\beta}=0$.  If $\quadQ{x_a}{x_b}{x_c}{x_d}{\alpha}{\beta}=0$ is the algebraic type-Q quad equation $Q1_{(0)}$, then the same holds true with each of $(\Qa)$--$(\Qd)$ being replaced with additive equations of the form \eqref{3legadd}.

Similarly, if $\quadH{x_a}{x_b}{x_c}{x_d}{\alpha}{\beta}=0$ is one of the hyperbolic or rational type-H quad equations from Table \ref{tab:ABSlist}, then there is a system of three-leg equations of the form
\begin{equation}
\begin{alignedat}{2}\label{43legsH}
(\Ha)&&\qquad\ds \frac{c^\ast(x_a;x_c;\beta)}{c^\ast(x_a;x_b;\alpha)a^\ast(x_a;x_d;{\beta-\alpha})}=1, \\
(\Hb)&&\qquad\ds \frac{c(x_b;x_a;\alpha)}{c(x_b;x_d;\beta)a(x_b;x_c;{\alpha-\beta})}=1, \\
(\Hc)&&\qquad\ds \frac{c(x_c;x_d;\alpha)}{c(x_c;x_a;\beta)a(x_c;x_b;{\alpha-\beta})}=1, \\
(\Hd)&&\qquad\ds \frac{c^\ast(x_d;x_b;\beta)}{c^\ast(x_d;x_c;\alpha)a^\ast(x_d;x_a;{\beta-\alpha})}=1,
\end{alignedat}
\end{equation}
which are satisfied on solutions of $\quadH{x_a}{x_b}{x_c}{x_d}{\alpha}{\beta}=0$, and where the functions $a(x_i;x_j;\al)$ and $a^\ast(x_i;x_j;\al)$ also satisfy systems of the form \eqref{43legsQ} on solutions of respective type-Q equations.  If $\quadH{x_a}{x_b}{x_c}{x_d}{\alpha}{\beta}=0$ is the algebraic type-H equation $H1_{(\varepsilon)}$, then for $\varepsilon=1$ the same holds true with $(\Ha)$ and $(\Hd)$ being replaced with additive equations of the form \eqref{3legadd}, and for $\varepsilon=0$ with each of $(\Ha)$--$(\Hd)$ being replaced with additive equations.

\end{thm}


\begin{proof}
Theorem \ref{thm:3leg} can be verified directly by using the functions $a(x_i,x_j;\al)$, $a^\ast(x_i;x_j;\al)$, $c(x_i,x_j;\al)$, and $c^\ast(x_i;x_j;\al)$, that are given in Table \ref{tab:3legquadsH} of Appendix \ref{app:equations}.  With these functions an equation $(\Qi)$ from \eqref{43legsQ} is equivalent to (as was given in \eqref{quadfactor} for $(\Qa)$)
\begin{equation}
f(x_i;\alpha,\beta)\quadQ{x_a}{x_b}{x_c}{x_d}{\alpha}{\beta}=0,\qquad i\in\{a,b,c,d\},
\end{equation}
where $Q$ is some type-Q equation from Table \ref{tab:ABSlist}, and an equation $(\Hi)$ from \eqref{43legsH} is equivalent to (as was given in \eqref{quadfactor} for $(\Ha)$)
\begin{equation}
\begin{split}
f_1(x_i;\alpha,\beta)\quadH{x_a}{x_b}{x_c}{x_d}{\alpha}{\beta}=0,\qquad i\in\{a,d\}, \\
f_2(x_i;\alpha,\beta)\quadH{x_a}{x_b}{x_c}{x_d}{\alpha}{\beta}=0,\qquad i\in\{b,c\},
\end{split}
\end{equation}
where $H$ is some type-H equation from Table \ref{tab:ABSlist}.
\end{proof}

The existence of consistent systems of three-leg equations \eqref{43legsQ} and \eqref{43legsH} may be naturally understood through the quasi-classical limit of the star-triangle relations.  The idea is as follows. Equating both sides at leading order of a quasi-classical expansion of \eqref{STRH}, implies a classical star-triangle relation formula that takes the form (up to some irrelevant factors independent of the variables)
\begin{equation}\label{CSTR}
\begin{split}
\olam_{\alpha_1}(x_1,x_0)+\lam_{\alpha_1+\alpha_3}(x_2,x_0)+\ol_{\alpha_3}(x_0,x_3)=
\lam_{\alpha_1}(x_2,x_3)+\olam_{\alpha_1+\alpha_3}(x_1,x_3)+\lag_{\alpha_3}(x_2,x_1).
\end{split}
\end{equation}
The left hand side is the same that appears inside the exponential for the second integral of \eqref{STRqcl}, and the right hand side comes from a quasi-classical expansion of the right hand side of \eqref{STRH}.  The equation \eqref{CSTR} is required to hold on solutions of the second saddle point equation in \eqref{STRsaddlepoint}.  Recall that this saddle point equation leads to the three-leg equation $(\Ha)$ given in \eqref{3legmult}.  Assuming \eqref{CSTR} holds, there are also three other partial derivatives that can be taken with respect to the three variables $x_1,x_2,x_3$, and following a similar procedure that was used to obtain \eqref{3legmult}, these derivatives would lead to the other three-leg equations $(\Hb)$--$(\Hd)$.

\subsection{Face-centered quad equations and star-star relations}

A face-centered quad equation may be regarded as an extension of the regular quad equation \eqref{afflinquad}, having dependence on an additional variable and parameters.  The parameters may be written with two components as
\begin{equation}
\bal=(\al_1,\al_2),\quad \bbt=(\bt_1,\bt_2).    
\end{equation}
A face-centered quad equation may then be written as
\begin{equation}\label{afflin}
\At{x}{x_a}{x_b}{x_c}{x_d}{\bal}{\bbt}=0,
\end{equation}
where $A$ is a multivariate polynomial of five variables $x,x_a,x_b,x_c,x_d$.  It is multilinear in the four variables $x_a,x_b,x_c,x_d$, but there is no restriction imposed on the degree of $x$.  Thus, it is a polynomial with the general form
\begin{equation}\label{afflinpoly}
\begin{split}
\At{\ccx}{x_a}{x_b}{x_c}{x_d}{\bal}{\bbt} 
=\kappa_1x_ax_bx_cx_d+\kappa_2x_ax_bx_c+\kappa_3x_ax_bx_d+\kappa_4x_ax_cx_d+\kappa_5x_bx_cx_d\phantom{,}&
\\
+\kappa_6x_ax_b+\kappa_7x_ax_c+\kappa_8x_ax_d+\kappa_9x_bx_c+\kappa_{10}x_bx_d+\kappa_{11}x_cx_d\phantom{,}& \\
+\kappa_{12}x_a+\kappa_{13}x_b+\kappa_{14}x_c+\kappa_{15}x_d+\kappa_{16},&
\end{split}
\end{equation}
where the coefficients $\kappa_i=\kappa_i(x,\bal,\bbt)$ ($i=1,\ldots,16$) depend on the face variable $\ccx$ and the components of the two parameters $\bal,\bbt$.  The multilinear expression \eqref{afflinpoly} resembles the multilinear expression for regular quad equations \eqref{afflinquadpoly}, but with more general coefficients $\kappa_i$.

The face-centered quad equation \eqref{afflin} may be associated to the square of Figure \ref{fig:fcquadA}, where the variables $x_a,x_b,x_c,x_d$ are assigned to vertices at corners, the variable $x$ is assigned to a central vertex on the face, and the components of the parameters are assigned to edges.  When evolving the equations in the square lattice it is never required to solve for the variable $x$, which is the reason that the dependence on $x$ does not have to be linear.

\begin{figure}[htb!]
\centering
\begin{tikzpicture}[scale=1.5]

\coordinate (A0) at (0,0);
\coordinate (A1) at (-1,1);
\coordinate (A2) at (1,1);
\coordinate (A3) at (-1,-1);
\coordinate (A4) at (1,-1);

\draw[gray,very thin,dashed] (A1)--(A2)--(A4)--(A3)--cycle;

\draw (A0)--(A1);
\draw (A0)--(A2);
\draw (A0)--(A3);
\draw (A0)--(A4);

\fill (A3) circle (1.7pt)
node[left=1.5pt]{\color{black} $x_c$};
\fill (A1) circle (1.7pt)
node[left=1.5pt]{\color{black} $x_a$};
\fill (A2) circle (1.7pt)
node[right=1.5pt]{\color{black} $x_b$};
\fill (A4) circle (1.7pt)
node[right=1.5pt]{\color{black} $x_d$};
\fill (A0) circle (1.7pt)
node[left=1.5pt]{\color{black} $x$};

\draw[-,dotted] (-1.2,0.5)--(1.2,0.5);
\draw[-,dotted] (-1.2,-0.5)--(1.2,-0.5);
\draw[-,dotted] (-0.5,-1.2)--(-0.5,1.2);
\draw[-,dotted] (0.5,-1.2)--(0.5,1.2);

\draw (-1.2,0.5) circle (0.01pt)
node[left=0pt]{\small $\al_2$};
\draw (1.2,0.5) circle (0.01pt)
node[right=0pt]{\small $\al_2$};
\draw (-1.2,-0.5) circle (0.01pt)
node[left=0pt]{\small $\al_1$};
\draw (1.2,-0.5) circle (0.01pt)
node[right=0pt]{\small $\al_1$};
\draw (-0.5,-1.2) circle (0.01pt)
node[below=0pt]{\small $\bt_1$};
\draw (-0.5,1.2) circle (0.01pt)
node[above=0pt]{\small $\bt_1$};
\draw (0.5,-1.2) circle (0.01pt)
node[below=0pt]{\small $\bt_2$};
\draw (0.5,1.2) circle (0.01pt)
node[above=0pt]{\small $\bt_2$};

\fill (0,-1.5) circle (0.01pt)
node[below=0.5pt]{\color{black}\small $\At{\ccx}{x_a}{x_b}{x_c}{x_d}{\bal}{\bbt}=0$};

\end{tikzpicture}
\caption{The face-centered quad equation \eqref{afflin}.}
\label{fig:fcquadA}
\end{figure}
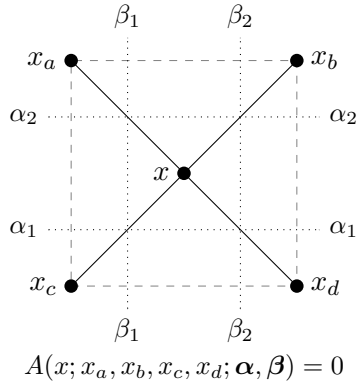

The face-centered quad equation \eqref{afflin} may also be written in the form
\begin{equation}\label{afflinsum}
\sum_{i=0}^nx^iP_i(x_a,x_b,x_c,x_d;\bal,\bbt)=0,
\end{equation}
where $n$ is the degree of the polynomial \eqref{afflinpoly} in $x$, and the $P_i$ ($i=0,1,\ldots,n$) are multilinear polynomials in the four variables $x_a,x_b,x_c,x_d$, {\it i.e.}, the $P_i$ are regular quad equations (but depending on the two-component parameters $\bal,\bbt$, instead of two scalar parameters).

For face-centered quad equations the analogue of the star-triangle relations \eqref{STRQ} and \eqref{STRH} are the star-star relations \cite{Bazhanov:1992jqa,Baxter:1997tn,Bazhanov:2011mz}.  Two examples of such star-star relations for the case of continuous spin variables have the forms
\begin{align}\label{SSRa}
W_{v_1-v_2}(\spn_1,\spn_2)W_{u_1-u_2}(\spn_1,\spn_3)\int_C w^{(1)}_{\bu\bv}(\bspn) &=
W_{v_1-v_2}(\spn_3,\spn_4)W_{u_1-u_2}(\spn_2,\spn_4)\int_C w^{(2)}_{\bu\bv}(\bspn), \\
\label{SSRc}
V_{v_1-v_2}(\spn_1,\spn_2)W_{u_1-u_2}(\spn_1,\spn_3)\int_C v^{(1)}_{\bu\bv}(\bspn) &=
V_{v_1-v_2}(\spn_3,\spn_4)W_{u_1-u_2}(\spn_2,\spn_4)\int_C v^{(2)}_{\bu\bv}(\bspn),
\end{align}
where $\bspn$, $\bu$, $\bv$, represent variables $\bspn=(\spn_0,\spn_1,\spn_2,\spn_3,\spn_4)$ and parameters $\bu=(u_1,u_2)$, $\bv=(v_1,v_2)$, and the integrands are
\begin{equation}
\begin{split}
w^{(1)}_{\bu\bv}(\bspn)&=d\spn_0
W_{u_2-v_1}(\spn_1,\spn_0)\oW_{u_2-v_2}(\spn_0,\spn_2)\oW_{u_1-v_1}(\spn_3,\spn_0)W_{u_1-v_2}(\spn_0,\spn_4), \\
w^{(2)}_{\bu\bv}(\bspn)&=d\spn_0
W_{u_2-v_1}(\spn_0,\spn_4)\oW_{u_2-v_2}(\spn_3,\spn_0)\oW_{u_1-v_1}(\spn_0,\spn_2)W_{u_1-v_2}(\spn_1,\spn_0), \\
v^{(1)}_{\bu\bv}(\bspn)&=d\spn_0
W_{u_2-v_1}(\spn_1,\spn_0)\oV_{u_2-v_2}(\spn_0,\spn_2)\oW_{u_1-v_1}(\spn_3,\spn_0)V_{u_1-v_2}(\spn_0,\spn_4), \\
v^{(2)}_{\bu\bv}(\bspn)&=d\spn_0
W_{u_2-v_1}(\spn_0,\spn_4)\oV_{u_2-v_2}(\spn_3,\spn_0)\oW_{u_1-v_1}(\spn_0,\spn_2)V_{u_1-v_2}(\spn_1,\spn_0).
\end{split}
\end{equation}
Similarly to the star-triangle relations \eqref{STRQ} and \eqref{STRH}, each of $V_\theta(\spn_i,\spn_j)$, $\oV_\theta(\spn_i,\spn_j)$, $W_\theta(\spn_i,\spn_j)$, $\oW_\theta(\spn_i,\spn_j)$ are complex-valued functions of $\spn_i,\spn_j,\theta$, while the formulas \eqref{SSRa} and \eqref{SSRc} depend on the independent variables $\spn_1,\spn_2,\spn_3,\spn_4$, and parameters $u_1,u_2,v_1,v_2$.  The contours $C$ are typically chosen to separate infinite sequences of poles that go to zero and/or infinity. In terms of hypergeometric integrals, examples of these formulas have been shown \cite{Bazhanov:2013bh} to be equivalent to the simplest ($n=1$) examples of transformation formulas between hypergeometric integrals associated to the $A_n$ root system \cite{RainsT}.

Consider the left hand sides of the star-star relations \eqref{SSRa} and \eqref{SSRc}. In a quasi-classical limit these left hand sides are written in the respective forms
\begin{equation}\label{qclssr}
\begin{gathered}
\int_Cdx_0\,\EXP^{\hbar^{-1}\bigl(\lag_{u_2'-v_1'}(x_1,x_0)+\ol_{u_2'-v_2'}(x_0,x_2)+\ol_{u_1'-v_1'}(x_3,x_0)+\lag_{u_1'-v_2'}(x_0,x_4)+\lag_{v_1'-v_2'}(x_1,x_2)+\lag_{u_1'-u_2'}(x_1,x_3)\bigr)+O(1)}\!, \\
\int_Cdx_0\,\EXP^{\hbar^{-1}\bigl(\lag_{u_2'-v_1'}(x_1,x_0)+\olam_{u_2'-v_2'}(x_0,x_2)+\ol_{u_1'-v_1'}(x_3,x_0)+\lam_{u_1'-v_2'}(x_0,x_4)+\lam_{v_1'-v_2'}(x_1,x_2)+\lag_{u_1'-u_2'}(x_1,x_3)\bigr)+O(1)}\!,
\end{gathered}
\end{equation}
for some parameter $\hbar\to0$, where each of $\lam_\alpha(x_i,x_j)$, $\olam_\alpha(x_i,x_j)$, $\lag_\alpha(x_i,x_j)$, $\ol_\alpha(x_i,x_j)$, are complex-valued functions of $x_i,x_j,\alpha$, which arise as the leading $O(\hbar^{-1})$ asymptotics of the respective functions $V_\theta(\spn_i,\spn_j)$, $\oV_\theta(\spn_i,\spn_j)$, $W_\theta(\spn_i,\spn_j)$, $\oW_\theta(\spn_i,\spn_j)$.  Analogously to \eqref{YBEcov}, the variables and parameters of \eqref{SSRa}, \eqref{SSRc}, and \eqref{qclssr}, are related by
\begin{equation}\label{YBEcov2}
\spn_i=F_i(x_i),\quad i=0,1,2,3,4\qquad u_j'=G(u_j),\quad v_j'=G(v_j),\quad j=1,2,
\end{equation}
where $F_i(z)$ and $G(z)$ are some M\"obius transformations that depend on $\hbar$.  The saddle-point equations of these integrals are given by
\begin{equation}\label{SSRsaddle}
\begin{split}
\frac{\partial}{\partial x_0}\bigl(\lag_{u_2'-v_1'}(x_1,x_0)+\ol_{u_2'-v_2'}(x_0,x_2)+\ol_{u_1'-v_1'}(x_3,x_0)+\lag_{u_1'-v_2'}(x_0,x_4)\bigr)=0, \\
\frac{\partial}{\partial x_0}\bigl(\lag_{u_2'-v_1'}(x_1,x_0)+\olam_{u_2'-v_2'}(x_0,x_2)+\ol_{u_1'-v_1'}(x_3,x_0)+\lam_{u_1'-v_2'}(x_0,x_4)\bigr)=0. 
\end{split}
\end{equation}
Analogously to \eqref{quadcov}, to obtain a discrete integrable equation from \eqref{SSRsaddle} a point transformation
\begin{equation}\label{quadcov2}
\begin{gathered}
x=f_0(x_0),\quad x_a=f_1(x_1),\quad x_b=f_2(x_2),\quad x_c=f_3(x_3),\quad x_d=f_4(x_4), \\
\al_j=g(u_j'),\quad \bt_j=g(v_j'),\quad j=1,2,
\end{gathered}
\end{equation}
of the variables and parameters should be used that will turn the saddle-point equations \eqref{STRsaddlepoint} into equations of the form 
\begin{align}
\label{4lega}
\frac{
a(x;x_a;\bt_1-\al_2)a(x;x_d;\bt_2-\al_1)}{a(x;x_b;\bt_2-\al_2)a(x;x_c;\bt_1-\al_1)}=1, \\
\label{4legc}
\frac{a(x;x_a;\bt_1-\al_2)c(x;x_d;\bt_2-\al_1)}{a(x;x_b;\bt_2-\al_2)c(x;x_c;\bt_1-\al_1)}=1,
\end{align}
or the form
\begin{align}
\label{4legadd}
a(x;x_a;\bt_1-\al_2)+a(x;x_d;\bt_2-\al_1)-a(x;x_b;\bt_2-\al_2)-a(x;x_c;\bt_1-\al_1)=0, \\
\label{4legcdd}
a(x;x_a;\bt_1-\al_2)+c(x;x_d;\bt_2-\al_1)-a(x;x_b;\bt_2-\al_2)-c(x;x_c;\bt_1-\al_1)=0,
\end{align}
where the functions $a(x;y;\al)$ and $c(x;y;\al)$ are linear fractional functions of the variable $y$.  Typical choices of the functions $f_i$ and $g$ for \eqref{quadcov2} are the same as those listed in Table \ref{tab:cov}, depending on whether the equations \eqref{SSRa}, \eqref{SSRc}, \eqref{qclssr}, \eqref{SSRsaddle}, are elliptic, hyperbolic, rational, or algebraic.  The algebraic cases give equations of the form \eqref{4legadd} and \eqref{4legcdd}, and all other cases give equations of the form \eqref{4lega} and \eqref{4legc}.

For the equations of the form \eqref{4legadd} and \eqref{4legcdd}, if after combining each term under a common denominator the numerator of the resulting expression has the form \eqref{afflinpoly}, then this numerator provides the desired face-centered quad equation.

For the equations of the form \eqref{4lega} and \eqref{4legc}, since $a(x;y;\alpha)$ and $c(x;y;\alpha)$ are linear fractional in $y$, these equations have the form
\begin{equation}\label{tlg2}
\frac{N(x;x_a;x_b,x_c,x_d;\bal,\bbt)}{D(x;x_a;x_b,x_c,x_d;\bal,\bbt)}=1,
\end{equation}
where $N(x;x_a;x_b,x_c,x_d;\bal,\bbt)$ and $D(x;x_a;x_b,x_c,x_d;\bal,\bbt)$ are multilinear polynomials of the four variables $x_a$, $x_b$, $x_c$, $x_d$.  Rewriting the equation \eqref{tlg2} as
\begin{equation}\label{tlgg2}
N(x;x_a;x_b,x_c,x_d;\bal,\bbt)-D(x;x_a;x_b,x_c,x_d;\bal,\bbt)=0,
\end{equation}
the desired multilinear face-centered quad equation is then obtained if a factorisation of \eqref{tlgg2} can be found of the form
\begin{equation}\label{quadfactor2}
f(x;\bal,\bbt)\At{x}{x_a}{x_b}{x_c}{x_d}{\bal}{\bbt}=0,
\end{equation}
where $f$ is some function independent of $x_a,x_b,x_c,x_d$, and $A$ is a polynomial of the form \eqref{afflinpoly}.

The star-star relations \eqref{SSRa} and \eqref{SSRc} are known to imply a Yang-Baxter equation for interaction-round-a-face (IRF) models \cite{Bazhanov:1992jqa,Bazhanov:2011mz}, and this connection was used \cite{Kels:2020zjn} with the Boltzmann weights for the star-triangle relations corresponding to ABS equations \cite{Kels:2018xge} to derive the type-A and type-C face-centered quad equations that are listed in Table \ref{tab:CAFCClist}, and given explicitly in Appendix \ref{app:equations}.  The type-A equations and type-C equations may be derived from the quasi-classical expansion of star-star relations of the form \eqref{SSRa} and \eqref{SSRc}, respectively.  The third type of equations known as type-B would come from star-star relations involving only $V_\theta(\spn_i,\spn_j)$ and $\oV_\theta(\spn_i,\spn_j)$, but no such relations are known to the author, and for this reason type-B equations do not appear in this paper. 
Unlike the quad equations listed in Table \ref{tab:CAFCClist}, there currently does not exist a classification of face-centered quad equations of the form \eqref{afflin}.

\begin{table}[htb!]
\centering
\begin{tabular}{c|c|c}

 & Type-A equations & Type-C equations 
 
 \\
 
 \hline
 
 
Elliptic & $A4$ & - \\

Hyperbolic & $A3_{(1)}$, $A3_{(0)}$ & $C3_{(1/2;\,1/2;\,0)}$, $C3_{(1/2;\,0;\,1/2)}$, $C3_{(1;\,0;\,0)}$, $C3_{(0;\,0;\,0)}$ \\[0.0cm]

Rational & $A2_{(1;\,1)}$, $A2_{(1;\,0)}$ & $C2_{(1;\,1;\,0)}$, $C2_{(1;\,0;\,1)}$, $C2_{(1;\,0;\,0)}$, $C1_{(1)}$ \\[0.0cm]

Algebraic & $A2_{(0;\,0)}$ & $C2_{(0;\,0;\,0)}$, $C1_{(0)}$

\\[0.0cm]

\hline 
\end{tabular}
\caption{Type-A and type-C face-centered quad equations grouped according to whether they come from a quasi-classical expansion involving elliptic, hyperbolic, rational, or algebraic equations \cite{Kels:2020zjn}. The equations are given in Appendix \ref{app:equations}.}
\label{tab:CAFCClist}
\end{table}

The analogues of the square symmetries \eqref{symmetriesquad} for face-centered quad equations are
\begin{equation}\label{symmetriesfcqe}
\begin{split}
 \At{x}{x_a}{x_b}{x_c}{x_d}{\al}{\bt}=
-\At{x}{x_b}{x_a}{x_d}{x_c}{\al}{\hat{\bt}}, \\
 \At{x}{x_a}{x_b}{x_c}{x_d}{\al}{\bt}=
-\At{x}{x_c}{x_d}{x_a}{x_b}{\hat{\al}}{\bt}, \\
 \At{x}{x_a}{x_b}{x_c}{x_d}{\al}{\bt}=
-\At{x}{x_d}{x_b}{x_c}{x_a}{\bt}{\al}.
\end{split}
\end{equation}
Type-A equations listed in Table \ref{tab:CAFCClist} satisfy each of these symmetries, while type-C equations only satisfy the first of these symmetries.

It was found \cite{GubbiottiKels} that in the square lattice not all type-C equations from Table \ref{tab:CAFCClist} have vanishing algebraic entropy on their own, but can have vanishing algebraic entropy in specific arrangements of the following pairs (some equations are simply paired with themselves)
\begin{equation}\label{typeCpairs}
\begin{gathered}
\bigl(C3_{(1/2;\,1/2;\,0)},C3_{(1/2;\,0;\,1/2)}\bigr),\;
\bigl(C3_{(1;\,0;\,0)},C3_{(1;\,0;\,0)}\bigr),\;
\bigl(C3_{(0;\,0;\,0)},C3_{(0;\,0;\,0)}\bigr),\\
\bigl(C2_{(1;\,1;\,0)},C2_{(1;\,0;\,1)}\bigr),\;
\bigl(C2_{(1;\,0;\,0)},C2_{(1;\,0;\,0)}\bigr),\;
\bigl(C2_{(0;\,0;\,0)},C1_{(1)}\bigr),\;
\bigl(C1_{(0)},C1_{(0)}\bigr).
\end{gathered}
\end{equation}
These pairs will be needed for constructing consistent systems of equations involving type-C equations.

\subsubsection{Consistency-around-a-hexagon}

Let $X$ denote the six-tuple of variables 
\begin{equation}
X=(\xa,\xb,\xf,\xd,\xc,\xe),
\end{equation}
and let $X_i$ ($i\in\mathbb{Z}/6\mathbb{Z}$) denote the i\textsuperscript{th} (mod $6$) element of $X$.  These may be considered as variables assigned to six consecutive vertices of a hexagon, as shown in Figure \ref{fig:hexagon}.

\begin{figure}[htb!]
\centering
\begin{tikzpicture}[scale=1.5]

\pgfmathsetmacro\tr{sqrt(3)/2}

\coordinate (A1) at (-1,0);
\coordinate (A2) at (-1/2,{\tr});
\coordinate (A3) at (1/2,{\tr});
\coordinate (A4) at (1,0);
\coordinate (A5) at (1/2,{-\tr});
\coordinate (A6) at (-1/2,{-\tr});

\draw (A1)--(A2);
\draw (A2)--(A3);
\draw (A3)--(A4);
\draw (A4)--(A5);
\draw (A5)--(A6);
\draw (A6)--(A1);

\fill (A1) circle (1.4pt) node[left]{$\xa$};
\fill (A2) circle (1.4pt) node[above left]{$\xb$};
\fill (A3) circle (1.4pt) node[above right]{$\xf$};
\fill (A4) circle (1.4pt) node[right]{$\xd$};
\fill (A5) circle (1.4pt) node[below right]{$\xc$};
\fill (A6) circle (1.4pt) node[below left]{$\xe$};

 \end{tikzpicture}
 
 \caption{Variables assigned to vertices of a hexagon.}
 \label{fig:hexagon}

\end{figure}
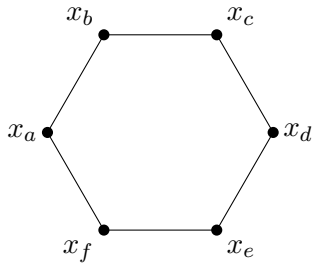

The elements of $X$ may be split up into two disjoint sets $S_{i}$ and $\overline{S}_i$ that respectively contain consecutive elements of $X$
\begin{equation}
S_{i}=\{X_{i},X_{i+1},X_{i+2},X_{i+3}\},\quad \overline{S}_i=\{X_{i-2},X_{i-1}\},\qquad i\in\mathbb{Z}/6\mathbb{Z}.
\end{equation}
An overdetermined system of face-centered quad equations can be used to consistently determine any two variables on consecutive vertices of the hexagon (from $\overline{S}_i$) in terms of the other four (from $S_i$) as follows.
\begin{thm}[Consistency-around-a-hexagon]\label{thm:CAH}
Let $A$ be a polynomial for one of the type-A face-centered quad equations listed in Table \ref{tab:CAFCClist}. For any $i\in\mathbb{Z}/6\mathbb{Z}$, the two equations from
\begin{equation}\label{64legsA}
\begin{alignedat}{2}
(\Aa)&&\qquad\At{\xa}{\xe}{\xc}{\xb}{\xf}{(\bt_2,\al_2)}{(\bt_1,\al_1)}=0, \\
(\Ab)&&\qquad\At{\xb}{\xf}{\xd}{\xa}{\xe}{(\bt_2,\al_1)}{(\bt_1,\al_2)}=0, \\
(\Af)&&\qquad\At{\xf}{\xb}{\xd}{\xa}{\xc}{(\bt_2,\bt_1)}{(\al_1,\al_2)}=0, \\
(\Ad)&&\qquad\At{\xd}{\xf}{\xb}{\xc}{\xe}{(\bt_2,\al_2)}{(\bt_1,\al_1)}=0, \\
(\Ac)&&\qquad\At{\xc}{\xe}{\xa}{\xd}{\xf}{(\bt_2,\al_1)}{(\bt_1,\al_2)}=0, \\
(\Ae)&&\qquad\At{\xe}{\xc}{\xa}{\xd}{\xb}{(\bt_2,\bt_1)}{(\al_1,\al_2)}=0,
\end{alignedat}
\end{equation}
which are respectively independent of $X_{i-1}$ and $X_{i-2}$ may be used to solve uniquely for the two variables $X_{i-2},X_{i-1}\in\overline{S}_i$ in terms of the four variables from $S_i$, and the remaining four equations from \eqref{64legsA} are also satisfied on these solutions.

In a similar way, let one of $(C,\overline{C})$ or $(\overline{C},C)$ denote one of the pairs of type-C face-centered quad equations given in \eqref{typeCpairs}.  For any $i\in\mathbb{Z}/6\mathbb{Z}$, the two equations from
\begin{equation}\label{64legsC}
\begin{alignedat}{2}
(\Ca)&&\qquad\Cft{\xa}{\xe}{\xc}{\xb}{\xf}{(\bt_2,\al_2)}{(\bt_1,\al_1)}=0, \\
(\Cb)&&\qquad \Ct{\xb}{\xf}{\xd}{\xa}{\xe}{(\bt_2,\al_1)}{(\bt_1,\al_2)}=0, \\
(\Cf)&&\qquad \Ct{\xf}{\xb}{\xd}{\xa}{\xc}{(\bt_2,\bt_1)}{(\al_1,\al_2)}=0, \\
(\Cd)&&\qquad \Ct{\xd}{\xf}{\xb}{\xc}{\xe}{(\bt_2,\al_2)}{(\bt_1,\al_1)}=0, \\
(\Cc)&&\qquad\Cft{\xc}{\xe}{\xa}{\xd}{\xf}{(\bt_2,\al_1)}{(\bt_1,\al_2)}=0, \\
(\Ce)&&\qquad\Cft{\xe}{\xc}{\xa}{\xd}{\xb}{(\bt_2,\bt_1)}{(\al_1,\al_2)}=0,
\end{alignedat}
\end{equation}
which are respectively independent of $X_{i-1}$ and $X_{i-2}$ may be used to solve uniquely for the two variables $X_{i-2},X_{i-1}\in\overline{S}_i$ in terms of the four variables from $S_i$, and the remaining four equations from \eqref{64legsC} are also satisfied on these solutions. 
\end{thm}

Theorem \ref{thm:CAH} can be verified directly.  
%
The initial conditions for Theorem \ref{thm:CAH} are on four consecutive vertices of the hexagon, but this could be relaxed to initial conditions on three consecutive vertices plus any one of the other three vertices of the hexagon.  The result of Theorem \ref{thm:CAH} will be important for considering systems of equations in the hexagonal lattice in Section \ref{sec:hexagon} as it allows for well-defined evolutions of the equations in six different lattice directions. 

Theorem \ref{thm:CAH} is in a sense an analogue for face-centered quad equations of Theorem \ref{thm:3leg} for the three-leg equations associated to ABS quad equations.  However, Theorem \ref{thm:3leg} is only relevant for systems of three-leg equations, because the corresponding property for quad equations is trivial (that four copies of the same quad equation on a square are consistent with themselves).  In contrast, Theorem \ref{thm:CAH} is a non-trivial property for consistency of overdetermined systems of face-centered quad equations, where each equation involves different combinations of five out of six variables on the hexagon.

Just as the equations of Theorem \ref{thm:3leg} may be understood from the quasi-classical limit of the star-triangle relations, the construction of the consistent systems of equations \eqref{64legsA} and \eqref{64legsC} is based on equations obtained from the quasi-classical limit of the star-star relations \eqref{SSRa} and \eqref{SSRc}.  The idea is as follows.  Equating both sides at leading order of a quasi-classical expansion of \eqref{SSRc} implies a classical star-star relation formula that takes the form
\begin{equation}\label{CSSR}
\begin{split}
\lag_{u_2'-v_1'}(x_1,x_0)+\olam_{u_2'-v_2'}(x_0,x_2)+\ol_{u_1'-v_1'}(x_3,x_0)+\lam_{u_1'-v_2'}(x_0,x_4)\phantom{,} \\
+\lam_{v_1'-v_2'}(x_1,x_2)+\lag_{u_1'-u_2'}(x_1,x_3)\phantom{,} \\
=
\lag_{u_2'-v_1'}(x_0',x_4)+\olam_{u_2'-v_2'}(x_3,x_0')+\ol_{u_1'-v_1'}(x_0,x_2')+\lam_{u_1'-v_2'}(x_1,x_0')\phantom{,} \\
+\lam_{v_1'-v_2'}(x_3,x_4)+\lag_{u_1'-u_2'}(x_2,x_4).
\end{split}
\end{equation}
The left hand side is the same that appears inside the exponential for the second integral of \eqref{qclssr}, and the right hand side comes from a quasi-classical expansion of the right hand side of \eqref{SSRc}.  The equation \eqref{CSSR} is required to hold (up to irrelevant factors) on solutions of the saddle points of the integrals in \eqref{SSRc}, {\it i.e.}, when the variables of \eqref{CSSR} are constrained such that the partial derivative of the left hand side with respect to $x_0$, and the partial derivative of the right hand side with respect to $x_0'$, are both equal to zero.  Assuming \eqref{CSSR} holds, there are also four other partial derivatives that can be taken with respect to the four variables $x_1,x_2,x_3,x_4$.  Then following a similar procedure that was outlined to derive the face-centered quad equation \eqref{quadfactor2} from the saddle-point equations \eqref{SSRsaddle}, the six different partial derivatives of the classical star-star relation \eqref{CSSR} would lead to the combination of six equations given in \eqref{64legsC}.  The details of these computations depend on the functions used for \eqref{CSSR} and should be considered case-by-case.  For the purposes here, it is sufficient to simply verify directly if systems of face-centered quad equations \eqref{64legsA} and \eqref{64legsC} are consistent.

\subsection{Connections between quad equations and face-centered quad equations}

The ABS quad equations of Table \ref{tab:ABSlist} and the face-centered quad equations of Table \ref{tab:CAFCClist} are closely related.  Starting either from an ABS quad equation from Table \ref{tab:ABSlist} or a face-centered quad equation from Table \ref{tab:CAFCClist}, one may derive a respective equation from the other table.   Going from ABS quad equations to face-centered quad equations may be done through their associated three-leg and four-leg equations, while going from face-centered quad equations to ABS quad equations may be done through the multilinear polynomial forms of the equations themselves.

For the latter, the type-Q ABS equations may be identified as special cases of the quad equation $P_1$ in \eqref{afflinsum} for type-A face-centered quad equations (except for the elliptic case), and the type-H ABS equations may be identified as special cases of the quad equations $P_1$ for type-C face-centered quad equations.

Note that if $n=2$ in \eqref{afflinsum}, then $P_1$ may be defined from the limit
\begin{equation}
\begin{split}
xP_1(x_a,x_b,x_c,x_d;\bal,\bbt)&=\At{x}{x_a}{x_b}{x_c}{x_d}{\bal}{\bbt} \\
&-\lim_{x\to0}\bigl(x^2\At{x^{-1}}{x_a}{x_b}{x_c}{x_d}{\bal}{\bbt} 
+\At{x}{x_a}{x_b}{x_c}{x_d}{\bal}{\bbt}\bigr),
\end{split}
\end{equation}
since this isolates the term linear in $x$.  If $\At{x}{x_a}{x_b}{x_c}{x_d}{\bal}{\bbt}=0$ is a type-A (resp.~type-C) face-centered quad equation from Table \ref{tab:CAFCClist}, then 
\begin{equation}\label{finalCAFCCtoABS}
P_1\bigl(x_a,-x_d,x_c,x_b;(0,\bt),(\al,\bt)\bigr)=0,
\end{equation}
is equivalent to some type-Q (resp.~type-H) ABS quad equation from Table \ref{tab:ABSlist}.  
This connection was previously observed in the original derivation of face-centered quad equations in the context of consistency-around-a-face-centered-cube \cite{Kels:2020zjn}. A list of face-centered quad equations and the corresponding ABS quad equation that can be found as the equation \eqref{finalCAFCCtoABS} is given in Table \ref{tab:CAFCCtoABS}.

\begin{table}[htb!]
\centering
\begin{tabular}{l|l}

Type-A & Eq. \eqref{finalCAFCCtoABS}
 
 \\
 
 \hline
 
 
 $A3_{(1)}$ & $Q3_{(1)}$ \\
 $A3_{(0)}$ & $Q3_{(0)}$ \\
 $A2_{(1;\,1)}$ & $Q2$ \\
 $A2_{(1;\,0)}$ & $Q1_{(1)}$ \\
 $A2_{(0;\,0)}$ & $Q1_{(0)}$

\\[0.0cm]

\hline 
\end{tabular}
 \hspace{2cm}
\begin{tabular}{l|l}

Type-C & Eq. \eqref{finalCAFCCtoABS}
 
 \\
 
 \hline
 
 
 $C3_{(1/2;\,1/2;\,0)}$ & $H3_{(1;\,1)}$ $(x_a\leftrightarrow x_c,x_b\leftrightarrow x_d)$ \\
 $C3_{(1/2;\,0;\,1/2)}$ & $H3_{(1;\,1)}$ \\
 $C3_{(1;\,0;\,0)}$ & $H3_{(1;\,0)}$ \\
 $C3_{(0;\,0;\,0)}$ & $H3_{(0;\,0)}$ \\
 $C2_{(1;\,1;\,0)}$ & $H2_{(1)}$ $(x_a\leftrightarrow x_c,x_b\leftrightarrow x_d)$ \\
 $C2_{(1;\,0;\,1)}$ & $H2_{(1)}$ \\
 $C2_{(1;\,0;\,0)}$ & $H2_{(0)}$ \\
 $C1_{(1)}$ & $H1_{(1)}$ $(x_a\leftrightarrow x_c,x_b\leftrightarrow x_d)$ \\
 $C2_{(0;\,0;\,0)}$ & $H1_{(1)}$ \\
 $C1_{(0)}$ &  $H1_{(0)}$

\\[0.0cm]

\hline 
\end{tabular}
\caption{ABS equations from Table \ref{tab:ABSlist} that arise as the equation \eqref{finalCAFCCtoABS} for a face-centered quad equation from Table \ref{tab:CAFCClist}.  The expressions of the equations are given in Appendix \ref{app:equations}.}%
\label{tab:CAFCCtoABS}
\end{table}

The above procedure does not work with the expression given in Appendix \ref{app:equations} for the elliptic face-centered quad equation $A4$, which unfortunately is more complicated and has $n=10$ in the form \eqref{afflinsum}.  Instead, for $A4$ the following limit
\begin{equation}\label{A4Q4}
\lim_{\gm\to\bt}\frac{P_9\bigl(x_a,-x_d,x_c,x_b;(0,\gm),(\al,\bt)\bigr)+4\bigl(\wp(\al-\bt)+\wp(\bt)\bigr)P_{10}\bigl(x_a,-x_d,x_c,x_b;(0,\gm),(\al,\bt)\bigr)}{\wp(\bt-\gm)^2}
\end{equation}
can be used to explicitly obtain the ABS quad equation $Q4$.  The limit \eqref{A4Q4} is finite and isolates the coefficient of the term $\wp(\bt-\gm)^2$ that appears in $P_9$ and $P_{10}$ for $A4$.  It is expected that $Q4$ appears explicitly through more simpler limits, which could perhaps be found using some M\"obius transformations of the variables that gives a simpler form of the equation.

As mentioned above, going from ABS quad equations to face-centered quad equations may be done through their associated three-leg and four-leg equations.  This can be done because the three-leg and four-leg equations share the same functions, which can be clearly seen from the lists of these functions given in Tables \ref{tab:4legquads} and \ref{tab:3legquadsH} in Appendix \ref{app:equations}.
First, consider the following four type-Q equations
\begin{equation}\label{ABStoCAFCC1}
\begin{gathered}
\quadQ{x_a}{x_{ab}}{x_{ac}}{x}{\bt_1}{\al_2}=0, \qquad
\quadQ{x_{ab}}{x_b}{x}{x_{bd}}{\bt_2}{\al_2}=0, \\
\quadQ{x_{ac}}{x}{x_c}{x_{cd}}{\bt_1}{\al_1}=0, \qquad
\quadQ{x}{x_{bd}}{x_{cd}}{x_d}{\bt_2}{\al_1}=0,
\end{gathered}
\end{equation}
which are arranged around a common vertex $x$ as drawn in Figure \ref{fig:ABStoCAFCC}.  Multiplying together the three-leg equation $(\Qd)$ for the first equation and $(\Qa)$ for the fourth equation, and dividing by both $(\Qc)$ for the second equation and $(\Qb)$ for the third equation, precisely results in a type-A four-leg equation of the form \eqref{4lega}.  The above is a known procedure to arrive at discrete-Toda (also known as discrete-Laplace) equations from the additive three-leg equations for either type-Q or type-H ABS equations \cite{AdlerPlanarGraphs,BobSurQuadGraphs,MR2467378,Suris_DiscreteTimeToda}, where the functions associated to a common edge of two neighbouring three-leg equations cancel.

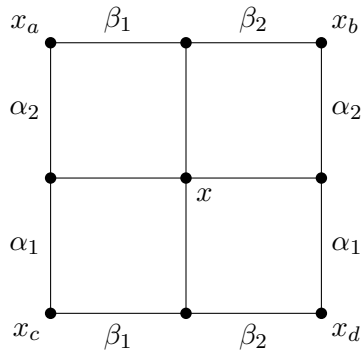
\begin{figure}[htb!]
\centering
\begin{tikzpicture}[scale=0.9]

\coordinate (A1) at (-2,2);
\coordinate (A2) at (2,2);
\coordinate (A3) at (-2,-2);
\coordinate (A4) at (2,-2);
\coordinate (A12) at (0,2);
\coordinate (A13) at (-2,0);
\coordinate (A24) at (2,0);
\coordinate (A34) at (0,-2);
\coordinate (A0) at (0,0);

\draw (A1)--(A12) node[midway,above=0pt]{$\beta_1$};
\draw (A12)--(A2) node[midway,above=0pt]{$\beta_2$};
\draw (A1)--(A13) node[midway,left=0pt]{$\alpha_2$};
\draw (A13)--(A3) node[midway,left=0pt]{$\alpha_1$};
\draw (A3)--(A34) node[midway,below=0pt]{$\beta_1$};
\draw (A34)--(A4) node[midway,below=0pt]{$\beta_2$};
\draw (A2)--(A24) node[midway,right=0pt]{$\alpha_2$};
\draw (A24)--(A4) node[midway,right=0pt]{$\alpha_1$};

\draw (A12)--(A34);\draw (A13)--(A24);

\fill (A3) circle (2.4pt) node[below left]{$x_c$};
\fill (A1) circle (2.4pt) node[above left]{$x_a$};
\fill (A2) circle (2.4pt) node[above right]{$x_b$};
\fill (A4) circle (2.4pt) node[below right]{$x_d$};
\fill (A0) circle (2.4pt) node[below right]{$x$};
\fill (A12) circle (2.4pt);
\fill (A34) circle (2.4pt);
\fill (A13) circle (2.4pt);
\fill (A24) circle (2.4pt);

\end{tikzpicture}
\caption{The four quad equations \eqref{ABStoCAFCC1}.  The three-leg equations \eqref{43legsQ} for these type-Q quad equations may be used to derive a four-leg equation \eqref{4lega} associated to a type-A face-centered quad equation.}
\label{fig:ABStoCAFCC}
\end{figure}

A similar procedure can be used to arrive at a four-leg equation for a type-C face-centered quad equation by using the three-leg equations associated to type-H quad equations.  However, this requires the introduction of a type-H quad equation with a different choice of parameters, that is known in the literature as the trapezoidal version of a type-H equation \cite{BollSuris11}.  For convenience, a different notation will be introduced for the latter.  The regular type-H quad equation will be denoted as usual by $\quadH{x_a}{x_b}{x_c}{x_d}{\alpha}{\beta}=0$, and the trapezoidal type-H version will be denoted separately as
\begin{equation}\label{typeHtrapezoidal}
\quadHu{x_a}{x_b}{x_c}{x_d}{\alpha}{\beta}=
\quadH{x_a}{x_d}{x_c}{x_b}{\beta-\alpha}{\beta}, 
\end{equation}
These different instances of the type-H equation will also be distinguished graphically as shown in Figure \ref{fig:5quads}.  These diagrams are based on the three-leg equations \eqref{43legsH}, where each of the four functions $a(x_i,x_j;\al)$, $a^\ast(x_i;x_j;\al)$, $c(x_i,x_j;\al)$, and $c^\ast(x_i;x_j;\al)$, are assigned to different edges that connect the vertices of two variables. 

\begin{figure}[htb!]
\centering
\begin{tikzpicture}[scale=1.2]

\coordinate (A1) at (-1,-1);
\coordinate (A2) at (-1,1);
\coordinate (A3) at (1,1);
\coordinate (A4) at (1,-1);

\draw[double,->-=.6] (A1)--(A2) node[midway,left=2pt]{$\beta$};
\draw[double,-<-=.55] (A2)--(A3) node[midway,above=2pt]{$\alpha$};
\draw[double,->-=.6] (A3)--(A4) node[midway,right=2pt]{$\beta$};
\draw[double,-<-=.55] (A4)--(A1) node[midway,below=2pt]{$\alpha$}node[midway,below=20pt]{$\quadH{x_a}{x_b}{x_c}{x_d}{\alpha}{\beta}=0$};

\fill (A1) circle (2pt) node[below left]{$x_c$};
\fill (A2) circle (2pt) node[above left]{$x_a$};
\fill (A3) circle (2pt) node[above right]{$x_b$};
\fill (A4) circle (2pt) node[below right]{$x_d$};

\begin{scope}[xshift=5cm]
\coordinate (A1) at (-1,-1);
\coordinate (A2) at (-1,1);
\coordinate (A3) at (1,1);
\coordinate (A4) at (1,-1);

\draw[double,-<-=.55] (A2)--(A1) node[midway,left=2pt]{$\beta$};
\draw[-!-=.53] (A2)--(A3) node[midway,above=2pt]{$\alpha$};
\draw[double,->-=.6] (A4)--(A3) node[midway,right=2pt]{$\beta$};
\draw[-] (A4)--(A1) node[midway,below=2pt]{$\alpha$}node[midway,below=15pt]{$\quadHu{x_a}{x_b}{x_c}{x_d}{\alpha}{\beta}=0$};

\fill (A1) circle (2pt) node[below left]{$x_c$};
\fill (A2) circle (2pt) node[above left]{$x_a$};
\fill (A3) circle (2pt) node[above right]{$x_b$};
\fill (A4) circle (2pt) node[below right]{$x_d$};
\end{scope}

\end{tikzpicture}
\caption{A type-H quad equation on the left, and an instance of this type-H equation with different parameters (trapezoidal version) on the right.}
\label{fig:5quads}
\end{figure}
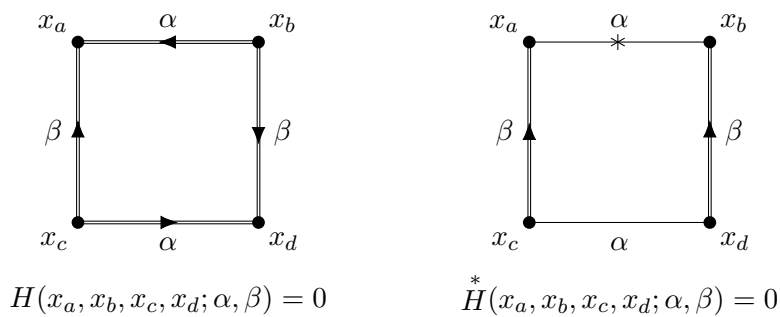

When making arrangements of type-H equations, the orientation of the equations should be chosen so that edges of neighbouring equations match according to the diagrams of Figure \ref{fig:5quads}.  This corresponds to matching the individual functions $a(x_i,x_j;\al)$, $a^\ast(x_i;x_j;\al)$, $c(x_i,x_j;\al)$, and $c^\ast(x_i;x_j;\al)$, which come from neighbouring equations.  Taking this into account, consider the following four type-H quad equations
\begin{equation}\label{ABStoCAFCC2}
\begin{gathered}
\quadH{x_{ac}}{x_a}{x}{x_{ab}}{\al_2}{\bt_1}=0, \qquad
\quadH{x_{ab}}{x_b}{x}{x_{bd}}{\bt_2}{\al_2}=0, \\
\quadHu{x_c}{x_{ac}}{x_{cd}}{x}{\al_1}{\bt_1}=0, \qquad
\quadHu{x_{bd}}{x_{d}}{x}{x_{cd}}{\al_1}{\bt_2}=0,
\end{gathered}
\end{equation}
which are arranged around a common vertex $x$ as drawn in the diagram on the left of Figure \ref{fig:ABStoCAFCC2}.  Dividing the three-leg equation $(\Hb)$ for the third equation, by the three-leg equations $(\Hc)$ for the other three equations, gives precisely the type-C four-leg equation of the form \eqref{4legc}.  It appears that discrete-Toda equations for this arrangement of type-H equations have not yet been considered in the literature, but this resembles irregular arrangements of quad equations that were previously studied in connection with weak Lax pairs and algebraic entropy \cite{HietarintaWeakLax}.

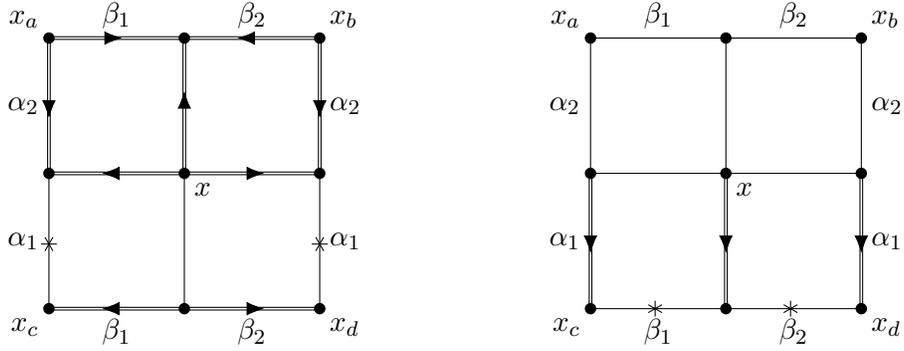
\begin{figure}[htb!]
\centering
\begin{tikzpicture}[scale=0.9]

\begin{scope}
\coordinate (A1) at (-2,2);
\coordinate (A2) at (2,2);
\coordinate (A3) at (-2,-2);
\coordinate (A4) at (2,-2);
\coordinate (A12) at (0,2);
\coordinate (A13) at (-2,0);
\coordinate (A24) at (2,0);
\coordinate (A34) at (0,-2);
\coordinate (A0) at (0,0);

\draw[-<-=.6,double] (A12)--(A1) node[midway,above=0pt]{$\beta_1$};
\draw[->-=.65,double] (A2)--(A12) node[midway,above=0pt]{$\beta_2$};
\draw[-<-=.55,double] (A13)--(A1) node[midway,left=0pt]{$\alpha_2$};
\draw[-!-=.53] (A3)--(A13) node[midway,left=0pt]{$\alpha_1$};
\draw[->-=.65,double] (A34)--(A3) node[midway,below=0pt]{$\beta_1$};
\draw[-<-=.55,double] (A4)--(A34) node[midway,below=0pt]{$\beta_2$};
\draw[-<-=.55,double] (A24)--(A2) node[midway,right=0pt]{$\alpha_2$};
\draw[-!-=.53] (A4)--(A24) node[midway,right=0pt]{$\alpha_1$};

\draw[->-=.65,double] (A0)--(A12);
\draw (A0)--(A34);
\draw[->-=.65,double] (A0)--(A13);
\draw[-<-=.55,double] (A24)--(A0);

\fill (A3) circle (2.4pt) node[below left]{$x_c$};
\fill (A1) circle (2.4pt) node[above left]{$x_a$};
\fill (A2) circle (2.4pt) node[above right]{$x_b$};
\fill (A4) circle (2.4pt) node[below right]{$x_d$};
\fill (A0) circle (2.4pt) node[below right]{$x$};
\fill (A12) circle (2.4pt);
\fill (A34) circle (2.4pt);
\fill (A13) circle (2.4pt);
\fill (A24) circle (2.4pt);
\end{scope}

\begin{scope}[xshift=8cm]
\coordinate (A1) at (-2,2);
\coordinate (A2) at (2,2);
\coordinate (A3) at (-2,-2);
\coordinate (A4) at (2,-2);
\coordinate (A12) at (0,2);
\coordinate (A13) at (-2,0);
\coordinate (A24) at (2,0);
\coordinate (A34) at (0,-2);
\coordinate (A0) at (0,0);

\draw (A1)--(A12) node[midway,above=0pt]{$\beta_1$};
\draw (A12)--(A2) node[midway,above=0pt]{$\beta_2$};
\draw (A1)--(A13) node[midway,left=0pt]{$\alpha_2$};
\draw[-<-=.55,double] (A3)--(A13) node[midway,left=0pt]{$\alpha_1$};
\draw[-!-=.53] (A3)--(A34) node[midway,below=0pt]{$\beta_1$};
\draw[-!-=.53] (A34)--(A4) node[midway,below=0pt]{$\beta_2$};
\draw (A2)--(A24) node[midway,right=0pt]{$\alpha_2$};
\draw[-<-=.55,double] (A4)--(A24) node[midway,right=0pt]{$\alpha_1$};

\draw (A12)--(A0);
\draw[-<-=.55,double] (A34)--(A0);
\draw (A13)--(A0);
\draw (A0)--(A24);

\fill (A3) circle (2.4pt) node[below left]{$x_c$};
\fill (A1) circle (2.4pt) node[above left]{$x_a$};
\fill (A2) circle (2.4pt) node[above right]{$x_b$};
\fill (A4) circle (2.4pt) node[below right]{$x_d$};
\fill (A0) circle (2.4pt) node[below right]{$x$};
\fill (A12) circle (2.4pt);
\fill (A34) circle (2.4pt);
\fill (A13) circle (2.4pt);
\fill (A24) circle (2.4pt);
\end{scope}

\end{tikzpicture}
\caption{The four quad equations for \eqref{ABStoCAFCC2} on the left, and \eqref{ABStoCAFCC3} on the right.  The three-leg equations \eqref{43legsQ} and \eqref{43legsH} for these type-Q and type-H quad equations may be used to derive a four-leg equation \eqref{4legc} associated to a type-C face-centered quad equation.}
\label{fig:ABStoCAFCC2}
\end{figure}

The arrangement of equations in \eqref{ABStoCAFCC2} is not the only one that will lead to four-leg equations for type-C equations.  For example, the three-leg equations for the following type-Q and type-H equations (shown on the diagram on the right of of Figure \ref{fig:ABStoCAFCC2})
\begin{equation}\label{ABStoCAFCC3}
\begin{gathered}
\quadQ{x_a}{x_{ab}}{x_{ac}}{x}{\bt_1}{\al_2}=0, \qquad
\quadQ{x_{ab}}{x_b}{x}{x_{bd}}{\bt_2}{\al_2}=0, \\
\quadHu{x_{cd}}{x_{c}}{x}{x_{ac}}{\bt_1}{\al_2}=0, \qquad
\quadHu{x_{d}}{x_{cd}}{x_{bd}}{x}{\bt_2}{\al_1}=0,
\end{gathered}
\end{equation}
will also lead to the type-C four-leg equation \eqref{4legc}, but with opposite signs of the parameters.

\section{Systems of hex equations in the hexagonal lattice}\label{sec:hexagon}


The consistent systems of face-centered quad equations given by \eqref{64legsA} and \eqref{64legsC} will be used to define new types of discrete systems of equations that evolve in the hexagonal lattice.  For this purpose it turns out that instead of the two-component parameters $\bal=(\al_1,\al_2)$ and $\bbt=(\bt_1,\bt_2)$, it is more convenient to consider equations in terms of three scalar parameters $\al,\bt,\gm$.  The relation between the two sets of parameters will be chosen as
\begin{equation}\label{parametercov}
\al=\bt_1-\al_1,\quad \bt=\bt_1-\al_2,\quad \gm=\bt_1-\bt_2.
\end{equation}
The parameters $\al,\bt,\gm$ will end up being assigned to edges in the hexagonal lattice, where opposite edges of each hexagonal unit cell are always assigned the same parameter.

Each of the face-centered quad equations of Table \ref{tab:CAFCClist} may be written in terms of the parameters $\al,\bt,\gm$, and such expressions can be found in Appendix \ref{app:equations}.  The face-centered quad equation of Figure \ref{fig:fcquadA} is drawn in terms of the three parameters $\al,\bt,\gm$, in the diagram on the left of Figure \ref{fig:3parameterfcqe}.  Also shown in Figure \ref{fig:3parameterfcqe} are graphical representations of a pair of type-C equations $C$ and $\overline{C}$, distinguished by the orientations of directed edges.

\begin{figure}[htb!]
\centering
\begin{tikzpicture}[scale=1.6]

\coordinate (A0) at (0,0);
\coordinate (A1) at (-1,1);
\coordinate (A2) at (1,1);
\coordinate (A3) at (-1,-1);
\coordinate (A4) at (1,-1);

\draw[gray,very thin,dashed] (A1)--(A2)--(A4)--(A3)--cycle;

\draw (A0)--(A1) node[midway,left=2pt]{\small$\bt$};
\draw (A0)--(A2) node[midway,right=2pt]{\small$\bt-\gm$};
\draw (A0)--(A3) node[midway,left=2pt]{\small$\al$};
\draw (A0)--(A4) node[midway,right=2pt]{\small$\al-\gm$};

\fill (A3) circle (1.7pt)
node[left=1.5pt]{\color{black} $x_c$};
\fill (A1) circle (1.7pt)
node[left=1.5pt]{\color{black} $x_a$};
\fill (A2) circle (1.7pt)
node[right=1.5pt]{\color{black} $x_b$};
\fill (A4) circle (1.7pt)
node[right=1.5pt]{\color{black} $x_d$};
\fill (A0) circle (1.7pt)
node[left=1.5pt]{\color{black} $x$};


\fill (0,-1.3) circle (0.01pt)
node[below=0.0pt]{\color{black}\small $\A{\ccx}{x_a}{x_b}{x_c}{x_d}{\al}{\bt}{\gm}=0$};

\begin{scope}[xshift=100pt]

\coordinate (A0) at (0,0);
\coordinate (A1) at (-1,1);
\coordinate (A2) at (1,1);
\coordinate (A3) at (-1,-1);
\coordinate (A4) at (1,-1);

\draw[gray,very thin,dashed] (A1)--(A2)--(A4)--(A3)--cycle;

\draw (A0)--(A1) node[midway,left=4pt]{\small$\bt$};
\draw (A0)--(A2) node[midway,right=4pt]{\small$\bt-\gm$};
\draw[->-=.6,double] (A0)--(A3) node[midway,left=4pt]{\small$\al$};
\draw[->-=.6,double] (A0)--(A4) node[midway,right=4pt]{\small$\al-\gm$};

\fill (A3) circle (1.7pt)
node[left=1.5pt]{\color{black} $x_c$};
\fill (A1) circle (1.7pt)
node[left=1.5pt]{\color{black} $x_a$};
\fill (A2) circle (1.7pt)
node[right=1.5pt]{\color{black} $x_b$};
\fill (A4) circle (1.7pt)
node[right=1.5pt]{\color{black} $x_d$};
\fill (A0) circle (1.7pt)
node[left=1.5pt]{\color{black} $x$};


\fill (0,-1.3) circle (0.01pt)
node[below=0.0pt]{\color{black}\small $\C{\ccx}{x_a}{x_b}{x_c}{x_d}{\al}{\bt}{\gm}=0$};

\end{scope}

\begin{scope}[xshift=200pt]

\coordinate (A0) at (0,0);
\coordinate (A1) at (-1,1);
\coordinate (A2) at (1,1);
\coordinate (A3) at (-1,-1);
\coordinate (A4) at (1,-1);

\draw[gray,very thin,dashed] (A1)--(A2)--(A4)--(A3)--cycle;

\draw[-!-=.53] (A0)--(A1) node[midway,left=4pt]{\small$\bt$};
\draw[-!-=.53] (A0)--(A2) node[midway,right=4pt]{\small$\bt-\gm$};
\draw[-<-=.6,double] (A0)--(A3) node[midway,left=4pt]{\small$\al$};
\draw[-<-=.6,double] (A0)--(A4) node[midway,right=4pt]{\small$\al-\gm$};

\fill (A3) circle (1.7pt)
node[left=1.5pt]{\color{black} $x_c$};
\fill (A1) circle (1.7pt)
node[left=1.5pt]{\color{black} $x_a$};
\fill (A2) circle (1.7pt)
node[right=1.5pt]{\color{black} $x_b$};
\fill (A4) circle (1.7pt)
node[right=1.5pt]{\color{black} $x_d$};
\fill (A0) circle (1.7pt)
node[left=1.5pt]{\color{black} $x$};


\fill (0,-1.3) circle (0.01pt)
node[below=0.0pt]{\color{black}\small $\Cbar{\ccx}{x_a}{x_b}{x_c}{x_d}{\al}{\bt}{\gm}=0$};

\end{scope}

\end{tikzpicture}
\caption{Face-centered quad equations in terms of the parameters \eqref{parametercov}.}
\label{fig:3parameterfcqe}
\end{figure}
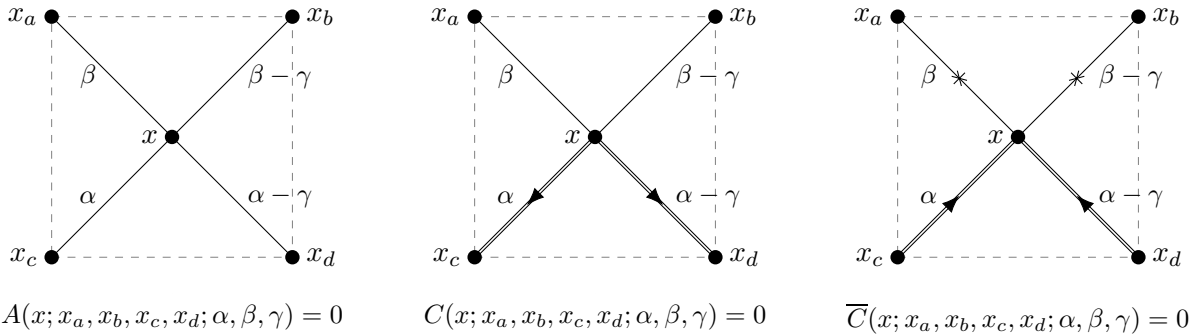

Based on \eqref{64legsA}, $\hexA{\xa}{\xb}{\xf}{\xd}{\xc}{\xe}{\al}{\bt}{\gm}$ is defined to be the following system of type-A face-centered quad equations
\begin{equation}\label{Ahexdef}
\begin{alignedat}{2}
(\Aa)&&\qquad\A{\xa}{\xe}{\xc}{\xb}{\xf}{\gm}{\bt}{\al}=0, \\
(\Ab)&&\qquad\A{\xb}{\xf}{\xd}{\xa}{\xe}{\gm}{\al}{\bt}=0, \\
(\Af)&&\qquad\A{\xf}{\xd}{\xc}{\xb}{\xa}{\al}{\bt}{\gm}=0, \\
(\Ad)&&\qquad\A{\xd}{\xf}{\xb}{\xc}{\xe}{\gm}{\bt}{\al}=0, \\
(\Ac)&&\qquad\A{\xc}{\xe}{\xa}{\xd}{\xf}{\gm}{\al}{\bt}=0, \\
(\Ae)&&\qquad\A{\xe}{\xa}{\xb}{\xc}{\xd}{\al}{\bt}{\gm}=0,
\end{alignedat}
\end{equation}
which each depend on different combinations of five of the six variables $\xa,\ldots,\xe$, and on the three parameters $\al,\bt,\gm$.  
The system of equations $\hexA{\xa}{\xb}{\xf}{\xd}{\xc}{\xe}{\al}{\bt}{\gm}$ is assigned to the hexagon shown in the diagram on the left of Figure \ref{fig:hexagonunitcell}.  The edges of this diagram are consistent with the edges  from Figure \ref{fig:3parameterfcqe} associated to the individual equations in \eqref{Ahexdef}.  To recognise that here the equations are associated to a hexagonal unit cell rather than a face-centered cubic structure, $\hexA{\xa}{\xb}{\xf}{\xd}{\xc}{\xe}{\al}{\bt}{\gm}$ will be referred to as a {\it system of type-A hex equations}.

\begin{figure}[htb!]
\centering
\begin{tikzpicture}[scale=2.0]

\pgfmathsetmacro\tr{sqrt(3)/2}

\coordinate (A1) at (-1,0);
\coordinate (A2) at (-1/2,{\tr});
\coordinate (A3) at (1/2,{\tr});
\coordinate (A4) at (1,0);
\coordinate (A5) at (1/2,{-\tr});
\coordinate (A6) at (-1/2,{-\tr});

\draw (A1)--(A2) node[midway,above left]{$\gamma$};
\draw (A2)--(A3) node[midway,above]{$\alpha$};
\draw (A3)--(A4) node[midway,above right]{$\beta$};
\draw (A4)--(A5) node[midway,below right]{$\gamma$};
\draw (A5)--(A6) node[midway,below]{$\alpha$};
\draw (A6)--(A1) node[midway,below left]{$\beta$};

\fill (A1) circle (1.4pt) node[left]{$\xa$};
\fill (A2) circle (1.4pt) node[above left]{$\xb$};
\fill (A3) circle (1.4pt) node[above right]{$\xf$};
\fill (A4) circle (1.4pt) node[right]{$\xd$};
\fill (A5) circle (1.4pt) node[below right]{$\xc$};
\fill (A6) circle (1.4pt) node[below left]{$\xe$};

\fill (0,-1.2) circle (0.01pt)
node[below=0.0pt]{\color{black}\small $\hexA{\xa}{\xb}{\xf}{\xd}{\xc}{\xe}{\al}{\bt}{\gm}$};

\begin{scope}[xshift=120pt]

\coordinate (A1) at (-1,0);
\coordinate (A2) at (-1/2,{\tr});
\coordinate (A3) at (1/2,{\tr});
\coordinate (A4) at (1,0);
\coordinate (A5) at (1/2,{-\tr});
\coordinate (A6) at (-1/2,{-\tr});

\draw[double,-<-=.55] (A1)--(A2) node[midway,above left]{$\gamma$};
\draw (A2)--(A3) node[midway,above]{$\alpha$};
\draw (A3)--(A4) node[midway,above right]{$\beta$};
\draw[double,->-=.6] (A4)--(A5) node[midway,below right]{$\gamma$};
\draw[-!-=.53] (A5)--(A6) node[midway,below]{$\alpha$};
\draw[-!-=.53] (A6)--(A1) node[midway,below left]{$\beta$};

\fill (A1) circle (1.4pt) node[left]{$\xa$};
\fill (A2) circle (1.4pt) node[above left]{$\xb$};
\fill (A3) circle (1.4pt) node[above right]{$\xf$};
\fill (A4) circle (1.4pt) node[right]{$\xd$};
\fill (A5) circle (1.4pt) node[below right]{$\xc$};
\fill (A6) circle (1.4pt) node[below left]{$\xe$};

\fill (0,-1.2) circle (0.01pt)
node[below=0.0pt]{\color{black}\small $\hexC{\xa}{\xb}{\xf}{\xd}{\xc}{\xe}{\al}{\bt}{\gm}$};

\end{scope}

\end{tikzpicture}
\caption{ The hexagon on the left is associated to the system of type-A hex equations \eqref{Ahexdef}, and the hexagon on the right is associated to the system of type-C hex equations \eqref{Chexdef}.}
\label{fig:hexagonunitcell}
\end{figure}
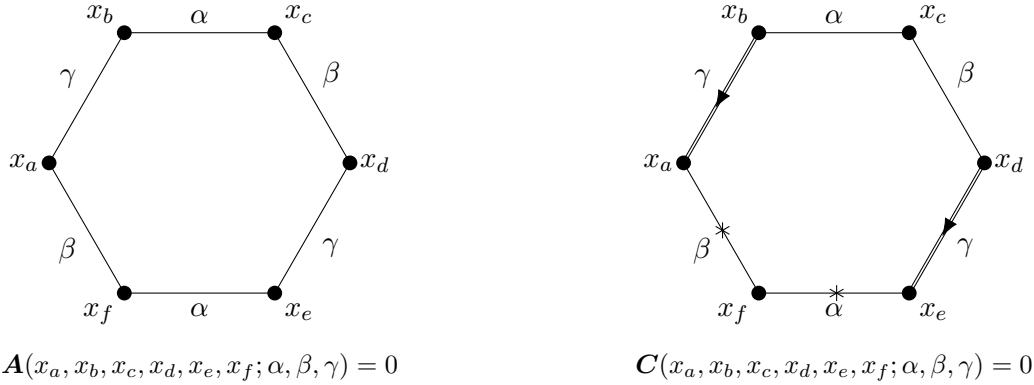

Based on \eqref{64legsC}, $\hexC{\xa}{\xb}{\xf}{\xd}{\xc}{\xe}{\al}{\bt}{\gm}$ is defined to be the following system of type-C face-centered quad equations
\begin{equation}\label{Chexdef}
\begin{alignedat}{2}
(\Ca)&&\qquad\Cbar{\xa}{\xe}{\xc}{\xb}{\xf}{\gm}{\bt}{\al}=0, \\
(\Cb)&&\qquad\C{\xb}{\xf}{\xd}{\xa}{\xe}{\gm}{\al}{\bt}=0, \\
(\Cf)&&\qquad\C{\xf}{\xb}{\xd}{\xa}{\xc}{\gm-\al}{-\al}{\bt-\al}=0, \\
(\Cd)&&\qquad\C{\xd}{\xf}{\xb}{\xc}{\xe}{\gm}{\bt}{\al}=0, \\
(\Cc)&&\qquad\Cbar{\xc}{\xe}{\xa}{\xd}{\xf}{\gm}{\al}{\bt}=0, \\
(\Ce)&&\qquad\Cbar{\xe}{\xc}{\xa}{\xd}{\xb}{\gm-\al}{-\al}{\bt-\al}=0,
\end{alignedat}
\end{equation}
which each depend on different combinations of five of the six variables $\xa,\ldots,\xe$, and on the three parameters $\al,\bt,\gm$.  
The system of equations $\hexC{\xa}{\xb}{\xf}{\xd}{\xc}{\xe}{\al}{\bt}{\gm}$ is assigned to the hexagon shown in the diagram on the right of Figure \ref{fig:hexagonunitcell}.  The edges of this diagram are consistent with the edges from Figure \ref{fig:3parameterfcqe} associated to the individual equations in \eqref{Chexdef}. Since here the equations are defined on a hexagonal unit cell, $\hexC{\xa}{\xb}{\xf}{\xd}{\xc}{\xe}{\al}{\bt}{\gm}$ will be referred to as a {\it system of type-C hex equations}.

In terms of the parameters $\al,\bt,\gm$, the square symmetries \eqref{symmetriesfcqe} unfortunately take a more cumbersome form
\begin{equation}\label{fcqsyms}
\begin{split}
\A{x}{x_a}{x_b}{x_c}{x_d}{\al}{\bt}{\gm}&=-\A{x}{x_b}{x_a}{x_d}{x_c}{\al-\gm}{\bt-\gm}{-\gm}, \\
\A{x}{x_a}{x_b}{x_c}{x_d}{\al}{\bt}{\gm}&=-\A{x}{x_c}{x_d}{x_a}{x_b}{\bt}{\al}{\gm}, \\
\A{x}{x_a}{x_b}{x_c}{x_d}{\al}{\bt}{\gm}&=-\A{x}{x_d}{x_b}{x_c}{x_a}{-\al}{\al-\gm}{\bt-\al}.
\end{split}
\end{equation}
However, the hexagonal symmetries satisfied by systems of hex equations can be better expressed in terms of $\al,\bt,\gm$.  
Unfortunately, there appears to be no symmetry that can be used to write the equations $(\Cf)$ and $(\Ce)$ with simpler forms of the parameter dependence as has been given for the type-A hex equations $(\Af)$ and $(\Ae)$ in \eqref{Ahexdef}.

The system of type-A hex equations $\hexA{\xa}{\xb}{\xf}{\xd}{\xc}{\xe}{\al}{\bt}{\gm}$ satisfies each of the following hexagonal symmetries 
\begin{equation}\label{d6syms}
\begin{split}
&\hexAneq{\xa}{\xb}{\xf}{\xd}{\xc}{\xe}{\al}{\bt}{\gm}
=\hexAneq{\xc}{\xd}{\xf}{\xb}{\xa}{\xe}{\bt}{\al}{\gm}, \\
&\hexAneq{\xa}{\xb}{\xf}{\xd}{\xc}{\xe}{\al}{\bt}{\gm}
=\hexAneq{\xb}{\xa}{\xe}{\xc}{\xd}{\xf}{\bt}{\al}{\gm}, \\
&\hexAneq{\xa}{\xb}{\xf}{\xd}{\xc}{\xe}{\al}{\bt}{\gm}
=\hexAneq{\xb}{\xf}{\xd}{\xc}{\xe}{\xa}{\bt}{\gm}{\al}. \\
\end{split}
\end{equation}
The symmetries \eqref{d6syms} are shown in the diagrams of Figure \ref{fig:d6syms}.  The equalities \eqref{d6syms} follow from the use of the second symmetry of \eqref{fcqsyms}.  The first of \eqref{d6syms} is a reflection symmetry which simply has the effect of exchanging the equations $(\Aa)\leftrightarrow(\Ac)$ and $(\Ab)\leftrightarrow(\Ad)$.  The second of \eqref{d6syms} is another reflection symmetry which simply has the effect of exchanging the equations $(\Aa)\leftrightarrow(\Ab)$, $(\Af)\leftrightarrow(\Ae)$, $(\Ad)\leftrightarrow(\Ac)$.  The third of \eqref{d6syms} is a symmetry of counterclockwise rotation by $\pi/3$ which simply has the effect of cycling the equations $\bigl((\Aa),(\Ab),(\Af),(\Ad),(\Ac),(\Ae)\bigr)\to\bigl((\Ab),(\Af),(\Ad),(\Ac),(\Ae),(\Aa)\bigr)$.

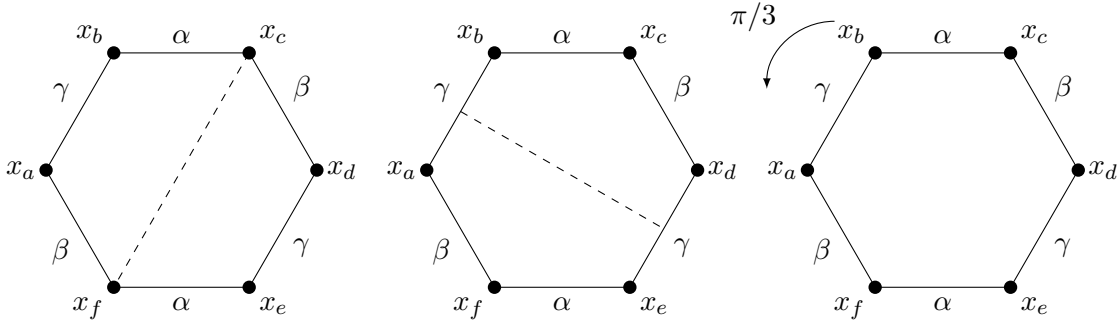
\begin{figure}[htb!]
\centering
\begin{tikzpicture}[scale=1.8]

\begin{scope}[xshift=80pt]

\pgfmathsetmacro\tr{sqrt(3)/2}

\coordinate (A1) at (-1,0);
\coordinate (A2) at (-1/2,{\tr});
\coordinate (A3) at (1/2,{\tr});
\coordinate (A4) at (1,0);
\coordinate (A5) at (1/2,{-\tr});
\coordinate (A6) at (-1/2,{-\tr});

\draw (A1)--(A2) node[midway,above left]{$\gamma$} coordinate[midway](A12);
\draw (A2)--(A3) node[midway,above]{$\alpha$};
\draw (A3)--(A4) node[midway,above right]{$\beta$};
\draw (A4)--(A5) node[midway,below right]{$\gamma$} coordinate[midway](A45);
\draw (A5)--(A6) node[midway,below]{$\alpha$};
\draw (A6)--(A1) node[midway,below left]{$\beta$};

\draw[dashed] (A12)--(A45);


\fill (A1) circle (1.4pt) node[left]{$\xa$};
\fill (A2) circle (1.4pt) node[above left]{$\xb$};
\fill (A3) circle (1.4pt) node[above right]{$\xf$};
\fill (A4) circle (1.4pt) node[right]{$\xd$};
\fill (A5) circle (1.4pt) node[below right]{$\xc$};
\fill (A6) circle (1.4pt) node[below left]{$\xe$};

\end{scope}

\begin{scope}[xshift=0pt]

\pgfmathsetmacro\tr{sqrt(3)/2}

\coordinate (A1) at (-1,0);
\coordinate (A2) at (-1/2,{\tr});
\coordinate (A3) at (1/2,{\tr});
\coordinate (A4) at (1,0);
\coordinate (A5) at (1/2,{-\tr});
\coordinate (A6) at (-1/2,{-\tr});

\draw (A1)--(A2) node[midway,above left]{$\gamma$} coordinate[midway](A12);
\draw (A2)--(A3) node[midway,above]{$\alpha$};
\draw (A3)--(A4) node[midway,above right]{$\beta$};
\draw (A4)--(A5) node[midway,below right]{$\gamma$} coordinate[midway](A45);
\draw (A5)--(A6) node[midway,below]{$\alpha$};
\draw (A6)--(A1) node[midway,below left]{$\beta$};

\draw[dashed] (A3)--(A6);


\fill (A1) circle (1.4pt) node[left]{$\xa$};
\fill (A2) circle (1.4pt) node[above left]{$\xb$};
\fill (A3) circle (1.4pt) node[above right]{$\xf$};
\fill (A4) circle (1.4pt) node[right]{$\xd$};
\fill (A5) circle (1.4pt) node[below right]{$\xc$};
\fill (A6) circle (1.4pt) node[below left]{$\xe$};

\end{scope}

\begin{scope}[xshift=160pt]

\pgfmathsetmacro\tr{sqrt(3)/2}

\coordinate (A1) at (-1,0);
\coordinate (A2) at (-1/2,{\tr});
\coordinate (A3) at (1/2,{\tr});
\coordinate (A4) at (1,0);
\coordinate (A5) at (1/2,{-\tr});
\coordinate (A6) at (-1/2,{-\tr});

\draw (A1)--(A2) node[midway,above left]{$\gamma$} coordinate[midway](A12);
\draw (A2)--(A3) node[midway,above]{$\alpha$};
\draw (A3)--(A4) node[midway,above right]{$\beta$};
\draw (A4)--(A5) node[midway,below right]{$\gamma$} coordinate[midway](A45);
\draw (A5)--(A6) node[midway,below]{$\alpha$};
\draw (A6)--(A1) node[midway,below left]{$\beta$};

\draw[latex-] (-1.3,0.6) arc (180:90:0.5) node[midway,above left]{$\pi/3$};


\fill (A1) circle (1.4pt) node[left]{$\xa$};
\fill (A2) circle (1.4pt) node[above left]{$\xb$};
\fill (A3) circle (1.4pt) node[above right]{$\xf$};
\fill (A4) circle (1.4pt) node[right]{$\xd$};
\fill (A5) circle (1.4pt) node[below right]{$\xc$};
\fill (A6) circle (1.4pt) node[below left]{$\xe$};

\end{scope}

\end{tikzpicture}
\caption{Line of reflection for the first of \eqref{d6syms} is shown on the left, and for the second of \eqref{d6syms} in the center.  The third of \eqref{d6syms} represents a counterclockwise rotation by $\pi/3$ as shown on the right.  The system of type-A hex equations is invariant under each of these actions, while the system of type-C hex equations is only invariant under the first reflection.}
\label{fig:d6syms}
\end{figure}

As might be expected, the system of type-C hex equations $\hexC{\xa}{\xb}{\xf}{\xd}{\xc}{\xe}{\al}{\bt}{\gm}$ has less symmetry, and is only invariant under the the first reflection symmetry of \eqref{d6syms}, {\it i.e.},
\begin{equation}
\hexCneq{\xa}{\xb}{\xf}{\xd}{\xc}{\xe}{\al}{\bt}{\gm}
=\hexCneq{\xc}{\xd}{\xf}{\xb}{\xa}{\xe}{\bt}{\al}{\gm}.
\end{equation}
The second reflection symmetry of \eqref{d6syms} has the effect of exchanging the equations $(\Ca)\leftrightarrow(\Cb)$, $(\Cf)\leftrightarrow(\Ce)$, $(\Cd)\leftrightarrow(\Cc)$, so that equations $C$ become equations $\overline{C}$, and vice-versa (so the system is invariant when $\overline{C}=C$).  Both of the latter symmetries follow from the use of the first symmetry of \eqref{fcqsyms} (the only one satisfied by type-C equations).  The third symmetry in \eqref{d6syms} is also not satisfied by $\hexC{\xa}{\xb}{\xf}{\xd}{\xc}{\xe}{\al}{\bt}{\gm}$, but the following rotation of the hexagon by $\pi$
\begin{equation}
\xa\leftrightarrow\xd,\quad \xb\leftrightarrow\xc, \quad \xf\leftrightarrow\xe,
\end{equation}
has the effect of exchanging the equations $(\Ca)\leftrightarrow(\Cd)$, $(\Cb)\leftrightarrow(\Cc)$, $(\Cf)\leftrightarrow(\Ce)$, so that once again the equations $C$ become equations $\overline{C}$, and vice-versa.

Under appropriate initial conditions, if the systems of hex equations \eqref{Ahexdef} and \eqref{Chexdef} are consistent (in the sense of Theorem \ref{thm:CAH}) one may define unique evolutions of systems of hex equations in the hexagonal lattice. 
This relies on the fact that once the four variables associated to four consecutive vertices of a hexagonal face have been determined, a system of hex equations may be consistently solved for the variables on the remaining two vertices.

For example, if the four initial variables on an individual hexagon are chosen as $x_a,x_b,x_e,x_f$, the remaining variables $x_c$ and $x_d$ may be uniquely determined from two of the six type-A hex equations from \eqref{Ahexdef}, or from two of the six type-C hex equations from \eqref{Chexdef}.  There are a total of eight ways to solve for these variables as is illustrated in the diagram of Figure \ref{fig:hexevolution} for a system of type-A hex equations.  For consistent systems, the choice of which of the four equations is used to solve for $x_c$ or $x_d$ in the second step does not matter.

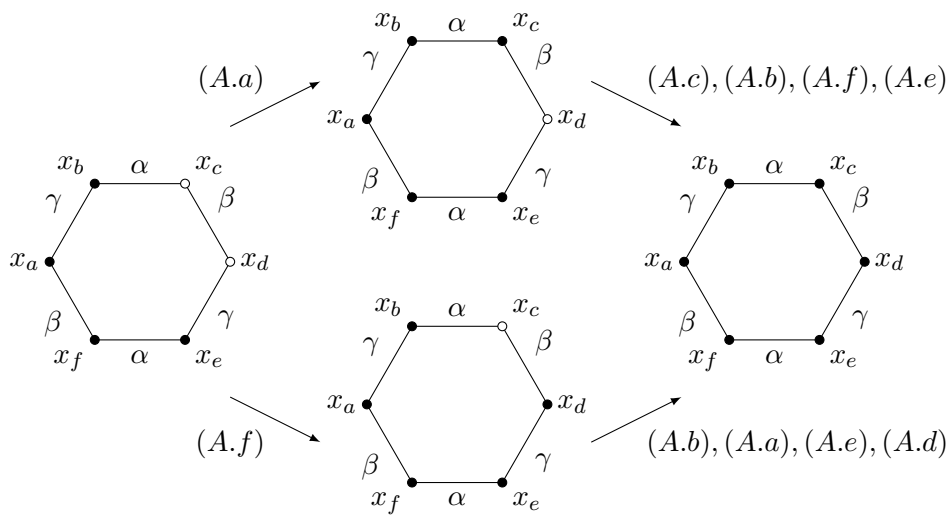
\begin{figure}[htb!]
\centering
\begin{tikzpicture}[scale=1.2]

\pgfmathsetmacro\tr{sqrt(3)/2}

\coordinate (A1) at (-1,0);
\coordinate (A2) at (-1/2,{\tr});
\coordinate (A3) at (1/2,{\tr});
\coordinate (A4) at (1,0);
\coordinate (A5) at (1/2,{-\tr});
\coordinate (A6) at (-1/2,{-\tr});

\draw (A1)--(A2) node[midway,above left]{$\gamma$};
\draw (A2)--(A3) node[midway,above]{$\alpha$};
\draw (A3)--(A4) node[midway,above right]{$\beta$};
\draw (A4)--(A5) node[midway,below right]{$\gamma$};
\draw (A5)--(A6) node[midway,below]{$\alpha$};
\draw (A6)--(A1) node[midway,below left]{$\beta$};


\filldraw[fill=black,draw=black] (A1) circle (1.4pt) node[left]{$\xa$};
\filldraw[fill=black,draw=black] (A2) circle (1.4pt) node[above left]{$\xb$};
\filldraw[fill=white,draw=black] (A3) circle (1.4pt) node[above right]{$\xf$};
\filldraw[fill=white,draw=black] (A4) circle (1.4pt) node[right]{$\xd$};
\filldraw[fill=black,draw=black] (A5) circle (1.4pt) node[below right]{$\xc$};
\filldraw[fill=black,draw=black] (A6) circle (1.4pt) node[below left]{$\xe$};

\draw[-latex] (1,1.5)--(2,2) node[midway,above left]{$(\Aa)$};
\draw[-latex] (1,-1.5)--(2,-2) node[midway,below left]{$(\Ae)$};

\draw[-latex] (5,2)--(6,1.5) node[midway,above right]{$(\Af),(\Ab),(\Ae),(\Ac)$};
\draw[-latex] (5,-2)--(6,-1.5) node[midway,below right]{$(\Ab),(\Aa),(\Ac),(\Ad)$};

\begin{scope}[xshift=100pt,yshift=45pt]

\coordinate (A1) at (-1,0);
\coordinate (A2) at (-1/2,{\tr});
\coordinate (A3) at (1/2,{\tr});
\coordinate (A4) at (1,0);
\coordinate (A5) at (1/2,{-\tr});
\coordinate (A6) at (-1/2,{-\tr});

\draw (A1)--(A2) node[midway,above left]{$\gamma$};
\draw (A2)--(A3) node[midway,above]{$\alpha$};
\draw (A3)--(A4) node[midway,above right]{$\beta$};
\draw (A4)--(A5) node[midway,below right]{$\gamma$};
\draw (A5)--(A6) node[midway,below]{$\alpha$};
\draw (A6)--(A1) node[midway,below left]{$\beta$};


\filldraw[fill=black,draw=black] (A1) circle (1.4pt) node[left]{$\xa$};
\filldraw[fill=black,draw=black] (A2) circle (1.4pt) node[above left]{$\xb$};
\filldraw[fill=black,draw=black] (A3) circle (1.4pt) node[above right]{$\xf$};
\filldraw[fill=white,draw=black] (A4) circle (1.4pt) node[right]{$\xd$};
\filldraw[fill=black,draw=black] (A5) circle (1.4pt) node[below right]{$\xc$};
\filldraw[fill=black,draw=black] (A6) circle (1.4pt) node[below left]{$\xe$};

\end{scope}

\begin{scope}[xshift=100pt,yshift=-45pt]

\coordinate (A1) at (-1,0);
\coordinate (A2) at (-1/2,{\tr});
\coordinate (A3) at (1/2,{\tr});
\coordinate (A4) at (1,0);
\coordinate (A5) at (1/2,{-\tr});
\coordinate (A6) at (-1/2,{-\tr});

\draw (A1)--(A2) node[midway,above left]{$\gamma$};
\draw (A2)--(A3) node[midway,above]{$\alpha$};
\draw (A3)--(A4) node[midway,above right]{$\beta$};
\draw (A4)--(A5) node[midway,below right]{$\gamma$};
\draw (A5)--(A6) node[midway,below]{$\alpha$};
\draw (A6)--(A1) node[midway,below left]{$\beta$};


\filldraw[fill=black,draw=black] (A1) circle (1.4pt) node[left]{$\xa$};
\filldraw[fill=black,draw=black] (A2) circle (1.4pt) node[above left]{$\xb$};
\filldraw[fill=white,draw=black] (A3) circle (1.4pt) node[above right]{$\xf$};
\filldraw[fill=black,draw=black] (A4) circle (1.4pt) node[right]{$\xd$};
\filldraw[fill=black,draw=black] (A5) circle (1.4pt) node[below right]{$\xc$};
\filldraw[fill=black,draw=black] (A6) circle (1.4pt) node[below left]{$\xe$};

\end{scope}

\begin{scope}[xshift=200pt,yshift=0pt]

\coordinate (A1) at (-1,0);
\coordinate (A2) at (-1/2,{\tr});
\coordinate (A3) at (1/2,{\tr});
\coordinate (A4) at (1,0);
\coordinate (A5) at (1/2,{-\tr});
\coordinate (A6) at (-1/2,{-\tr});

\draw (A1)--(A2) node[midway,above left]{$\gamma$};
\draw (A2)--(A3) node[midway,above]{$\alpha$};
\draw (A3)--(A4) node[midway,above right]{$\beta$};
\draw (A4)--(A5) node[midway,below right]{$\gamma$};
\draw (A5)--(A6) node[midway,below]{$\alpha$};
\draw (A6)--(A1) node[midway,below left]{$\beta$};


\filldraw[fill=black,draw=black] (A1) circle (1.4pt) node[left]{$\xa$};
\filldraw[fill=black,draw=black] (A2) circle (1.4pt) node[above left]{$\xb$};
\filldraw[fill=black,draw=black] (A3) circle (1.4pt) node[above right]{$\xf$};
\filldraw[fill=black,draw=black] (A4) circle (1.4pt) node[right]{$\xd$};
\filldraw[fill=black,draw=black] (A5) circle (1.4pt) node[below right]{$\xc$};
\filldraw[fill=black,draw=black] (A6) circle (1.4pt) node[below left]{$\xe$};

\end{scope}

\end{tikzpicture}
\caption{Initial value problem on a hexagon.  Initial variables are indicated by black vertices. For the first step either $(\Aa)$ may be used to solve for $\xf$, or $(\Ae)$ may be used to solve for $\xd$.  For either choice, there will then be four equations that can be used to determine the remaining variable.}
\label{fig:hexevolution}
\end{figure}

The three different parameters $\al$, $\bt$, $\gm$, on the hexagonal faces of Figure \ref{fig:hexagonunitcell}, may be associated with three different directions in the hexagonal lattice, such that opposite edges of each hexagonal face are always associated with the same parameter. The example in Figure \ref{fig:hexevolution} would correspond to an evolution in the hexagonal lattice in the direction associated to the parameter $\bt$. The evolutions of systems of hex equations in other directions of the hexagonal lattice would involve different choices of four initial variables and different choices of the hex equations that would be used to solve for the two unknown variables.

Initial conditions based on the diagram of Figure \ref{fig:hexevolution} are given on the hexagonal lattices shown in Figure \ref{fig:hexlattice1}.  These may be regarded as analogues of the staircase (or sawtooth) and corner initial conditions, respectively, which are two commonly used types of initial conditions for the square lattice.  In both cases, the consistent systems of hex equations can be used to uniquely determine the unknown variables on white vertices.  Initial conditions for other directions in the hexagonal lattice may be obtained simply by rotations by $\pi/3$.

\newcommand\hx{12pt}
\newcommand\vs{1.7pt}

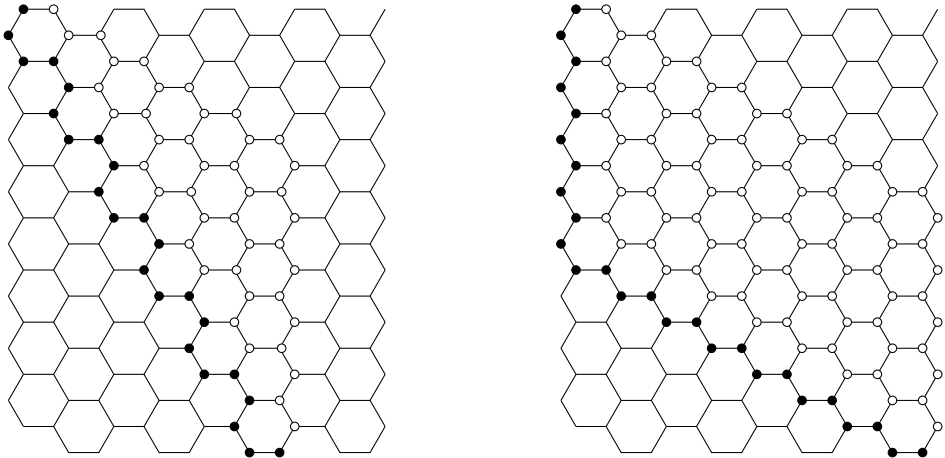
\begin{figure}[htb!]
\centering
\begin{tikzpicture}[rotate=0,scale=0.95]

\pgfmathsetmacro\cs{cos(30)}

\foreach \y in {0,1,2,3,4,6,7,8,5}{

\ifnum\y<8{
\draw (-15*\hx,2*\y*\hx*\cs+2*\hx*\cs) -- ++(-120:\hx);
\draw (18*\hx-21*\hx,2*\y*\hx*\cs) -- ++(120:\hx);
}\fi


\foreach \x in {0,...,5}{

\ifnum\x<4{
\draw (3*\x*\hx-15*\hx,2*\y*\hx*\cs) -- ++(0:\hx) -- ++(-60:\hx) -- ++(0:\hx) -- ++(60:\hx);

\ifnum\y<8{
\draw (3*\x*\hx-15*\hx,2*\y*\hx*\cs) -- ++(120:\hx);
\draw (3*\x*\hx+\hx-15*\hx,2*\y*\hx*\cs) -- ++(60:\hx);
}\fi
}\fi

}

}

\foreach \y in {0,1,2,3,4,6,7,8,5}{

\foreach \x in {0,...,5}{

\pgfmathsetmacro\tmp{-\y+5-1}
\pgfmathsetmacro\tmpp{\y-\tmp}

\pgfmathsetmacro\tmpo{\y+1}

\ifnum\y>2{\ifnum\x<\tmpo{\ifnum\x<6{\ifnum\y<6{

\ifnum\y=5{
\filldraw[fill=black,draw=black] (3*\x*\hx/2 -3*\y*\hx,2*8*\hx*\cs-3*\x*\hx*\cs) circle (\vs);
\filldraw[fill=black,draw=black] (3*\x*\hx/2 -3*\y*\hx-\hx/2,2*8*\hx*\cs-3*\x*\hx*\cs-\hx*\cs) circle (\vs);
\filldraw[fill=black,draw=black] (3*\x*\hx/2 -3*\y*\hx,2*8*\hx*\cs-3*\x*\hx*\cs-2*\hx*\cs) circle (\vs);
\filldraw[fill=black,draw=black] (3*\x*\hx/2 -3*\y*\hx+\hx,2*8*\hx*\cs-3*\x*\hx*\cs-2*\hx*\cs) circle (\vs);
}
\else{
\ifnum\y=4{
\filldraw[fill=white,draw=black] (3*\x*\hx/2 -3*\y*\hx-\hx*\cs/2,2*8*\hx*\cs-3*\x*\hx*\cs-\hx*\cs) circle (\vs);
\filldraw[fill=white,draw=black] (3*\x*\hx/2 -3*\y*\hx,2*8*\hx*\cs-3*\x*\hx*\cs-2*\hx*\cs) circle (\vs);
\ifnum\x<\y{
\filldraw[fill=white,draw=black] (3*\x*\hx/2 -3*\y*\hx+\hx,2*8*\hx*\cs-3*\x*\hx*\cs-2*\hx*\cs) circle (\vs);
\filldraw[fill=white,draw=black] (3*\x*\hx/2 -3*\y*\hx+3*\hx/2,2*8*\hx*\cs-3*\x*\hx*\cs-3*\hx*\cs) circle (\vs);
}\fi
}
\else{
\ifnum\x>0{
\ifnum\x<3{
\filldraw[fill=white,draw=black] (3*\x*\hx/2 -3*\y*\hx-\hx*\cs/2,2*8*\hx*\cs-3*\x*\hx*\cs-\hx*\cs) circle (\vs);
\filldraw[fill=white,draw=black] (3*\x*\hx/2 -3*\y*\hx,2*8*\hx*\cs-3*\x*\hx*\cs-2*\hx*\cs) circle (\vs);
\ifnum\x<2{
\filldraw[fill=white,draw=black] (3*\x*\hx/2 -3*\y*\hx+\hx,2*8*\hx*\cs-3*\x*\hx*\cs-2*\hx*\cs) circle (\vs);
\filldraw[fill=white,draw=black] (3*\x*\hx/2 -3*\y*\hx+3*\hx/2,2*8*\hx*\cs-3*\x*\hx*\cs-3*\hx*\cs) circle (\vs);
}\fi}\fi}\fi
}\fi}\fi

\ifnum\x>\tmp{
\ifnum\x<\tmpp{
\filldraw[fill=white,draw=black] (3*\x*\hx/2+3*\hx/2 -3*\y*\hx,2*6*\hx*\cs+3*\hx*\cs-3*\x*\hx*\cs) circle (\vs);
\filldraw[fill=white,draw=black] (3*\x*\hx/2+2*\hx/2 -3*\y*\hx,2*6*\hx*\cs+\hx*\cs-3*\x*\hx*\cs+3*\hx*\cs) circle (\vs);
}\fi}\fi

}\fi}\fi}\fi}\fi

}

}

\begin{scope}[xshift=220pt]

\pgfmathsetmacro\cs{cos(30)}

\foreach \y in {0,1,2,3,4,6,7,8,5}{

\ifnum\y<8{
\draw (-15*\hx,2*\y*\hx*\cs+2*\hx*\cs) -- ++(-120:\hx);
\draw (18*\hx-21*\hx,2*\y*\hx*\cs) -- ++(120:\hx);
}\fi


\foreach \x in {0,...,7}{

\ifnum\x<4{
\draw (3*\x*\hx-15*\hx,2*\y*\hx*\cs) -- ++(0:\hx) -- ++(-60:\hx) -- ++(0:\hx) -- ++(60:\hx);

\ifnum\y<8{
\draw (3*\x*\hx-15*\hx,2*\y*\hx*\cs) -- ++(120:\hx);
\draw (3*\x*\hx+\hx-15*\hx,2*\y*\hx*\cs) -- ++(60:\hx);
}\fi
}\fi
}
}

\foreach \y in {0,1,2,3,4,6,7,8,5}{

\foreach \x in {0,...,7}{

\pgfmathsetmacro\tmp{-\y+5-1}
\pgfmathsetmacro\tmpp{\y-\tmp}

\pgfmathsetmacro\tmpo{\y+1}

\ifnum\y=5{
\filldraw[fill=black,draw=black] (3*\x*\hx/2 -4*\y*\hx+6*\hx,2*3*\hx*\cs-1*\x*\hx*\cs) circle (\vs);
\ifnum\x<7{
\filldraw[fill=black,draw=black] (3*\x*\hx/2 -4*\y*\hx+9*\hx/2+2*\hx,2*3*\hx*\cs-\x*\hx*\cs-\hx*\cs) circle (\vs);
}\fi
}\fi

\ifnum\x=0{
\ifnum\y>2{
\filldraw[fill=black,draw=black] (-15*\hx+\hx*\x,2*\y*\hx*\cs) circle (\vs);
\ifnum\y<8{
\filldraw[fill=black,draw=black] (-\hx/2-15*\hx+\hx*\x,2*\y*\hx*\cs+\cs*\hx) circle (\vs);
}\fi}\fi
}
\else{
\ifnum\y>3{
\filldraw[fill=white,draw=black] (-15*\hx+2*\hx*\x-\hx*\x/2,2*\y*\hx*\cs-2*\x*\hx*\cs+\hx*\cs*0+\x*\hx*\cs) circle (\vs);
\ifnum\y<9{
\filldraw[fill=white,draw=black] (-\hx/2-15*\hx+2*\hx*\x-\hx*\x/2,2*\y*\hx*\cs+\cs*\hx-2*\x*\hx*\cs+0*\hx*\cs+\x*\hx*\cs) circle (\vs);
}\fi}\fi
}\fi

}

\ifnum\y<5{
\filldraw[fill=white,draw=black] (-3*\hx-\hx/2,2*\y*\hx*\cs+\hx*\cs) circle (\vs);
\filldraw[fill=white,draw=black] (-3*\hx,2*\y*\hx*\cs) circle (\vs);
}\fi

}

\end{scope}
\end{tikzpicture}
\caption{Variables on black vertices for ``staircase'' initial conditions (left) and ``corner'' initial conditions (right).  In both cases, the variables at white vertices may be uniquely determined if the systems of hex equations are consistent.}
\label{fig:hexlattice1}
\end{figure}

There are obviously many other examples of initial conditions that might be obtained simply by taking different paths (``staircases'') in the hexagonal lattice, beyond the simple examples shown in Figures \ref{fig:hexlattice1}.  One could also consider disconnected initial conditions, as shown in the examples of Figure \ref{fig:hexlattice2}.  In contrast with Figure \ref{fig:hexlattice1}, both of the examples of Figure \ref{fig:hexlattice2} have initial conditions that are on broken paths in the hexagonal lattice.  Nonetheless, these two examples of initial conditions may also be used to define unique evolutions of consistent systems of hex equations in the hexagonal lattice.   

\begin{figure}[htb!]
\centering
\begin{tikzpicture}[scale=0.95]

\pgfmathsetmacro\cs{cos(30)}

\foreach \y in {0,...,8}{

\pgfmathsetmacro\tmp{8-\y+1}

\ifnum\y<8{
\draw (0,2*\y*\hx*\cs+2*\hx*\cs) -- ++(-120:\hx);
\draw (15*\hx,2*\y*\hx*\cs) -- ++(120:\hx);
}\fi

\foreach \x in {0,...,4}{
\draw (3*\x*\hx,2*\y*\hx*\cs) -- ++(0:\hx) -- ++(-60:\hx) -- ++(0:\hx) -- ++(60:\hx);

\ifnum\y<8{
\draw (3*\x*\hx,2*\y*\hx*\cs) -- ++(120:\hx);
\draw (3*\x*\hx+\hx,2*\y*\hx*\cs) -- ++(60:\hx);
}\fi

\ifnum\x=0{
\filldraw[fill=black,draw=black] (3*\x*\hx,2*\y*\hx*\cs) circle (\vs);
}
\else{\ifnum\x<\tmp{
\filldraw[fill=white,draw=black] (3*\x*\hx,2*\y*\hx*\cs) circle (\vs);
}\fi}\fi

\ifnum\y=0{
\filldraw[fill=black,draw=black] (3*\x*\hx+\hx,2*\y*\hx*\cs) circle (\vs);
\filldraw[fill=black,draw=black] (3*\x*\hx+\hx+\hx/2,2*\y*\hx*\cs-\cs*\hx) circle (\vs);
\filldraw[fill=black,draw=black] (3*\x*\hx+2*\hx+\hx/2,2*\y*\hx*\cs-\cs*\hx) circle (\vs);
}
\else{
\ifnum\x<\tmp{
\filldraw[fill=white,draw=black] (3*\x*\hx+\hx,2*\y*\hx*\cs) circle (\vs);
\filldraw[fill=white,draw=black] (3*\x*\hx+\hx+\hx/2,2*\y*\hx*\cs-\cs*\hx) circle (\vs);
\filldraw[fill=white,draw=black] (3*\x*\hx+2*\hx+\hx/2,2*\y*\hx*\cs-\cs*\hx) circle (\vs);
}\fi
}\fi
}

\ifnum\y<8{
\filldraw[fill=black,draw=black] (-\hx/2,2*\y*\hx*\cs+\cs*\hx) circle (\vs);
\ifnum\y<3{
\filldraw[fill=white,draw=black] (15*\hx,2*\y*\hx*\cs) circle (\vs);
}\fi}\fi

}


\begin{scope}[xshift=230pt]

\pgfmathsetmacro\cs{cos(30)}

\foreach \y in {0,...,8}{

\pgfmathsetmacro\tmpa{8-\y-2}

\ifnum\y<8{
\draw (0,2*\y*\hx*\cs+2*\hx*\cs) -- ++(-120:\hx);
\draw (15*\hx,2*\y*\hx*\cs) -- ++(120:\hx);
}\fi

\foreach \x in {0,...,4}{
\draw (3*\x*\hx,2*\y*\hx*\cs) -- ++(0:\hx) -- ++(-60:\hx) -- ++(0:\hx) -- ++(60:\hx);

\ifnum\y<8{
\draw (3*\x*\hx,2*\y*\hx*\cs) -- ++(120:\hx);
\draw (3*\x*\hx+\hx,2*\y*\hx*\cs) -- ++(60:\hx);
}\fi

\ifnum\x=0{
}
\else{
}\fi

\ifnum\y=0{
\filldraw[fill=black,draw=black] (3*\x*\hx,2*\y*\hx*\cs) circle (\vs);
\filldraw[fill=black,draw=black] (3*\x*\hx+\hx,2*\y*\hx*\cs) circle (\vs);
\filldraw[fill=black,draw=black] (3*\x*\hx+\hx+\hx/2,2*\y*\hx*\cs-\cs*\hx) circle (\vs);
\filldraw[fill=black,draw=black] (3*\x*\hx+2*\hx+\hx/2,2*\y*\hx*\cs-\cs*\hx) circle (\vs);
}
\else{
\ifnum\x<\tmpa{
\filldraw[fill=white,draw=black] (3*\x*\hx,2*\y*\hx*\cs) circle (\vs);
\filldraw[fill=white,draw=black] (3*\x*\hx+\hx,2*\y*\hx*\cs) circle (\vs);
\filldraw[fill=white,draw=black] (3*\x*\hx+\hx+\hx/2,2*\y*\hx*\cs-\cs*\hx) circle (\vs);
\filldraw[fill=white,draw=black] (3*\x*\hx+2*\hx+\hx/2,2*\y*\hx*\cs-\cs*\hx) circle (\vs);
}\fi}\fi
}

\ifnum\y<5{
\filldraw[fill=black,draw=black] (-\hx/2,2*\y*\hx*\cs+\cs*\hx) circle (\vs);
\ifnum\y=0{
\filldraw[fill=black,draw=black] (15*\hx,2*\y*\hx*\cs) circle (\vs);
}\fi
}\fi

}

\end{scope}
\end{tikzpicture}
\caption{Left: Initial conditions (black vertices) for evolutions starting from the leftmost column and going to the right, column-by-column. Right: Initial conditions for evolutions starting from the lowest row and going up, row-by-row. In both cases, the white vertices may be uniquely determined if the systems of hex equations are consistent.  
In contrast with Figure \ref{fig:hexlattice1} these are examples of disconnected initial conditions.}
\label{fig:hexlattice2}
\end{figure}
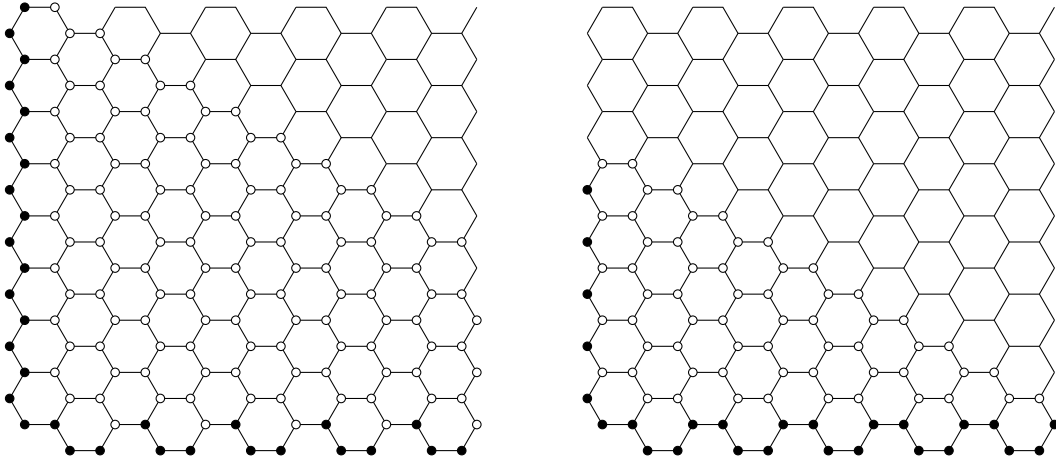

\section{Consistent systems of hex and quad equations on polytopes}\label{sec:polytope}

The properties of CAC and CAFCC allow to respectively embed systems of quad equations or face-centered quad equations consistently into higher-dimensional $\mathbb{Z}^n$-type lattices.  Analogues of the properties of CAC and CAFCC will be formulated here for hex equations in terms of their consistency on $n$-dimensional polytopes (or $n$-polytopes), for $n\geq3$, which have faces that are hexagons and quadrilaterals.  The consistent systems of hex equations will be assigned to hexagonal faces of such a polytope, and regular quad equations will be assigned to quadrilateral faces.  There are obviously a wide range of polytopes that could be used for defining consistent systems of equations.  Some explicit examples will be considered for consistent systems of equations on polytopes in 3 and 4 dimensions.

To define consistent systems of hex and quad equations on polytopes the graphical representations of the equations will be utilised.  These are given in Figure \ref{fig:hexagonunitcell} for the systems of type-A and type-C hex equations \eqref{Ahexdef} and \eqref{Chexdef}, and in Figure \ref{fig:5quads} for type-H ABS quad equations.  
In addition to these, a pair of type-Q quad equations will be denoted as $\quadQ{x_a}{x_b}{x_c}{x_d}{\alpha}{\beta}=0$ and $\quadQs{x_a}{x_b}{x_c}{x_d}{\alpha}{\beta}=0$ and these two equations will be associated to the respective diagrams shown in Figure \ref{fig:2quads}.

\begin{figure}[htb!]
\centering
\begin{tikzpicture}[scale=1.2]

\coordinate (A1) at (-1,-1);
\coordinate (A2) at (-1,1);
\coordinate (A3) at (1,1);
\coordinate (A4) at (1,-1);

\draw (A1)--(A2) node[midway,left]{$\beta$};
\draw (A2)--(A3) node[midway,above]{$\alpha$};
\draw (A3)--(A4) node[midway,right]{$\beta$};
\draw (A4)--(A1) node[midway,below]{$\alpha$}node[midway,below=15pt]{$\quadQ{x_a}{x_b}{x_c}{x_d}{\alpha}{\beta}=0$};

\fill (A1) circle (2pt) node[below left]{$x_c$};
\fill (A2) circle (2pt) node[above left]{$x_a$};
\fill (A3) circle (2pt) node[above right]{$x_b$};
\fill (A4) circle (2pt) node[below right]{$x_d$};

\begin{scope}[xshift=6cm]
\coordinate (A1) at (-1,-1);
\coordinate (A2) at (-1,1);
\coordinate (A3) at (1,1);
\coordinate (A4) at (1,-1);

\draw[-!-=.53] (A1)--(A2) node[midway,left=2pt]{$\beta$};
\draw[-!-=.53] (A2)--(A3) node[midway,above=2pt]{$\alpha$};
\draw[-!-=.53] (A3)--(A4) node[midway,right=2pt]{$\beta$};
\draw[-!-=.53] (A4)--(A1) node[midway,below=2pt]{$\alpha$}node[midway,below=15pt]{$\quadQs{x_a}{x_b}{x_c}{x_d}{\alpha}{\beta}=0$};

\fill (A1) circle (2pt) node[below left]{$x_c$};
\fill (A2) circle (2pt) node[above left]{$x_a$};
\fill (A3) circle (2pt) node[above right]{$x_b$};
\fill (A4) circle (2pt) node[below right]{$x_d$};
\end{scope}

\end{tikzpicture}
\caption{Two type-Q ABS quad equations.}
\label{fig:2quads}
\end{figure}
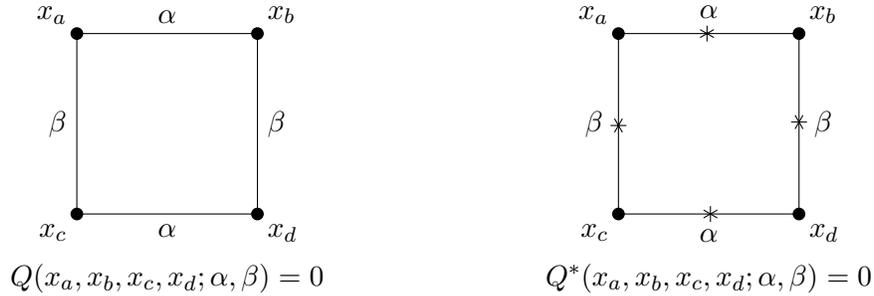

To begin with, for some polytope that contains only quadrilateral and hexagonal faces, variables are assigned to each of its vertices and parameters are assigned to each of its edges, such that opposite edges of every quadrilateral and hexagonal face in the polytope are associated to the same parameter.  
The simplest case to consider is when the same systems of type-A hex equations $\bm{A}=0$ are assigned to each hexagonal face, and the same type-Q ABS quad equations $Q=0$ are assigned to each quadrilateral face.  Because these equations satisfy the hexagonal symmetries \eqref{d6syms} and square symmetries \eqref{symmetriesquad} respectively, their orientation on faces of the polytope does not matter.  The pairs of quad and face-centered quad equations listed on the left of Table \ref{tab:CAFCCtoABS}, as well as the elliptic pair $Q4$ and $A4$, may respectively be used to collectively form consistent systems of equations on polytopes. 

An example is shown for the hexagonal prism of Figure \ref{fig:hexprismexample}.  This polyhedron has the four parameters $\al,\bt,\gm,\rho$ assigned to its edges, and twelve variables assigned to its vertices.  A type-A system of hex equations $\hexA{x_a}{x_b}{x_c}{x_d}{x_e}{x_f}{\al}{\bt}{\gm}$ is assigned to the bottom hexagon, a type-A system of hex equations $\hexA{y_a}{y_b}{y_c}{y_d}{y_e}{y_f}{\al}{\bt}{\gm}$ is assigned to the top hexagon, and there are six type-Q ABS quad equations $Q=0$ assigned to the six quadrilaterals.  An initial value problem may be posed by choosing four out of the six variables $x_a,x_b,x_c,x_d,x_e,x_f$ that are on consecutive vertices, plus one of the six variables $y_a,y_b,y_c,y_d,y_e,y_f$, as initial values.  One may check directly that systems of type-A hex equations and type-Q ABS equations (formed from the pairs listed in Table \ref{tab:CAFCCtoABS}) will give consistent solutions for the unknown variables.

\begin{figure}[htb!]
\centering
\begin{tikzpicture}[line join=bevel,z=-5.5,scale=2.6]

\pgfmathsetmacro\tr{sqrt(3)/2}
\pgfmathsetmacro\sn{sin(-73)}
\pgfmathsetmacro\cs{cos(-73)}

\coordinate (A0) at (1,{-\sn*1/2},{\cs*1/2});
\coordinate (A1) at (1/2,{\cs*\tr+\sn/2},{\sn*\tr-\cs*1/2});
\coordinate (A2) at (1/2,{\cs*\tr-\sn/2},{\sn*\tr+\cs*1/2});
\coordinate (A3) at (-1/2,{\cs*\tr+\sn/2},{\sn*\tr-\cs*1/2});
\coordinate (A4) at (-1/2,{\cs*\tr-\sn/2},{\sn*\tr+\cs*1/2});
\coordinate (A5) at (-1,{\sn*1/2},{-\cs*1/2});
\coordinate (A6) at (-1,{-\sn*1/2},{\cs*1/2});
\coordinate (A7) at (-1/2,{-\cs*\tr+\sn/2},{-\sn*\tr-\cs*1/2});
\coordinate (A8) at (-1/2,{-\cs*\tr-\sn/2},{-\sn*\tr+\cs*1/2});
\coordinate (A9) at (1/2,{-\cs*\tr+\sn/2},{-\sn*\tr-\cs*1/2});
\coordinate (A10) at (1/2,{-\cs*\tr-\sn/2},{-\sn*\tr+\cs*1/2});
\coordinate (A11) at (1,{\sn*1/2},{-\cs*1/2});

\draw (A6)--(A4);
\draw (A4)--(A2);
\draw (A2)--(A0);
\draw (A0)--(A10);
\draw (A10)--(A8);
\draw (A8)--(A6);

\draw[dashed] (A5)--(A3);
\draw[dashed] (A3)--(A1);
\draw[dashed] (A1)--(A11);
\draw (A11)--(A9) node[midway,below right]{$\gamma$};
\draw (A9)--(A7)node[midway,below]{$\alpha$};
\draw (A7)--(A5) node[midway,below left]{$\beta$};

\draw (A6)--(A5) node[midway,left]{$\rho$};
\draw[dashed] (A4)--(A3);
\draw[dashed] (A2)--(A1);
\draw (A0)--(A11);
\draw (A10)--(A9);
\draw (A8)--(A7);

\fill (A5) circle (0.9pt) node[left]{$\xa$};
\fill (A3) circle (0.9pt) node[below=1pt]{$\xb$};
\fill (A1) circle (0.9pt) node[below left]{$\xf$};
\fill (A11) circle (0.9pt) node[right]{$\xd$};
\fill (A9) circle (0.9pt) node[below right]{$\xc$};
\fill (A7) circle (0.9pt) node[below left]{$\xe$};

\fill (A6) circle (0.9pt) node[left]{$\ya$};
\fill (A4) circle (0.9pt) node[above left]{$\yb$};
\fill (A2) circle (0.9pt) node[above right]{$\yf$};
\fill (A0) circle (0.9pt) node[right]{$\yd$};
\fill (A10) circle (0.9pt) node[above]{$\yc$};
\fill (A8) circle (0.9pt) node[above right]{$\ye$};

\end{tikzpicture}
\caption{Variables and parameters on the hexagonal prism.  Opposite edges of each face are assigned the same parameter. On faces of this hexagonal prism there are two systems of type-A hex equations and six type-Q quad equations.}
\label{fig:hexprismexample}
\end{figure}
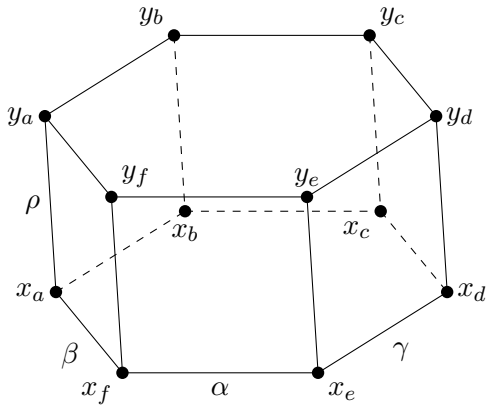

More complicated (and arguably more interesting) combinations of equations arise when considering systems of type-C hex equations $\bm{C}=0$, because orientations of both hex and type-H quad equations must be taken into account.  The assignment of variables and parameters to a polytope is the same as described above, but to have consistent systems of equations requires assigning to quadrilaterals some particular combinations of two type-Q quad equations $Q=0$ and $Q^\ast=0$ represented by the diagrams of Figure \ref{fig:2quads}, and a type-H quad equation $H=0$ and its trapezoidal version $\quadHun=0$ 
represented by the diagrams of Figure \ref{fig:5quads}.  The combination of equations is chosen so that two neighbouring faces of the polytope share the same type of edge according to the diagrams of Figures \ref{fig:5quads}, \ref{fig:hexagonunitcell}, and \ref{fig:2quads}.  For a given system of type-C hex equations, there can be more than one way to arrange the quad equations to form a consistent system on a polytope.  The combinations of type-C hex equations and quad equations that may be used to collectively form consistent systems on polytopes are listed in Table \ref{tab:typeCpolytope}.

\begin{table}[htb!]
\centering

\begin{tabular}{c|c|c|c|c}

\multicolumn{2}{c|}{$\bm{C}$} & \multicolumn{3}{c}{Quad equations} 

\\

\hline

\\[-0.4cm]

 $C$ & $\overline{C}$ & $H$ & $Q$ & $Q^\ast$  
 
 \\
 
 \hline
 
 \\[-0.4cm]
 
$C3_{(1/2;\,1/2;\,0)}$  &  
$C3_{(1/2;\,0;\,1/2)}$  &
$H3_{(1;\,1)}$ &
$Q3_{(1)}$ & $Q3_{(0)}$

\\[0.11cm]
 
$C3_{(1/2;\,0;\,1/2)}$  &  
$C3_{(1/2;\,1/2;\,0)}$  &
$H3_{(1;\,1)}\; (x_a\leftrightarrow x_b,x_c\leftrightarrow x_d)$ & 
$Q3_{(0)}$ & $Q3_{(1)}$

\\[0.11cm]
 
$C3_{(1;\,0;\,0)}$  &  
\textrm{same as} $C$ &
$H3_{(1;\,0)}$ &
$Q3_{(0)}$ & \textrm{same as} $Q$

\\[0.11cm]
 
$C3_{(0;\,0;\,0)}$  &  
\textrm{same as} $C$ &
$H3_{(0;\,0)}$ &
$Q3_{(0)}$ & \textrm{same as} $Q$

\\[0.11cm]
 
$C2_{(1;\,1;\,0)}$  &  
$C2_{(1;\,0;\,1)}$  &
$H2_{(1)}$ &
$Q2$ & $Q1_{(1)}$

\\[0.11cm]
 
$C2_{(1;\,0;\,1)}$  &  
$C2_{(1;\,1;\,0)}$  &
$H2_{(1)}\; (x_a\leftrightarrow x_b,x_c\leftrightarrow x_d)$ &
$Q1_{(1)}$ & $Q2$

\\[0.11cm]
 
$C2_{(1;\,0;\,0)}$ &
\textrm{same as} $C$ &
$H2_{(0)}$ &
$Q1_{(1)}$ & \textrm{same as} $Q$

\\[0.11cm]
 
$C1_{(1)}$  &  
$C2_{(0;\,0;\,0)}$  &
$H1_{(1)}$ &
$Q1_{(1)}$ & $Q1_{(0)}$

\\[0.11cm]
 
$C2_{(0;\,0;\,0)}$ &
$C1_{(1)}$  &
$H1_{(1)}\; (x_a\leftrightarrow x_b,x_c\leftrightarrow x_d)$ &
$Q1_{(0)}$ & $Q1_{(1)}$

\\[0.11cm]
 
$C1_{(0)}$ &
\textrm{same as} $C$ &
$H1_{(0)}$ &
$Q1_{(0)}$ & \textrm{same as} $Q$

\\[0.05cm]

\hline 
\end{tabular}
\caption{Combinations of type-C hex equations and quad equations that can form consistent systems of equations on polytopes with quadrilateral and hexagonal faces.  ``Same as $C$'' indicates that $\overline{C}=C$, and ``same as $Q$'' indicates that $Q^\ast=Q$.}
\label{tab:typeCpolytope}
\end{table}

Following the above guidelines, the combinations of equations listed in Table \ref{tab:typeCpolytope} will be used to form consistent systems of equations on three examples of polyhedra given by the hexagonal prism, the elongated dodecahedron, and the truncated octahedron.  Each of these polyhedra are space-filling and tessellate three-dimensional space by translations, where they respectively form honeycombs known as the hexagonal prismatic honeycomb, the elongated dodecahedral honeycomb, and the bitruncated cubic honeycomb. The consistency of systems of equations on these polyhedra implies that the equations may be consistently embedded into these three-dimensional honeycombs.  In particular, the hexagonal prismatic honeycomb may be regarded as being composed of infinite stackings of hexagonal prisms on each of the faces of the hexagonal lattice, and this provides a natural extension of the systems of hex equations into a third dimension orthogonal to the hexagonal lattice.  This may be regarded as the analogue of extending systems of regular quad equations in the square lattice $\mathbb{Z}^2$ into the cubic lattice $\mathbb{Z}^3$ using the consistency of the equations on the cube.

As a four-dimensional example, consistent combinations of type-A hex equations and type-Q quad equations will be considered for a 4-polytope known as the 6-6-duoprism. The 6-6-duoprism was chosen since it is one of the simpler examples of a 4-polytope that contains only quadrilateral and hexagonal faces, and it also arises as one of the facets of the omnitruncated 5-simplex, which also has only quadrilateral and hexagonal faces.  The omnitruncated 5-simplex is the permutahedron $P_5$, where in general a permutahedron $P_{n-1}\subseteq \mathbb{R}^n$ may be defined as the convex hull of all vectors that are obtained by permuting the coordinates of the vector $(1,\ldots,n)$ \cite{ZieglerPolytopes,ConwaySloaneSpherePackings}.  Thus, $P_{n-1}$ has $n!$ vertices, and the hexagon and truncated octahedron are $P_{2}$ and $P_{3}$ respectively.  $P_{n-1}$ lies entirely in an $(n-1)$-dimensional hyperplane and tessellates this hyperplane through translations.  Thus, consistency on permutahedra (and on facets of permutahedra, such as the 6-6-duoprism) potentially would lead to interesting systems of consistent equations that are embedded in higher-dimensional space.

\subsection{Consistency-around-a-hexagonal-prism (CAHP)}

The hexagonal prism is a polyhedron that has two hexagonal faces and six quadrilateral faces, as shown in the diagram of Figure \ref{fig:CAHP1}. The hexagonal prism is assigned an overdetermined system of eighteen equations (two systems of hex equations and six quad equations) for seven unknown variables.  If the overdetermined system of equations on the hexagonal prism has a consistent solution, then the system of equations will be said to satisfy consistency-around-a-hexagonal prism (CAHP).  

\begin{figure}[htb!]
\centering
\begin{tikzpicture}[line join=bevel,z=-5.5,scale=3.2]

\pgfmathsetmacro\tr{sqrt(3)/2}
\pgfmathsetmacro\sn{sin(-73)}
\pgfmathsetmacro\cs{cos(-73)}

\coordinate (A0) at (1,{-\sn*1/2},{\cs*1/2});
\coordinate (A1) at (1/2,{\cs*\tr+\sn/2},{\sn*\tr-\cs*1/2});
\coordinate (A2) at (1/2,{\cs*\tr-\sn/2},{\sn*\tr+\cs*1/2});
\coordinate (A3) at (-1/2,{\cs*\tr+\sn/2},{\sn*\tr-\cs*1/2});
\coordinate (A4) at (-1/2,{\cs*\tr-\sn/2},{\sn*\tr+\cs*1/2});
\coordinate (A5) at (-1,{\sn*1/2},{-\cs*1/2});
\coordinate (A6) at (-1,{-\sn*1/2},{\cs*1/2});
\coordinate (A7) at (-1/2,{-\cs*\tr+\sn/2},{-\sn*\tr-\cs*1/2});
\coordinate (A8) at (-1/2,{-\cs*\tr-\sn/2},{-\sn*\tr+\cs*1/2});
\coordinate (A9) at (1/2,{-\cs*\tr+\sn/2},{-\sn*\tr-\cs*1/2});
\coordinate (A10) at (1/2,{-\cs*\tr-\sn/2},{-\sn*\tr+\cs*1/2});
\coordinate (A11) at (1,{\sn*1/2},{-\cs*1/2});

\draw[double,-<-=.55] (A6)--(A4);
\draw (A4)--(A2);
\draw (A2)--(A0);
\draw[double,->-=.6] (A0)--(A10);
\draw[-!-=.53] (A10)--(A8);
\draw[-!-=.53] (A8)--(A6);

\draw[dashed,double,-<-=.55] (A5)--(A3);
\draw[dashed] (A3)--(A1);
\draw[dashed] (A1)--(A11);
\draw[double,->-=.6] (A11)--(A9) node[midway,below right=1pt]{$\gamma$};
\draw[-!-=.53] (A9)--(A7)node[midway,below=1pt]{$\alpha$};
\draw[-!-=.53] (A7)--(A5) node[midway,below left=1pt]{$\beta$};

\draw[-!-=.53] (A6)--(A5) node[midway,left=1pt]{$\rho$};
\draw[dashed] (A4)--(A3);
\draw[dashed] (A2)--(A1);
\draw (A0)--(A11);
\draw[-!-=.53] (A10)--(A9);
\draw[-!-=.53] (A8)--(A7);

\filldraw[fill=black,draw=black] (A5) circle (0.7pt) node[left]{$\xa$};
\filldraw[fill=black,draw=black] (A3) circle (0.7pt) node[below=1pt]{$\xb$};
\filldraw[fill=white,draw=black] (A1) circle (0.7pt) node[below left]{$\xf$};
\filldraw[fill=white,draw=black] (A11) circle (0.7pt) node[right]{$\xd$};
\filldraw[fill=black,draw=black] (A9) circle (0.7pt) node[below right]{$\xc$};
\filldraw[fill=black,draw=black] (A7) circle (0.7pt) node[below left]{$\xe$};

\filldraw[fill=black,draw=black] (A6) circle (0.7pt) node[left]{$\ya$};
\filldraw[fill=white,draw=black] (A4) circle (0.7pt) node[above left]{$\yb$};
\filldraw[fill=white,draw=black] (A2) circle (0.7pt) node[above right]{$\yf$};
\filldraw[fill=white,draw=black] (A0) circle (0.7pt) node[right]{$\yd$};
\filldraw[fill=white,draw=black] (A10) circle (0.7pt) node[above]{$\yc$};
\filldraw[fill=white,draw=black] (A8) circle (0.7pt) node[above right]{$\ye$};

\end{tikzpicture}
\caption{CAHP. Initial variables are at black vertices and there are four different parameters $\al,\bt,\gm,\rho$ that are assigned to edges, where opposite edges of a face have the same parameter.}
\label{fig:CAHP1}
\end{figure}
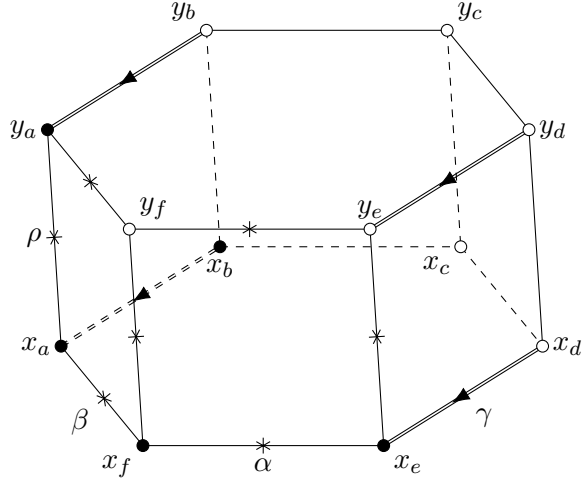

The initial variables are chosen as
\begin{equation}
\xa,\xb,\xc,\xe,\ya,
\end{equation}
as indicated by black vertices in Figure \ref{fig:CAHP1}.  For these initial conditions, the property of CAHP for the hexagonal prism of Figure \ref{fig:CAHP1} can be checked with the following steps.
\begin{enumerate}[(i)]
\item
The system of hex equations
\begin{equation}
\hexC{\xa}{\xb}{\xf}{\xd}{\xc}{\xe}{\alpha}{\beta}{\gamma}
\end{equation}
is used to uniquely solve for the two variables $\xf$ and $\xd$, and the quad equations
\begin{equation}
\begin{split}
\quadHu{\xa}{\ya}{\xb}{\yb}{\rho}{\gm}=0, \\
\quadQs{\ya}{\ye}{\xa}{\xe}{\bt}{\rho}=0,
\end{split}
\end{equation}
are used to uniquely solve for the two variables $\yb$ and $\ye$, respectively.
\item
The quad equations
\begin{equation}
\begin{split}
\quadQ{\yb}{\yf}{\xb}{\xf}{\al}{\rho}=0, \\
\quadQs{\ye}{\yc}{\xe}{\xc}{\al}{\rho}=0,
\end{split}
\end{equation}
are used to uniquely solve for the two variables $\yf$ and $\yc$, respectively.
\item
The system of hex equations 
\begin{equation}
\hexC{\ya}{\yb}{\yf}{\yd}{\yc}{\ye}{\alpha}{\beta}{\gamma}
\end{equation}
and the quad equations
\begin{equation}
\begin{split}
\quadQ{\yf}{\yd}{\xf}{\xd}{\bt}{\rho}=0, \\
\quadHu{\xc}{\yc}{\xd}{\yd}{\rho}{\gm}=0,
\end{split}
\end{equation}
must agree for the final variable $\yd$.
\end{enumerate}

Each of the combinations of equations listed in Table \ref{tab:typeCpolytope} satisfy the above notion of CAHP.

Figure \ref{fig:CAHP2} shows a different arrangement of equations for the hexagonal prism, where each of the quad equations are different, and the upper system of hex equations is rotated by $\pi$ (this rotation only changes the system of type-C hex equations \eqref{Chexdef} if $C$ and $\overline{C}$ are different).

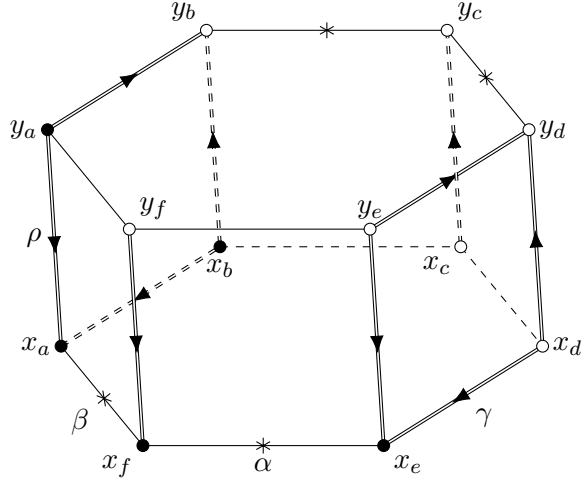
\begin{figure}[htb!]
\centering
\begin{tikzpicture}[line join=bevel,z=-5.5,scale=3.2]

\pgfmathsetmacro\tr{sqrt(3)/2}
\pgfmathsetmacro\sn{sin(-73)}
\pgfmathsetmacro\cs{cos(-73)}

\coordinate (A0) at (1,{-\sn*1/2},{\cs*1/2});
\coordinate (A1) at (1/2,{\cs*\tr+\sn/2},{\sn*\tr-\cs*1/2});
\coordinate (A2) at (1/2,{\cs*\tr-\sn/2},{\sn*\tr+\cs*1/2});
\coordinate (A3) at (-1/2,{\cs*\tr+\sn/2},{\sn*\tr-\cs*1/2});
\coordinate (A4) at (-1/2,{\cs*\tr-\sn/2},{\sn*\tr+\cs*1/2});
\coordinate (A5) at (-1,{\sn*1/2},{-\cs*1/2});
\coordinate (A6) at (-1,{-\sn*1/2},{\cs*1/2});
\coordinate (A7) at (-1/2,{-\cs*\tr+\sn/2},{-\sn*\tr-\cs*1/2});
\coordinate (A8) at (-1/2,{-\cs*\tr-\sn/2},{-\sn*\tr+\cs*1/2});
\coordinate (A9) at (1/2,{-\cs*\tr+\sn/2},{-\sn*\tr-\cs*1/2});
\coordinate (A10) at (1/2,{-\cs*\tr-\sn/2},{-\sn*\tr+\cs*1/2});
\coordinate (A11) at (1,{\sn*1/2},{-\cs*1/2});

\draw[double,->-=.6] (A6)--(A4);
\draw[-!-=.53] (A4)--(A2);
\draw[-!-=.53] (A2)--(A0);
\draw[double,-<-=.55] (A0)--(A10);
\draw (A10)--(A8);
\draw (A8)--(A6);

\draw[dashed,double,-<-=.55] (A5)--(A3);
\draw[dashed] (A3)--(A1);
\draw[dashed] (A1)--(A11);
\draw[double,->-=.6] (A11)--(A9) node[midway,below right=1pt]{$\gamma$};
\draw[-!-=.53] (A9)--(A7)node[midway,below=1pt]{$\alpha$};
\draw[-!-=.53] (A7)--(A5) node[midway,below left=1pt]{$\beta$};

\draw[double,->-=.6] (A6)--(A5) node[midway,left=1pt]{$\rho$};
\draw[dashed,double,-<-=.55] (A4)--(A3);
\draw[dashed,double,-<-=.55] (A2)--(A1);
\draw[double,-<-=.55] (A0)--(A11);
\draw[double,->-=.6] (A10)--(A9);
\draw[double,->-=.6] (A8)--(A7);

\filldraw[fill=black,draw=black] (A5) circle (0.7pt) node[left]{$\xa$};
\filldraw[fill=black,draw=black] (A3) circle (0.7pt) node[below=1pt]{$\xb$};
\filldraw[fill=white,draw=black] (A1) circle (0.7pt) node[below left]{$\xf$};
\filldraw[fill=white,draw=black] (A11) circle (0.7pt) node[right]{$\xd$};
\filldraw[fill=black,draw=black] (A9) circle (0.7pt) node[below right]{$\xc$};
\filldraw[fill=black,draw=black] (A7) circle (0.7pt) node[below left]{$\xe$};

\filldraw[fill=black,draw=black] (A6) circle (0.7pt) node[left]{$\ya$};
\filldraw[fill=white,draw=black] (A4) circle (0.7pt) node[above left]{$\yb$};
\filldraw[fill=white,draw=black] (A2) circle (0.7pt) node[above right]{$\yf$};
\filldraw[fill=white,draw=black] (A0) circle (0.7pt) node[right]{$\yd$};
\filldraw[fill=white,draw=black] (A10) circle (0.7pt) node[above]{$\yc$};
\filldraw[fill=white,draw=black] (A8) circle (0.7pt) node[above right]{$\ye$};

\end{tikzpicture}

\caption{CAHP for equations that are different from Figure \ref{fig:CAHP1}.}
\label{fig:CAHP2}
\end{figure}


The five initial variables are again chosen to be
\begin{equation}
\xa,\xb,\xc,\xe,\ya.
\end{equation}
as indicated by black vertices in Figure \ref{fig:CAHP2}.  For these initial conditions, the property of CAHP for the hexagonal prism of Figure \ref{fig:CAHP2} can be checked with the following steps.
\begin{enumerate}[(i)]
\item
The system of hex equations 
\begin{equation}
\hexC{\xa}{\xb}{\xf}{\xd}{\xc}{\xe}{\alpha}{\beta}{\gamma}
\end{equation}
is used to uniquely solve for the two variables $\xf$ and $\xd$, and the quad equations
\begin{equation}
\begin{split}
\quadH{\yb}{\ya}{\xb}{\xa}{\gm}{\rho}=0, \\
\quadHu{\xe}{\xa}{\ye}{\ya}{\bt}{\rho}=0,
\end{split}
\end{equation}
are used to uniquely solve for the two variables $\yb$ and $\ye$, respectively.
\item
The quad equations
\begin{equation}
\begin{split}
\quadHu{\yb}{\yf}{\xb}{\xf}{\al}{\rho}=0, \\
\quadHu{\xc}{\xe}{\yc}{\ye}{\al}{\rho}=0,
\end{split}
\end{equation}
are used to uniquely solve for the two variables $\yf$ and $\yc$, respectively.
\item
The system of hex equations 
\begin{equation}
\hexC{\yd}{\yc}{\ye}{\ya}{\yb}{\yf}{\alpha}{\beta}{\gamma}
\end{equation}
and the quad equations
\begin{equation}
\begin{split}
\quadHu{\yf}{\yd}{\xf}{\xd}{\bt}{\rho}=0, \\
\quadH{\yd}{\yc}{\xd}{\xc}{\gm}{\rho}=0,
\end{split}
\end{equation}
must agree for the final variable $\yd$.
\end{enumerate}

Each of the combinations of equations given in Table \ref{tab:typeCpolytope} satisfy the above notion of CAHP.

\subsection{Consistency-around-an-elongated-dodecahedron (CAED)}

The elongated dodecahedron is a polyhedron that has four hexagonal faces and eight quadrilateral faces, as shown in the diagram of Figure \ref{fig:CAED1}. The elongated dodecahedron is assigned an overdetermined system of thirty-two equations (four systems of hex equations and eight quad equations) for twelve unknown variables.  If the overdetermined system of equations on the elongated dodecahedron has a consistent solution, then the system of equations will be said to satisfy consistency-around-an-elongated-dodecahedron (CAED).

\begin{figure}[htb!]
\centering
\begin{tikzpicture}[line join=bevel,z=-5.5,scale=3.2]

\pgfmathsetmacro\tr{sqrt(3)/2}
\pgfmathsetmacro\sn{sin(13)}
\pgfmathsetmacro\cs{cos(13)}

\coordinate (A1) at ({-\tr},{-\cs*\tr+\sn/2},{-\sn*\tr-\cs/2});
\coordinate (A2) at ({-\tr},{-\cs*\tr-\sn/2},{-\sn*\tr+\cs/2});
\coordinate (A3) at ({-\tr},{\sn},{-\cs});
\coordinate (A4) at ({-\tr},{-\sn},{\cs});
\coordinate (A5) at ({-\tr},{\cs*\tr+\sn/2},{\sn*\tr-\cs/2});
\coordinate (A6) at ({-\tr},{\cs*\tr-\sn/2},{\sn*\tr+\cs/2});
\coordinate (A7) at (0,{-\cs*\tr+\sn},{-\sn*\tr-\cs});
\coordinate (A8) at (0,{-\cs*\tr-\sn},{-\sn*\tr+\cs});
\coordinate (A9) at (0,{3*\sn/2},{-3*\cs/2});
\coordinate (A10) at (0,{-3*\sn/2},{3*\cs/2});
\coordinate (A11) at (0,{\cs*\tr+\sn},{\sn*\tr-\cs});
\coordinate (A12) at (0,{\cs*\tr-\sn},{\sn*\tr+\cs});
\coordinate (A13) at ({\tr},{-\cs*\tr+\sn/2},{-\sn*\tr-\cs/2});
\coordinate (A14) at ({\tr},{-\cs*\tr-\sn/2},{-\sn*\tr+\cs/2});
\coordinate (A15) at ({\tr},{\sn},{-\cs});
\coordinate (A16) at ({\tr},{-\sn},{\cs});
\coordinate (A17) at ({\tr},{\cs*\tr+\sn/2},{\sn*\tr-\cs/2});
\coordinate (A18) at ({\tr},{\cs*\tr-\sn/2},{\sn*\tr+\cs/2});

\draw[double,-<-=.55] (A6)--(A5);
\draw (A5)--(A11);
\draw (A11)--(A17);
\draw[double,->-=.6] (A17)--(A18);
\draw[-!-=.53] (A18)--(A12);
\draw[-!-=.53] (A12)--(A6);

\draw[dashed,double,-<-=.55] (A2)--(A1);
\draw[dashed] (A1)--(A7);
\draw[dashed] (A7)--(A13);
\draw[double,->-=.6] (A13)--(A14) node[midway,right=1pt]{$\gamma$};
\draw[-!-=.53] (A14)--(A8) node[midway,below right=1pt]{$\alpha$};
\draw[-!-=.53] (A8)--(A2) node[midway,below left=1pt]{$\beta$};

\draw[-!-=.53] (A6)--(A4) node[midway,left=1pt]{$\rhob$}; 
\draw[-!-=.53] (A4)--(A2) node[midway,left=1pt]{$\rho$};
\draw[dashed] (A5)--(A3);
\draw[dashed] (A3)--(A1);
\draw[dashed] (A11)--(A9);
\draw[dashed] (A9)--(A7);
\draw (A17)--(A15);
\draw (A15)--(A13);
\draw[-!-=.53] (A18)--(A16);
\draw[-!-=.53] (A16)--(A14);
\draw[-!-=.53] (A12)--(A10);
\draw[-!-=.53] (A10)--(A8);

\draw[-!-=.53] (A4)--(A10);
\draw[-!-=.53] (A10)--(A16);
\draw[dashed] (A15)--(A9);
\draw[dashed] (A9)--(A3);

\filldraw[fill=black,draw=black] (A2) circle (0.7pt) node[left]{$\xa$};
\filldraw[fill=black,draw=black] (A1) circle (0.7pt) node[left]{$\xb$};
\filldraw[fill=black,draw=black] (A7) circle (0.7pt) node[above left]{$\xf$};
\filldraw[fill=white,draw=black] (A13) circle (0.7pt) node[right]{$\xd$};
\filldraw[fill=white,draw=black] (A14) circle (0.7pt) node[below right]{$\xc$};
\filldraw[fill=black,draw=black] (A8) circle (0.7pt) node[below left]{$\xe$};

\filldraw[fill=white,draw=black] (A6) circle (0.7pt) node[left]{$\ya$};
\filldraw[fill=white,draw=black] (A5) circle (0.7pt) node[above left]{$\yb$};
\filldraw[fill=white,draw=black] (A11) circle (0.7pt) node[above right]{$\yf$};
\filldraw[fill=white,draw=black] (A17) circle (0.7pt) node[right]{$\yd$};
\filldraw[fill=white,draw=black] (A18) circle (0.7pt) node[right]{$\yc$};
\filldraw[fill=white,draw=black] (A12) circle (0.7pt) node[below right]{$\ye$};

\filldraw[fill=black,draw=black] (A4) circle (0.7pt) node[left]{$\za$};
\filldraw[fill=black,draw=black] (A3) circle (0.7pt) node[left]{$\zb$};
\filldraw[fill=white,draw=black] (A9) circle (0.7pt) node[above left]{$\zf$};
\filldraw[fill=white,draw=black] (A15) circle (0.7pt) node[right]{$\zd$};
\filldraw[fill=white,draw=black] (A16) circle (0.7pt) node[right]{$\zc$};
\filldraw[fill=white,draw=black] (A10) circle (0.7pt) node[below right]{$\ze$};

\end{tikzpicture}

\caption{CAED.  Initial variables are at black vertices and there are five different parameters $\al,\bt,\gm,\rho,\rhob$ that are assigned to edges, where opposite edges of a face have the same parameter.}
\label{fig:CAED1}
\end{figure}

The six initial variables are chosen as
\begin{equation}
\xa,\xb,\xf,\xe,\za,\zb.
\end{equation}
as indicated by black vertices in Figure \ref{fig:CAED1}.  For these initial conditions, the property of CAED for the elongated dodecahedron of Figure \ref{fig:CAED1} can be checked with the following steps.
\begin{enumerate}[(i)]
\item
The system of hex equations 
\begin{equation}
\hexC{\xa}{\xb}{\xf}{\xd}{\xc}{\xe}{\alpha}{\beta}{\gamma}
\end{equation}
is used to uniquely solve for the two variables $\xd$ and $\xc$, the system of hex equations
\begin{equation}
\hexC{\ya}{\yb}{\zb}{\xb}{\xa}{\za}{\rho}{\rhob}{\gm}
\end{equation}
is used to uniquely solve for the two variables $\ya$ and $\yb$, and the two quad equations
\begin{equation}
\begin{split}
\quadQ{\zb}{\zf}{\xb}{\xf}{\al}{\rhob}=0, \\
\quadQs{\za}{\ze}{\xa}{\xe}{\bt}{\rho}=0,
\end{split}
\end{equation}
are used to uniquely solve for the two variables $\zf$ and $\ze$, respectively.
\item
The quad equations
\begin{equation}
\begin{split}
\quadQ{\zf}{\zd}{\xf}{\xd}{\bt}{\rhob}=0, \\
\quadQs{\ze}{\zc}{\xe}{\xc}{\al}{\rho}=0,
\end{split}
\end{equation}
are used to uniquely solve for the two variables $\zd$ and $\zc$, and the quad equations
\begin{equation}
\begin{split}
\quadQ{\yb}{\yf}{\zb}{\zf}{\al}{\rho}=0, \\
\quadQs{\ya}{\ye}{\za}{\ze}{\bt}{\rhob}=0,
\end{split}
\end{equation}
are used to uniquely solve for the two variables $\yf$ and $\ye$, respectively.
\item
The two systems of hex equations 
\begin{equation}
\begin{split}
\hexC{\ya}{\yb}{\yf}{\yd}{\yc}{\ye}{\alpha}{\beta}{\gamma}, \\
\hexC{\yc}{\yd}{\zd}{\xd}{\xc}{\zc}{\rho}{\rhob}{\gm},
\end{split}
\end{equation}
and the two quad equations 
\begin{equation}
\begin{split}
\quadQ{\yf}{\yd}{\zf}{\zd}{\bt}{\rho}=0, \\
\quadQs{\ye}{\yc}{\ze}{\zc}{\al}{\rhob}=0,
\end{split}
\end{equation}
must agree for the final two variables $\yd$ and $\yc$.
\end{enumerate}

Each of the combinations of equations given in Table \ref{tab:typeCpolytope} satisfy the above notion of CAED.

Figure \ref{fig:CAED2} shows a different arrangement of equations on the elongated dodecahedron, where the quad equations are different, and the systems of hex equations are rotated.

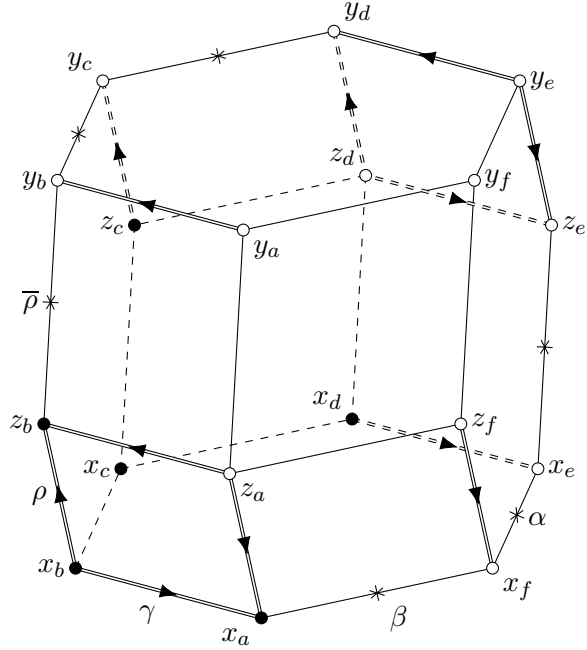
\begin{figure}[htb!]
\centering
\begin{tikzpicture}[line join=bevel,z=-5.5,scale=3.2]

\pgfmathsetmacro\tr{sqrt(3)/2}
\pgfmathsetmacro\sn{sin(13)}
\pgfmathsetmacro\cs{cos(13)}

\coordinate (A1) at ({-\tr},{-\cs*\tr+\sn/2},{-\sn*\tr-\cs/2});
\coordinate (A2) at ({-\tr},{-\cs*\tr-\sn/2},{-\sn*\tr+\cs/2});
\coordinate (A3) at ({-\tr},{\sn},{-\cs});
\coordinate (A4) at ({-\tr},{-\sn},{\cs});
\coordinate (A5) at ({-\tr},{\cs*\tr+\sn/2},{\sn*\tr-\cs/2});
\coordinate (A6) at ({-\tr},{\cs*\tr-\sn/2},{\sn*\tr+\cs/2});
\coordinate (A7) at (0,{-\cs*\tr+\sn},{-\sn*\tr-\cs});
\coordinate (A8) at (0,{-\cs*\tr-\sn},{-\sn*\tr+\cs});
\coordinate (A9) at (0,{3*\sn/2},{-3*\cs/2});
\coordinate (A10) at (0,{-3*\sn/2},{3*\cs/2});
\coordinate (A11) at (0,{\cs*\tr+\sn},{\sn*\tr-\cs});
\coordinate (A12) at (0,{\cs*\tr-\sn},{\sn*\tr+\cs});
\coordinate (A13) at ({\tr},{-\cs*\tr+\sn/2},{-\sn*\tr-\cs/2});
\coordinate (A14) at ({\tr},{-\cs*\tr-\sn/2},{-\sn*\tr+\cs/2});
\coordinate (A15) at ({\tr},{\sn},{-\cs});
\coordinate (A16) at ({\tr},{-\sn},{\cs});
\coordinate (A17) at ({\tr},{\cs*\tr+\sn/2},{\sn*\tr-\cs/2});
\coordinate (A18) at ({\tr},{\cs*\tr-\sn/2},{\sn*\tr+\cs/2});

\draw[double,->-=.6] (A12)--(A6);
\draw[-!-=.53] (A6)--(A5);
\draw[-!-=.53] (A5)--(A11);
\draw[double,-<-=.55] (A11)--(A17);
\draw (A17)--(A18);
\draw (A18)--(A12);

\draw[double,-<-=.55] (A8)--(A2) node[midway,below left=1pt]{$\gamma$};
\draw[dashed] (A2)--(A1);
\draw[dashed] (A1)--(A7);
\draw[dashed,double,->-=.6] (A7)--(A13);
\draw[-!-=.53] (A13)--(A14) node[midway,right=1pt]{$\alpha$};
\draw[-!-=.53] (A14)--(A8) node[midway,below right=1pt]{$\beta$};

\draw[-!-=.53] (A6)--(A4) node[midway,left=1pt]{$\rhob$}; 
\draw[double,-<-=.55] (A4)--(A2) node[midway,left=1pt]{$\rho$};
\draw[dashed,double,-<-=.55] (A5)--(A3);
\draw[dashed] (A3)--(A1);
\draw[dashed,double,-<-=.55] (A11)--(A9);
\draw[dashed] (A9)--(A7);
\draw[double,->-=.6] (A17)--(A15);
\draw[-!-=.53] (A15)--(A13);
\draw (A18)--(A16);
\draw[double,->-=.6] (A16)--(A14);
\draw (A12)--(A10);
\draw[double,->-=.6] (A10)--(A8);

\draw[double,-<-=.55] (A4)--(A10);
\draw (A10)--(A16);
\draw[dashed,double,-<-=.55] (A15)--(A9);
\draw[dashed] (A9)--(A3);

\filldraw[fill=black,draw=black] (A2) circle (0.7pt) node[left]{$\xb$};
\filldraw[fill=black,draw=black] (A1) circle (0.7pt) node[left]{$\xf$};
\filldraw[fill=black,draw=black] (A7) circle (0.7pt) node[above left]{$\xd$};
\filldraw[fill=white,draw=black] (A13) circle (0.7pt) node[right]{$\xc$};
\filldraw[fill=white,draw=black] (A14) circle (0.7pt) node[below right]{$\xe$};
\filldraw[fill=black,draw=black] (A8) circle (0.7pt) node[below left]{$\xa$};

\filldraw[fill=white,draw=black] (A6) circle (0.7pt) node[left]{$\yb$};
\filldraw[fill=white,draw=black] (A5) circle (0.7pt) node[above left]{$\yf$};
\filldraw[fill=white,draw=black] (A11) circle (0.7pt) node[above right]{$\yd$};
\filldraw[fill=white,draw=black] (A17) circle (0.7pt) node[right]{$\yc$};
\filldraw[fill=white,draw=black] (A18) circle (0.7pt) node[right]{$\ye$};
\filldraw[fill=white,draw=black] (A12) circle (0.7pt) node[below right]{$\ya$};

\filldraw[fill=black,draw=black] (A4) circle (0.7pt) node[left]{$\zb$};
\filldraw[fill=black,draw=black] (A3) circle (0.7pt) node[left]{$\zf$};
\filldraw[fill=white,draw=black] (A9) circle (0.7pt) node[above left]{$\zd$};
\filldraw[fill=white,draw=black] (A15) circle (0.7pt) node[right]{$\zc$};
\filldraw[fill=white,draw=black] (A16) circle (0.7pt) node[right]{$\ze$};
\filldraw[fill=white,draw=black] (A10) circle (0.7pt) node[below right]{$\za$};

\end{tikzpicture}

\caption{CAED for equations that are different from Figure \ref{fig:CAED1}.}
\label{fig:CAED2}
\end{figure}

Because of a different variable labelling, the six initial variables are now
\begin{equation}
\xa,\xb,\xf,\xd,\zb,\zf.
\end{equation}
as indicated by black vertices in Figure \ref{fig:CAED2}.  For these initial conditions, the property of CAED for the elongated dodecahedron of Figure \ref{fig:CAED2} can be checked with the following steps.
\begin{enumerate}[(i)]
\item
The system of hex equations 
\begin{equation}
\hexC{\xa}{\xb}{\xf}{\xd}{\xc}{\xe}{\alpha}{\beta}{\gamma}
\end{equation}
is used to uniquely solve for the two variables $\xc$ and $\xe$, the system of hex equations
\begin{equation}
\hexC{\yf}{\zf}{\xf}{\xb}{\zb}{\yb}{\rhob}{\al}{\rho}
\end{equation}
is used to uniquely solve for the two variables $\yb$ and $\yf$, and the two quad equations
\begin{equation}
\begin{split}
\quadQ{\zf}{\zd}{\xf}{\xd}{\bt}{\rhob}=0, \\
\quadH{\zb}{\za}{\xb}{\xa}{\gm}{\rho}=0,
\end{split}
\end{equation}
are used to uniquely solve for the two variables $\zd$ and $\za$, respectively.
\item
The quad equations
\begin{equation}
\begin{split}
\quadHu{\zc}{\xc}{\zd}{\xd}{\rhob}{\gm}=0, \\
\quadHu{\xe}{\xa}{\ze}{\za}{\bt}{\rho}=0,
\end{split}
\end{equation}
are used to uniquely solve for the two variables $\zc$ and $\ze$, and the quad equations
\begin{equation}
\begin{split}
\quadHu{\yf}{\yd}{\zf}{\zd}{\bt}{\rho}=0, \\
\quadHu{\zb}{\yb}{\za}{\ya}{\rhob}{\gm}=0,
\end{split}
\end{equation}
are used to uniquely solve for the two variables $\yd$ and $\ya$, respectively.
\item
The two systems of hex equations 
\begin{equation}
\begin{split}
\hexC{\yd}{\yc}{\ye}{\ya}{\yb}{\yf}{\alpha}{\beta}{\gamma}, \\ 
\hexC{\yc}{\zc}{\xc}{\xe}{\ze}{\ye}{\rhob}{\al}{\rho},
\end{split}
\end{equation}
and the two quad equations 
\begin{equation}
\begin{split}
\quadH{\yd}{\yc}{\zd}{\zc}{\gm}{\rho}=0, \\
\quadQ{\ya}{\ye}{\za}{\ze}{\bt}{\rhob}=0,
\end{split}
\end{equation}
must agree for the final two variables $\yc$ and $\ye$.
\end{enumerate}

Each of the combinations of equations given in Table \ref{tab:typeCpolytope} satisfy the above notion of CAED.

\subsubsection{Consistency-around-a-truncated-octahedron (CATO)}

The truncated octahedron is a polyhedron that has eight hexagonal faces and six quadrilateral faces, as shown in the diagram of Figure \ref{fig:CATO}.  The truncated octahedron is assigned an overdetermined system of fifty-four equations (eight systems of hex equations and six quad equations) for seventeen unknown variables.  If the overdetermined system of equations on the truncated octahedron has a consistent solution, then the system of equations will be said to satisfy consistency-around-an-truncated-octahedron (CATO)\footnote{In terms of permutahedra, an alternative name is consistency-around-a-$P_3$ (CA$P_3$). Similarly, the notion of CAH given by Theorem \ref{thm:CAH} also has the alternative name of CA$P_2$.}.


\begin{figure}[htb!]
\centering
\begin{tikzpicture}[line join=bevel,z=-5.5,scale=3.0]

\pgfmathsetmacro\tw{sqrt(2)}
\pgfmathsetmacro\sn{sin(44)}
\pgfmathsetmacro\cs{cos(44)}

\coordinate (A1) at ({-3/2},{-\sn/2},{-\sn/2});
\coordinate (A2) at ({-3/2},{\sn/2},{\sn/2});
\coordinate (A3) at ({-1},{-\sn+\sn/\tw},{-\sn-\sn/\tw});
\coordinate (A4) at ({-1},{-\sn-\sn/\tw},{-\sn+\sn/\tw});
\coordinate (A5) at ({-1},{\sn+\sn/\tw},{\sn-\sn/\tw});
\coordinate (A6) at ({-1},{\sn-\sn/\tw},{\sn+\sn/\tw});
\coordinate (A7) at ({-1/2},{-3*\sn/2},{-3*\sn/2});
\coordinate (A8) at ({-1/2},{-\sn/2+\sn*\tw},{-\sn-\sn*\tw});
\coordinate (A9) at ({-1/2},{-\sn/2-\sn*\tw},{-\sn+\sn*\tw});
\coordinate (A10) at ({-1/2},{\sn/2+\sn*\tw},{\sn-\sn*\tw});
\coordinate (A11) at ({-1/2},{\sn/2-\sn*\tw},{\sn+\sn*\tw});
\coordinate (A12) at ({-1/2},{3*\sn/2},{3*\sn/2});
\coordinate (A13) at ({1/2},{-3*\sn/2},{-3*\sn/2});
\coordinate (A14) at ({1/2},{-\sn/2+\sn*\tw},{-\sn-\sn*\tw});
\coordinate (A15) at ({1/2},{-\sn/2-\sn*\tw},{-\sn+\sn*\tw});
\coordinate (A16) at ({1/2},{\sn/2+\sn*\tw},{\sn-\sn*\tw});
\coordinate (A17) at ({1/2},{\sn/2-\sn*\tw},{\sn+\sn*\tw});
\coordinate (A18) at ({1/2},{3*\sn/2},{3*\sn/2});
\coordinate (A19) at ({1},{-\sn+\sn/\tw},{-\sn-\sn/\tw});
\coordinate (A20) at ({1},{-\sn-\sn/\tw},{-\sn+\sn/\tw});
\coordinate (A21) at ({1},{\sn+\sn/\tw},{\sn-\sn/\tw});
\coordinate (A22) at ({1},{\sn-\sn/\tw},{\sn+\sn/\tw});
\coordinate (A23) at ({3/2},{-\sn/2},{-\sn/2});
\coordinate (A24) at ({3/2},{\sn/2},{\sn/2});

\draw[double,->-=.6] (A5)--(A10);
\draw[-!-=.53] (A10)--(A16);
\draw[-!-=.53] (A16)--(A21);
\draw[double,-<-=.55] (A21)--(A18);
\draw (A18)--(A12);
\draw (A12)--(A5);

\draw[dashed,double,-<-=.55] (A4)--(A7);
\draw[dashed] (A7)--(A13);
\draw[dashed] (A13)--(A20);
\draw[double,->-=.6] (A20)--(A15) node[midway,below right=1pt]{$\gamma$};
\draw[-!-=.53] (A15)--(A9) node[midway,below=1pt]{$\alpha$};
\draw[-!-=.53] (A9)--(A4) node[midway,below left=1pt]{$\beta$};

\draw (A5)--(A2) node[midway,left=1pt]{$\rhot$}; \draw (A2)--(A1) node[midway,left=1pt]{$\rhob$}; \draw[double,->-=.6] (A1)--(A4) node[midway,left=1pt]{$\rho$};
\draw[dashed,-!-=.53] (A10)--(A8);
\draw[dashed,-!-=.53] (A8)--(A3);
\draw[dashed,double,-<-=.55] (A3)--(A7);
\draw[dashed,-!-=.53] (A16)--(A14);
\draw[dashed,double,-<-=.55] (A14)--(A19);
\draw[dashed] (A19)--(A13);
\draw[double,-<-=.55] (A21)--(A24);
\draw (A24)--(A23);
\draw (A23)--(A20);
\draw[double,->-=.6] (A18)--(A22);
\draw[-!-=.53] (A22)--(A17);
\draw[-!-=.53] (A17)--(A15);
\draw (A12)--(A6);
\draw[double,->-=.6] (A6)--(A11);
\draw[-!-=.53] (A11)--(A9);

\draw[dashed,double,->-=.6] (A1)--(A3);
\draw[dashed,-!-=.53] (A8)--(A14);
\draw[dashed] (A19)--(A23);
\draw[double,->-=.6] (A24)--(A22);
\draw[-!-=.53] (A17)--(A11);
\draw (A6)--(A2);

\filldraw[fill=black,draw=black] (A4) circle (0.75pt) node[left]{$\xa$};
\filldraw[fill=black,draw=black] (A7) circle (0.75pt) node[above right]{$\xb$};
\filldraw[fill=white,draw=black]  (A13) circle (0.75pt) node[above left]{$\xf$};
\filldraw[fill=white,draw=black] (A20) circle (0.75pt) node[below right]{$\xd$};
\filldraw[fill=black,draw=black] (A15) circle (0.75pt) node[below right]{$\xc$};
\filldraw[fill=black,draw=black] (A9) circle (0.75pt) node[below left]{$\xe$};

\filldraw[fill=white,draw=black] (A5) circle (0.75pt) node[above left]{$\ya$};
\filldraw[fill=white,draw=black] (A10) circle (0.75pt) node[above left]{$\yb$};
\filldraw[fill=white,draw=black] (A16) circle (0.75pt) node[above right]{$\yf$};
\filldraw[fill=white,draw=black] (A21) circle (0.75pt) node[above right]{$\yd$};
\filldraw[fill=white,draw=black] (A18) circle (0.75pt) node[below left]{$\yc$};
\filldraw[fill=white,draw=black] (A12) circle (0.75pt) node[below right]{$\ye$};

\filldraw[fill=black,draw=black] (A1) circle (0.75pt) node[left]{$\za$};
\filldraw[fill=white,draw=black] (A3) circle (0.75pt) node[right]{$\zb$};
\filldraw[fill=white,draw=black] (A19) circle (0.75pt) node[below=4pt]{$\zf$};
\filldraw[fill=black,draw=black] (A23) circle (0.75pt) node[right]{$\zd$};
\filldraw[fill=white,draw=black] (A17) circle (0.75pt) node[right]{$\zc$};
\filldraw[fill=black,draw=black] (A11) circle (0.75pt) node[above right]{$\ze$};

\filldraw[fill=white,draw=black] (A2) circle (0.75pt) node[left]{$\wa$};
\filldraw[fill=white,draw=black] (A8) circle (0.75pt) node[above right]{$\wb$};
\filldraw[fill=white,draw=black] (A14) circle (0.75pt) node[below left]{$\wf$};
\filldraw[fill=white,draw=black] (A24) circle (0.75pt) node[right]{$\wdd$};
\filldraw[fill=white,draw=black] (A22) circle (0.75pt) node[below right]{$\wc$};
\filldraw[fill=white,draw=black] (A6) circle (0.75pt) node[above=4pt]{$\we$};

\end{tikzpicture}

\caption{CATO or CA$P_3$. Initial variables are at black vertices and there are six different parameters $\al,\bt,\gm,\rho,\rhob,\rhot$ that are assigned to edges, where opposite edges of a face have the same parameter.}
\label{fig:CATO}
\end{figure}

The seven initial variables are chosen as
\begin{equation}
\xa,\xb,\xc,\xe,\za,\zd,\ze,
\end{equation}
as indicated by black vertices in Figure \ref{fig:CATO}.  For these initial conditions, the property of CATO for the truncated octahedron of Figure \ref{fig:CATO} can be checked with the following steps.
\begin{enumerate}[(i)]
\item
The system of hex equations 
\begin{equation}
\hexC{\xa}{\xb}{\xf}{\xd}{\xc}{\xe}{\alpha}{\beta}{\gamma}
\end{equation}
is used to uniquely solve for the two variables $\xf$ and $\xd$, the system of hex equations 
\begin{equation}
\hexC{\xa}{\za}{\wa}{\we}{\ze}{\xe}{\rhob}{\bt}{\rho}
\end{equation}
is used to uniquely solve for the two variables $\wa$ and $\we$, and the two quad equations
\begin{equation}
\begin{split}
\quadH{\zb}{\za}{\xb}{\xa}{\gm}{\rho}=0, \\
\quadQs{\ze}{\zc}{\xe}{\xc}{\al}{\rhob}=0,
\end{split}
\end{equation}
are used to uniquely solve for the two variables $\zb$ and $\zc$, respectively.
\item
The system of hex equations 
\begin{equation}
\hexC{\wc}{\wdd}{\zd}{\xd}{\xc}{\zc}{\rhob}{\rhot}{\gm}
\end{equation}
is used to uniquely solve for the two variables $\wdd$ and $\wc$, and the quad equation
\begin{equation}
\begin{split}
\quadQ{\zf}{\zd}{\xf}{\xd}{\bt}{\rhot}=0,
\end{split}
\end{equation}
is used to uniquely solve for the variable $\zf$.
\item
The system of hex equations 
\begin{equation}
\hexC{\ze}{\we}{\ye}{\yc}{\wc}{\zc}{\rhot}{\al}{\rho}
\end{equation}
is used to uniquely solve for the two variables $\yc$ and $\ye$, and the system of hex equations 
\begin{equation}
\hexC{\xb}{\zb}{\wb}{\wf}{\zf}{\xf}{\rhot}{\al}{\rho}
\end{equation}
is used to uniquely solve for the two variables $\wb$ and $\wf$.
\item
The system of hex equations 
\begin{equation}
\hexC{\zb}{\za}{\wa}{\ya}{\yb}{\wb}{\rhob}{\rhot}{\gm}
\end{equation}
is used to uniquely solve for the two variables $\ya$ and $\yb$, and the system of hex equations 
\begin{equation}
\hexC{\yd}{\wdd}{\zd}{\zf}{\wf}{\yf}{\rhob}{\bt}{\rho}
\end{equation}
is used to uniquely solve for the two variables $\yf$ and $\yd$.
\item
The system of hex equations 
\begin{equation}
\hexC{\yd}{\yc}{\ye}{\ya}{\yb}{\yf}{\alpha}{\beta}{\gamma}
\end{equation}
and the three quad equations 
\begin{equation}
\begin{split}
\quadH{\yd}{\yc}{\wdd}{\wc}{\gm}{\rho}=0, \\
\quadQs{\yb}{\yf}{\wb}{\wf}{\al}{\rhob}=0, \\
\quadQ{\ya}{\ye}{\wa}{\we}{\bt}{\rhot}=0,
\end{split}
\end{equation}
must each be automatically satisfied by the variables determined in the previous steps.
\end{enumerate}

Each of the combinations of equations given in Table \ref{tab:typeCpolytope} satisfy the above notion of CATO.

\subsubsection{Consistency-around-a-6-6-duoprism (CA66D)}

The 6-6 duoprism (also known as the hexagonal duoprism) is a 4-polytope that has twelve hexagonal faces and thirty-six quadrilateral faces and twelve hexagonal prism cells. A Schlegel diagram of the 6-6 duoprism is shown in Figure \ref{fig:CA66D}.

\begin{figure}[htb!]
\centering
\begin{tikzpicture}[line join=bevel,z=-5.5,scale=5.5]

\pgfmathsetmacro\tr{sqrt(3)/2}
\pgfmathsetmacro\sn{sin(-90)}
\pgfmathsetmacro\cs{cos(-90)}

\pgfmathsetmacro\fc{1/4}

\pgfmathsetmacro\fd{1/2}
\pgfmathsetmacro\sh{22/100}
\pgfmathsetmacro\snn{sin(-87)}
\pgfmathsetmacro\css{cos(-87)}

\coordinate (A0) at (1,{-\sn*1/2},{\cs*1/2});
\coordinate (A1) at (1/2,{\cs*\tr+\sn/2},{\sn*\tr-\cs*1/2});
\coordinate (A2) at (1/2,{\cs*\tr-\sn/2},{\sn*\tr+\cs*1/2});
\coordinate (A3) at (-1/2,{\cs*\tr+\sn/2},{\sn*\tr-\cs*1/2});
\coordinate (A4) at (-1/2,{\cs*\tr-\sn/2},{\sn*\tr+\cs*1/2});
\coordinate (A5) at (-1,{\sn*1/2},{-\cs*1/2});
\coordinate (A6) at (-1,{-\sn*1/2},{\cs*1/2});
\coordinate (A7) at (-1/2,{-\cs*\tr+\sn/2},{-\sn*\tr-\cs*1/2});
\coordinate (A8) at (-1/2,{-\cs*\tr-\sn/2},{-\sn*\tr+\cs*1/2});
\coordinate (A9) at (1/2,{-\cs*\tr+\sn/2},{-\sn*\tr-\cs*1/2});
\coordinate (A10) at (1/2,{-\cs*\tr-\sn/2},{-\sn*\tr+\cs*1/2});
\coordinate (A11) at (1,{\sn*1/2},{-\cs*1/2});

\coordinate (B0) at (\fc,{-\sn*\fc/2},{\cs*\fc/2});
\coordinate (B1) at (\fc/2,{\cs*\tr*\fc+\fc*\sn/2},{\sn*\tr*\fc-\cs*\fc/2});
\coordinate (B2) at (\fc/2,{\cs*\tr*\fc-\sn*\fc/2},{\sn*\tr*\fc+\cs*\fc/2});
\coordinate (B3) at (-\fc/2,{\cs*\tr*\fc+\sn*\fc/2},{\sn*\tr*\fc-\cs*\fc/2});
\coordinate (B4) at (-\fc/2,{\cs*\tr*\fc-\sn*\fc/2},{\sn*\tr*\fc+\cs*\fc/2});
\coordinate (B5) at (-\fc,{\sn*\fc/2},{-\cs*\fc/2});
\coordinate (B6) at (-\fc,{-\sn*\fc/2},{\cs*\fc/2});
\coordinate (B7) at (-\fc/2,{-\cs*\tr*\fc+\sn*\fc/2},{-\sn*\tr*\fc-\cs*\fc/2});
\coordinate (B8) at (-\fc/2,{-\cs*\tr*\fc-\sn*\fc/2},{-\sn*\tr*\fc+\cs*\fc/2});
\coordinate (B9) at (\fc/2,{-\cs*\tr*\fc+\sn*\fc/2},{-\sn*\tr*\fc-\cs*\fc/2});
\coordinate (B10) at (\fc/2,{-\cs*\tr*\fc-\sn*\fc/2},{-\sn*\tr*\fc+\cs*\fc/2});
\coordinate (B11) at (\fc,{\sn*\fc/2},{-\cs*\fc/2});

\coordinate (C0) at (\fd,{-\snn*\fd/2+\sh},{\css*\fd/2});
\coordinate (C1) at (\fd/2,{\css*\tr*\fd+\fd*\snn/2-\sh},{\snn*\tr*\fd-\css*\fd/2});
\coordinate (C2) at (\fd/2,{\css*\tr*\fd-\snn*\fd/2+\sh},{\snn*\tr*\fd+\css*\fd/2});
\coordinate (C3) at (-\fd/2,{\css*\tr*\fd+\snn*\fd/2-\sh},{\snn*\tr*\fd-\css*\fd/2});
\coordinate (C4) at (-\fd/2,{\css*\tr*\fd-\snn*\fd/2+\sh},{\snn*\tr*\fd+\css*\fd/2});
\coordinate (C5) at (-\fd,{\snn*\fd/2-\sh},{-\css*\fd/2});
\coordinate (C6) at (-\fd,{-\snn*\fd/2+\sh},{\css*\fd/2});
\coordinate (C7) at (-\fd/2,{-\css*\tr*\fd+\snn*\fd/2-\sh},{-\snn*\tr*\fd-\css*\fd/2});
\coordinate (C8) at (-\fd/2,{-\css*\tr*\fd-\snn*\fd/2+\sh},{-\snn*\tr*\fd+\css*\fd/2});
\coordinate (C9) at (\fd/2,{-\css*\tr*\fd+\snn*\fd/2-\sh},{-\snn*\tr*\fd-\css*\fd/2});
\coordinate (C10) at (\fd/2,{-\css*\tr*\fd-\snn*\fd/2+\sh},{-\snn*\tr*\fd+\css*\fd/2});
\coordinate (C11) at (\fd,{\snn*\fd/2-\sh},{-\css*\fd/2});

\draw(A0)--(C0)--(B0);
\draw(A2)--(C2)--(B2);
\draw(A4)--(C4)--(B4);
\draw(A6)--(C6)--(B6);
\draw(A8)--(C8)--(B8);
\draw(A10)--(C10)--(B10);

\draw(A1)--(C1)--(B1);
\draw(A3)--(C3)--(B3);
\draw(A5)--(C5);\draw(C5)--(B5) node[midway,above left]{$\rhot$};
\draw(A7)--(C7) node[midway,above left]{$\rhob$};\draw(C7)--(B7);
\draw(A9)--(C9)--(B9);
\draw(A11)--(C11)--(B11);

\draw (C6)--(C4);
\draw (C4)--(C2);
\draw (C2)--(C0);
\draw (C0)--(C10);
\draw (C10)--(C8);
\draw (C8)--(C6);

\draw (C5)--(C3);
\draw (C3)--(C1);
\draw (C1)--(C11);
\draw (C11)--(C9);
\draw (C9)--(C7);
\draw (C7)--(C5);

\draw (B6)--(B4);
\draw (B4)--(B2);
\draw (B2)--(B0);
\draw (B0)--(B10);
\draw (B10)--(B8);
\draw (B8)--(B6);

\draw (B5)--(B3);
\draw (B3)--(B1);
\draw (B1)--(B11);
\draw (B11)--(B9);
\draw (B9)--(B7);
\draw (B7)--(B5);

\draw (B6)--(B5);
\draw (B4)--(B3);
\draw (B2)--(B1);
\draw (B0)--(B11);
\draw (B10)--(B9);
\draw (B8)--(B7);

\draw (A6)--(A4);
\draw (A4)--(A2);
\draw (A2)--(A0);
\draw (A0)--(A10);
\draw (A10)--(A8);
\draw (A8)--(A6);

\draw (A5)--(A3);
\draw (A3)--(A1);
\draw (A1)--(A11);
\draw (A11)--(A9) node[midway,below right]{$\gm$};
\draw (A9)--(A7) node[midway,below]{$\al$};
\draw (A7)--(A5) node[midway,below left]{$\bt$};

\draw (A6)--(A5) node[midway,left]{$\rho$};
\draw (A4)--(A3);
\draw (A2)--(A1);
\draw (A0)--(A11);
\draw (A10)--(A9);
\draw (A8)--(A7);


\fill (B5) circle (0.4pt);
\fill (B3) circle (0.4pt);
\fill (B9) circle (0.4pt);
\fill (B7) circle (0.4pt);
\fill (B6) circle (0.4pt);
\fill (C6) circle (0.4pt);
\fill (A6) circle (0.4pt);

\end{tikzpicture}

\caption{CA66D. Initial variables are at black vertices and there are six different parameters $\al,\bt,\gm,\rho,\rhob,\rhot$ that are assigned to edges, where opposite edges of a face have the same parameter.}
\label{fig:CA66D}
\end{figure}

The 6-6 duoprism is assigned an overdetermined system of 108 equations (twelve systems of hex equations and thirty-six quad equations) for twenty-nine unknown variables.  If the overdetermined system of equations on the 6-6 duoprism has a consistent solution, then the system of equations will be said to satisfy consistency-around-a-6-6-duoprism (CA66D).

For simplicity, CA66D will only be considered in terms of systems of type-A hex equations in combination with type-Q ABS equations. First, denote by $\prsm{\bbx}{\bby}{\al}{\bt}{\gm}{\rho}$ the system of equations on the hexagonal prism of Figure \ref{fig:hexprismexample}, defined by
\begin{equation}\label{hexprism66}
\begin{alignedat}{2}
(PX.a)&&\qquad\hexA{\xa}{\xb}{\xf}{\xd}{\xc}{\xe}{\al}{\bt}{\gm}, \\
(PX.b)&&\qquad\hexA{\ya}{\yb}{\yf}{\yd}{\yc}{\ye}{\al}{\bt}{\gm}, \\
(PQ.a)&&\qquad\quadQ{\ya}{\yb}{\xa}{\xb}{\gm}{\rho}=0, \\
(PQ.b)&&\qquad\quadQ{\yb}{\yf}{\xb}{\xf}{\al}{\rho}=0, \\
(PQ.c)&&\qquad\quadQ{\yf}{\yd}{\xf}{\xd}{\bt}{\rho}=0, \\
(PQ.d)&&\qquad\quadQ{\yc}{\yd}{\xc}{\xd}{\gm}{\rho}=0, \\
(PQ.e)&&\qquad\quadQ{\ye}{\yc}{\xe}{\xc}{\al}{\rho}=0, \\
(PQ.f)&&\qquad\quadQ{\ya}{\ye}{\xa}{\xe}{\bt}{\rho}=0, 
\end{alignedat}
\end{equation}
where $\bm{A}=0$ denotes a system of type-A hex equations \eqref{Ahexdef} and $Q=0$ is a type-Q ABS equation. The equations $\bigl((PX.a),(PX.b),(PQ.a),\ldots,(PQ.f)\bigr)$ of $\bm{P}$ each depend on different combinations of the four parameters $\al,\bt,\gm,\rho$, and on different subsets of the twelve variables
\begin{equation}\label{xyvars}
\bbx=(\xa,\xb,\xf,\xd,\xc,\xe),\quad \bby=(\ya,\yb,\yf,\yd,\yc,\ye).
\end{equation}
The twelve hexagonal prism cells of a 6-6-duoprism are respectively assigned the equations
\begin{equation}\label{66equations}
\begin{split}
\prsm{\bbx}{\bby}{\al}{\bt}{\gm}{\rho},&\qquad
\prsm{\bm{X}_a}{\bm{X}_f}{\rhot}{\rhob}{\rho}{\bt}, \\
\prsm{\bby}{\bbz}{\al}{\bt}{\gm}{\rhob},&\qquad
\prsm{\bm{X}_f}{\bm{X}_e}{\rhot}{\rhob}{\rho}{\al}, \\
\prsm{\bbz}{\bbw}{\al}{\bt}{\gm}{\rhot},&\qquad
\prsm{\bm{X}_e}{\bm{X}_d}{\rhot}{\rhob}{\rho}{\gm}, \\
\prsm{\bbw}{\bbv}{\al}{\bt}{\gm}{\rho},&\qquad
\prsm{\bm{X}_d}{\bm{X}_c}{\rhot}{\rhob}{\rho}{\bt}, \\
\prsm{\bbv}{\bbu}{\al}{\bt}{\gm}{\rhob},&\qquad
\prsm{\bm{X}_c}{\bm{X}_b}{\rhot}{\rhob}{\rho}{\al}, \\
\prsm{\bbu}{\bbx}{\al}{\bt}{\gm}{\rhot},&\qquad
\prsm{\bm{X}_b}{\bm{X}_a}{\rhot}{\rhob}{\rho}{\gm},
\end{split}
\end{equation}
where $\bbx$ and $\bby$ are given in \eqref{xyvars}, and
\begin{equation}
\begin{gathered}
\bbw=(\wa,\wb,\wf,\wdd,\wc,\we),\quad
\bbz=(\za,\zb,\zf,\zd,\zc,\ze),\quad 
\bbu=(\ua,\ub,\uf,\ud,\uc,\ue),\quad \\
\bbv=(\va,\vb,\vf,\vd,\vc,\ve),\qquad
\bm{X}_i=(y_i,x_i,u_i,v_i,w_i,z_i),\; i\in\{a,b,c,d,e,f\}.
\end{gathered}
\end{equation}
The system of equations $\prsm{\bbw}{\bbv}{\al}{\bt}{\gm}{\rho}$ is assigned to the ``outer'' hexagonal prism that contains all other hexagonal prisms in the sense of the Schlegel diagram of Figure \ref{fig:CA66D}.  The twelve six-tuples of variables $\bbx,\bby,\bbz,\bbw,\bbv,\bbu$, and $\bm{X}_i$, $i\in\{a,b,c,d,e,f\}$, are respectively associated to vertices of one of the twelve hexagonal faces in the 6-6-duoprism.  

Consider two systems of equations $\prsm{\bbw}{\bbx}{\al}{\bt}{\gm}{\rho}$ and $\prsm{\bby}{\bbz}{\al}{\bt}{\gm}{\rho}$. If $\bbx=\bby$, the system of hex equations $(PX.b)$ from the former, and the system of hex equations $(PX.a)$ from the latter, are equivalent and lie on a shared hexagonal face associated to the variables $\bbx$.  Thus, the systems of equations on the left of \eqref{66equations} belong to a ``loop'' of six hexagonal prism cells that share hexagonal faces in the 6-6-duoprism, and the systems of equations on the right belong to the other loop of six hexagonal prism cells in the 6-6-duoprism.

To avoid repeating the steps of the CAHP algorithm, it is assumed that each of the twelve individual systems of equations \eqref{66equations} satisfy CAHP.  The seven initial variables are chosen as
\begin{equation}
\xa,\xb,\xc,\xe,\ya,\za,\wa,
\end{equation}
as indicated by black vertices in Figure \ref{fig:CA66D}.  For these initial conditions, the property of CA66D for the 6-6-duoprism of Figure \ref{fig:CA66D} can be checked with the following steps.
\begin{enumerate}[(i)]
\item
The systems of equations 
\begin{equation}
\prsm{\bbx}{\bby}{\al}{\bt}{\gm}{\rho},\quad 
\prsm{\bby}{\bbz}{\al}{\bt}{\gm}{\rhob},\quad
\prsm{\bbz}{\bbw}{\al}{\bt}{\gm}{\rhot},
\end{equation}
are used to uniquely determine each of the unknown variables in $\bbx,\bby,\bbz,\bbw$, and these systems are consistent as a consequence of CAHP.
\item
The systems of equations
\begin{equation}
\begin{split}
\prsm{\bm{X}_a}{\bm{X}_f}{\rhot}{\rhob}{\rho}{\bt}, \\
\prsm{\bm{X}_f}{\bm{X}_e}{\rhot}{\rhob}{\rho}{\al}, \\
\prsm{\bm{X}_e}{\bm{X}_d}{\rhot}{\rhob}{\rho}{\gm}, \\
\prsm{\bm{X}_d}{\bm{X}_c}{\rhot}{\rhob}{\rho}{\bt}, \\
\prsm{\bm{X}_c}{\bm{X}_b}{\rhot}{\rhob}{\rho}{\al}, 
\end{split}
\end{equation}
are used to uniquely determine each of the remaining unknown variables in $\bbu$ and $\bbv$, and these systems are consistent as a consequence of CAHP.
\item
The remaining equations that were not used to determine the unknown variables (and whose consistency is not implied through CAHP) must be consistent with the values of the unknown variables determined through the previous steps.  These equations are the three quad equations $(PQ.b),(PQ.c),(PQ.d)$, from
\begin{equation}
\prsm{\bm{X}_b}{\bm{X}_a}{\rhot}{\rhob}{\rho}{\gm},
\end{equation}
and the two systems of hex equations $(PX.a)$ and $(PX.b)$ from
\begin{equation}
\prsm{\bbv}{\bbu}{\al}{\bt}{\gm}{\rhob}.
\end{equation}
\end{enumerate}

Hex equations $\bm{A}=0$ and quad equations $Q=0$ constructed from pairs on the left of Table \ref{tab:CAFCCtoABS} satisfy the above notion of CA66D.

\section{Summary}

This paper has presented the new concept of a system of hex equations, which is an overdetermined system of six face-centered quad equations defined on six vertices of a hexagon.  Consistent systems of hex equations satisfy analogues of several of the properties that are characteristic of quad equations, including dependence on parameters which are the same on opposite edges of the hexagon, symmetries under rotations and reflections of the hexagon, well-defined initial value problems both on the hexagon itself and on connected staircases in the hexagonal lattice, and multidimensional consistency on polytopes with hexagonal faces.  This appears to be the first time that such properties have been collectively realised for systems of equations that evolve in hexagonal lattices.



The results of this paper open up interesting directions for future research, some of which include (in no particular order):
\begin{itemize}\setlength\itemsep{0em}
\item
Develop the analogue of consistency on polytopes for lattice models of statistical mechanics by using the connection to the quasi-classical limit of the Yang-Baxter equation.  This could potentially have some interesting physical interpretation for models on hexagonal lattices.  Reinterpreting the consistency of the hex equations on polytopes in terms of some equations satisfied by Boltzmann weights might also lead to new forms of the Yang-Baxter equations.
\item
Determine the solutions of the systems of hex equations. There are methods that have been developed for constructing the soliton solutions of integrable quad equations \cite{AHNsols3,HZsols1,ANsols1} and these could possibly be adapted for obtaining solutions of hex equations.
\item
Extend systems of hex equations to multicomponent cases by using the quasi-classical limit of multicomponent solutions of the Yang-Baxter equation \cite{Bazhanov:2011mz}.  This may also have connections with multicomponent discrete Landau-Lifschitz type equations \cite{DeliceNijhoffYooKong} that were proposed as extensions of $Q4$.
\item
Investigate reductions of the systems of hex equations on polytopes.  For example, reductions of quad equations on certain polyhedra have been shown to lead to discrete Painlev\'e equations \cite{JNS,JoshiNakazonoCUBOReduction}.  There could potentially be some interesting reductions from the permutahedron $P_{n-1}$ which may equivalently be interpreted as the Voronoi cell of the dual lattice of the root lattice $A_{n-1}$ \cite{ConwaySloaneSpherePackings}.
\item
Classification of hex equations that satisfy consistencty-around-a-hexagon (CAH).  CAH (six face-centered quad equations for two unknowns) is similar in complexity to CAC (six quad equations for four unknowns), although an individual face-centered quad equation involves an extra variable and parameters.  For quad equations there exists the important ABS classification \cite{ABS,ABS2}, and it might be feasible to use a similar approach to classify hex equations.   Such a classification should include all of the known type-A and type-C face-centered quad equations that satisfy CAFCC and might find some new ones.
\item
CAH for type-B CAFCC equations.  Besides type-A and type-C equations, there were also introduced the type-B equations in connection with CAFCC \cite{Kels:2020zjn}, but it is not known if there is a CAH property for type-B equations which could be used to construct the corresponding consistent systems of hex equations.  This might be obtained through a classification result as mentioned in the previous point.  
\end{itemize}

\section*{Acknowledgments}

The author thanks Adam Doliwa and Frank Nijhoff for helpful correspondence.

\begin{appendices}
\numberwithin{equation}{section}

\section{Quad equations and face-centered quad equations}\label{app:equations}

In the following, $\wp(z)$ and $\wp'(z)$ denote the Weierstrass elliptic function and its derivative respectively, with Weierstrass invariants $g_2$ and $g_3$, and $\zz{\alpha}$ and $\sh{\alpha}$ are defined by
\begin{equation}\label{def1}
\zz{\alpha}=\EXP^{\alpha},\qquad \sh{\alpha}=\sinh(\alpha).
\end{equation}

\subsection{Face-centered quad equations}

Two conventions have been used for face-centered quad equations in this paper.  In Section \ref{sec:quads} they are written with dependence on two-component parameters $\bm{\al}=(\al_1,\al_2)$ and $\bm{\bt}=(\bt_1,\bt_2)$ (which are more suitable for the square lattice), while in Section \ref{sec:hexagon} they are introduced for hex equations in terms of three scalar parameters $\al,\bt,\gm$ (more suitable for the hexagonal lattice).

In the following, the equations are given in terms of $\al,\bt,\gm$, where the relation between different sets of parameters is $\alpha={\beta_1}-{\alpha_1}$, $\beta={\beta_1}-{\alpha_2}$, $\gamma={\beta_1}-{\beta_2}$.  Also, $\phi$ is used to denote the following sum of parameters 
\begin{equation}\label{phidef}
    \phi=\alpha+\beta-\gamma
\end{equation}
and $\dxv$ is used to denote
\begin{equation}\label{def2}
\dxv=4\xv^3-g_2\xv-g_3.
\end{equation}

Define the following four functions
\begin{equation}
\begin{split}
\ta(\xv,\yv,\pa,\pb,\pc,\pd)&=
\pdb(\xv,\pa,\pb)\bigl(\pda(\xv,\pc)\pda(\xv,\pd)-256\wes(\xv,\pc)\wes(\xv,\pd)\yv\bigr), \\
\tb(\xv,\yva,\yvb,\pa,\pb,\pc,\pd)&=
\pdb(\xv,\pa,\pb)\bigl(\pda(\xv,\pc)\wes(\xv,\pd)\yva-\pda(\xv,\pd)\wes(\xv,\pc)\yvb\bigr), \\
\tc(\xv,\yva,\yvb,\pa,\pb,\pc,\pd)&=
\pdc(\xv,\pa,\pb,\pc)\bigl(16\wes(\xv,\pd)\yva-\pda(\xv,\pd)\yvb\bigr), \\
\td(\xv,\yv,\pa,\pb,\pc,\pd)&=
\pda(\xv,\pa)\wes(\xv,\pb)\Bigl(\dwepc\bigl(16\wes(\xv,\pd)\yv+\pda(\xv,\pd)\bigr) \\
&\hspace{3cm}+\dwepd\bigl(16\wes(\xv,\pc)\yv+\pda(\xv,\pc)\bigr)\Bigr),
\end{split}
\end{equation}
where
\begin{equation}\label{def3}
\begin{split}
\wes(\xv,\al)&=\bigl(\xv-\wp(\al)\bigr)^2,\\
\pdz(\xv,\al)&=2g_3+\bigl(\xv+\wp(\al)\bigr)\bigl(g_2-4\xv\wp(\al)\bigr), \\
\pda(\xv,\al)&=16\bigl(\xv+\wp(\al)\bigr)g_3 + \bigl(4\xv\wp(\al)+g_2\bigr)^2, \\
\pdb(\xv,\pa,\pb)&=\dwepb\pdz(\xv,\pa) + \dwepa\pdz(\xv,\pb), \\
\pdc(\xv,\pa,\pb,\pc)&= -4\dxv\prod_{1\leq j\leq3}\wp'(\al_j) -
\sum_{1\leq i\leq 3}\wp'(\al_1)\prod_{\substack{1\leq j\leq 3 \\ j\neq i}}\pdz(\xv,\al_j), \\
\pdd(\xv,\pa,\pb,\pc,\pd)&=
\sum_{1\leq i\leq4}\bigl(4\dxv\pdz(\xv,\al_i)\prod_{\substack{1\leq j\leq 4 \\ j\neq i}}\wp'(\al_j) +
\wp'(\al_i)\prod_{\substack{1\leq j\leq 4 \\ j\neq i}}\pdz(\xv,\al_j)\bigr).
\end{split}
\end{equation}

The following is a multilinear polynomial of the form \eqref{afflinpoly} that corresponds to the elliptic face-centered quad equation $A4$ 
\begin{equation}\label{a4}
\begin{alignedat}{2}
A4&& \bigl(32\xv^2g_3+(4x^2+g_2)^2-16\wp(\al-\bt)\dxv\bigr)\bigl(\wp(\gm)-\wp(\phi)\bigr)\\&&
\times\biggl(\ta(\xv,x_ax_bx_cx_d,-\pqaa,\pqba,\pqab,\pqbb) +
\ta(\xv,x_ax_bx_cx_d, \pqaa,\pqbb,\pqab,\pqba)\phantom{\!\!\bigr)} \\&&-
\ta(\xv,x_ax_bx_cx_d, \pqab,\pqba,\pqaa,\pqbb) -
\ta(\xv,x_ax_bx_cx_d,-(\pqab),\pqbb,\pqaa,\pqba)\phantom{\!\!\bigr)} \\&&- 16\Bigl(
\pdd(\xv, \pqaa, \pqab,\pqba,\pqbb)(x_bx_c-x_ax_d)+
\tb(\xv,x_d,x_c,\pqba,\pqbb,\pqaa,\pqab)(x_a+x_b) \phantom{\bigr)\!\!\bigr)} \\&& +
\pdd(\xv,-\pqaa,-(\pqab),\pqba,\pqbb)(x_ax_c-x_bx_d)+
\tb(\xv,x_a,x_b,\pqaa,\pqab,\pqbb,\pqba)(x_c+x_d)\phantom{\bigr)\!\!\bigr)} \\&&+
\tb(\xv,x_ax_b,x_cx_d,\pqaa,-(\pqbb),\pqab,\pqba) -
\tb(\xv,x_ax_b,x_cx_d,\pqaa, \pqba,\pqab,\pqbb)\phantom{\bigr)\!\!\bigr)} \\&&+
\tb(\xv,x_ax_b,x_cx_d,\pqab, \pqbb,\pqaa,\pqba) -
\tb(\xv,x_ax_b,x_cx_d,\pqab,-\pqba,\pqaa,\pqbb)\!\Bigr)\phantom{\!\!\bigr)} \\&&+ 4\Bigl(
\tc(\xv,x_bx_cx_d,x_a, \pqaa,\pqba, \pqbb,\pqab) -
\tc(\xv,x_bx_cx_d,x_a,-(\pqab),\pqba, \pqbb,\pqaa)\phantom{\bigr)\!\!\bigr)} \\&&+
\tc(\xv,x_ax_cx_d,x_b,-\pqaa,\pqba, \pqbb,\pqab) -
\tc(\xv,x_ax_cx_d,x_b, \pqab,\pqba, \pqbb,\pqaa)\phantom{\bigr)\!\!\bigr)} \\&&-
\tc(\xv,x_ax_bx_d,x_c, \pqaa,\pqab, \pqba,\pqbb) +
\tc(\xv,x_ax_bx_d,x_c, \pqaa,\pqab,-(\pqbb),\pqba)\phantom{\bigr)\!\!\bigr)} \\&&-
\tc(\xv,x_ax_bx_c,x_d, \pqaa,\pqab,-\pqba,\pqbb) +
\tc(\xv,x_ax_bx_c,x_d, \pqaa,\pqab, \pqbb,\pqba)\phantom{\bigr)\!\!\bigr)} \\&&-
\td(\xv,x_cx_d,\pqbb,\pqba,\pqaa,\pqab)x_a +
\td(\xv,x_cx_d,\pqba,\pqbb,\pqaa,\pqab)x_b\phantom{\bigr)\!\!\bigr)} \\&&+
\td(\xv,x_ax_b,\pqab,\pqaa,\pqba,\pqbb)x_c -
\td(\xv,x_ax_b,\pqaa,\pqab,\pqba,\pqbb)x_d
\Bigr)\!\!\biggr)
\end{alignedat}
\end{equation}
The above polynomial is degree 10 in $x$.  The product of factors $\bigl(32\xv^2g_3+(4x^2+g_2)^2-16\wp(\al-\bt)\dxv\bigr)\bigl(\wp(\gm)-\wp(\phi)\bigr)$ on the first line is degree 4 in $x$ and independent of $x_a,x_b,x_c,x_d$.    These factors are needed to preserve the third symmetry of  
\eqref{fcqsyms}.

The following are multilinear polynomials of the form \eqref{afflinpoly} for type-A and type-C face-centered quad equations

\begin{equation}\label{a3d}
\begin{alignedat}{2}
A3_{(\delta)}&&
xP^{A3}_{(\delta)}(x_a,x_b,x_c,x_d;\al,\bt,\gm)+\sh{\bt}(x_ax^2-x_bx_cx_d) \\&&\quad- \sh{\bt-\gm}(x_bx^2-x_ax_cx_d) - \sh{\al}(x_cx^2-x_ax_bx_d) + \sh{\al-\gm}(x_dx^2-x_ax_bx_c)  \\&&+ \delta\Bigl(
\sh{\bt-\gm}\sh{\al-\gm}(\sh{\al}x_a-\sh{\bt}x_c) +
 \sh{\al}\sh{\bt}(\sh{\gm-\al}x_b-\sh{\gm-\bt}x_d)
\Bigr)
\end{alignedat}
\end{equation}

\begin{equation}\label{a2dd}
\begin{alignedat}{2}
 A2_{(\delta_1;\,\delta_2)}&&\quad
xP^{A2}_{(\delta_1;\,\delta_2)}(x_a,x_b,x_c,x_d;\al,\bt,\gm)-\al(x_c-x_d)\bigl(x^2+x_ax_b-\delta_1\bt^2(x_a+x_b-\bt^2)^{\delta_2}\bigr) \\ &&+
\bt(x_a-x_b)\bigl(x^2+x_cx_d-\delta_1\al^2(x_c+x_d-\al^2)^{\delta_2}\bigr) \\&&+
\gm(x_b-x_d)\bigl(x^2+x_ax_c-\delta_1\al\bt(x_a+x_c-(\al-\bt)^2-\al\bt)^{\delta_2}\bigr) \\ &&+
\delta_1\gm\phi\Bigl( 
 \bigl(\al x_a-\bt x_c+\delta_2(\al-\bt)\al\bt\bigr)\bigl(x_b+x_d + (\al-\gm)(\gm-\bt)\bigr)^{\delta_2} \\ &&+
 \delta_2\bigl((\al-\bt)(x_ax_c-x^2)-\al^3x_a+\bt^3x_c\bigr)\Bigr)
\end{alignedat}
\end{equation}

\begin{equation}\label{c3ddd}
\begin{alignedat}{2}
C3_{(\delta_1;\,\delta_2;\,\delta_3)}&&  
xP^{C3}_{(\delta_1;\,\delta_2;\,\delta_3)}(x_a,x_b,x_c,x_d;\al,\bt,\gm) + \bigl(x_ax_b+2\delta_2\sh{\bt}\sh{\bt-\gm}\bigr)\bigl(x_d - x_c\zz{\gm}\bigr) \\ &&+ 
\bigl(x_d\zz{\gm}-x_c\bigr)x^2
+\delta_1\zz{-\al}\bigl(1+2\delta_2x_cx_d\zz{2\al-\gm}\bigr)\times \\ &&
\hspace{-0.4cm}\Bigl(\bigl(\zz{\bt}-2\delta_3\zz{2\al-\bt}x^2\bigr)\bigl(x_b\zz{\gm}-x_a\bigr) +
\bigl(x_a\zz{\gm}-x_b\bigr)\bigl(\zz{\gm-\bt}-2\delta_3\zz{\al+\phi}x^2\bigr)\Bigr)
\end{alignedat}
\end{equation}

\begin{equation}\label{c2ddd}
\begin{alignedat}{2}
 C2_{(\delta_1;\,\delta_2;\,\delta_3)} && 
xP^{C2}_{(\delta_1;\,\delta_2;\,\delta_3)}(x_a,x_b,x_c,x_d;\al,\bt,\gm) + (x^2-x_ax_b)\gm(\gm-2\al)^{\delta_3}-(x^2+x_ax_b)(x_c-x_d) \\ &&+
\delta_1\bt(\bt-\gm)\bigl(\gm(\gm-2\al)^{\delta_3}+(x_c-x_d)\bigr)\bigl(\bt(\gm-\bt)+x_a+x_b\bigr)^{\delta_2}-2\delta_3\gm(x_a+x_b)x^2 \\ &&
\hspace{-0.4cm}+\delta_1 \bigl(\bt(x_a-x_b)-\gm x_a\bigr)\Bigl(x_c+x_d+(2\al-\gm)(\gm-2\al)^{\delta_3}+2\delta_3\bigl(\al(\al-\gm)-x^2\bigr)\Bigr)(\gm-2\al)^{\delta_2} \\ &&+
2\delta_2\Bigl(\bt\gm(\bt-\gm)(\al+x_c)(\al-\gm+x_d) + \bigl(\gm x_a-\bt(x_a-x_b)\bigr)(\al(\gm-\al)+x_cx_d)\Bigr) 
\end{alignedat}
\end{equation}

\begin{equation}\label{c1d}
\begin{alignedat}{2}
C1_{(\delta)} && \quad xP^{C1}_{(\delta)}(x_a,x_b,x_c,x_d;\al,\bt,\gm) +
2\bigl(\bt(x_a-x_b)-\gm x_a\bigr)(-x_c/2-x_d/2)^\delta \\ && + \bigl(x^2+x_ax_b+\delta\bt(\gm-\bt)\bigr)(x_c-x_d)
\end{alignedat}
\end{equation}

The above polynomials are each quadratic in $x$, where the terms linear in $x$ (corresponding to $P_1$ in \eqref{afflinsum}) are given by
\begin{equation}\label{p1polys}
\begin{split}
P^{A3}_{(\delta)}&=(x_ax_c-x_bx_d)\sh{\al-\bt} + (x_cx_d-x_ax_b)\sh{\gm} + (x_bx_c-x_ax_d)\sh{\phi}+\delta \sh{\al-\bt}\sh{\phi}\sh{\gm}, \\
P^{A2}_{(\delta_1;\,\delta_2)}&= \al(x_c-x_d)(x_a+x_b+2\delta_2\bt^2) - \bt(x_a-x_b)(x_c+x_d+2\delta_2\al^2) \\ &\quad- \gm(x_b-x_d)(x_a+x_c+2\delta_2\al\bt), \\ &\quad
-\delta_1\phi\gm\bigl((\al-\bt)(x_a+x_b+x_c+x_d + \gm\phi-\al^2-\bt^2)^{\delta_2}-2\delta_2(\al x_a-\bt x_c)\bigr)\\
P^{C3}_{(\delta_1;\,\delta_2;\,\delta_3)}&=(x_ax_c - x_bx_d)\zz{\bt} + (x_bx_c - x_ax_d)\zz{\gm-\bt} \\ &\quad+ \delta_1\bigl(\zz{-\al}-2(\delta_3x_ax_b-\delta_2x_cx_d)\zz{\al-\gm}\bigr)\bigl(1 - \zz{2\gm}\bigr), \\
P^{C2}_{(\delta_1;\,\delta_2;\,\delta_3)}&=\bigl(x_a+x_b+2\delta_2(\bt-\gm)\bt\bigr)(x_c-x_d) + (x_a-x_b)(2\bt-\gm)(\gm-2\al)^{\delta_3} \\ &\quad+ \delta_1 \gm\bigl(x_c+x_d+(2\al-\gm)(\gm-2\al)^{\delta_3}\bigr)(\gm-2\al)^{\delta_2} +2\delta_2\gm\bigl((\al-\gm)\al-x_cx_d\bigr) \\ &\quad+2\delta_3\gm\bigl((\al-\bt)\phi+x_ax_b\bigr), \\
P^{C1}_{(\delta)}&=2\gm(-x_c/2-x_d/2)^\delta - (x_a + x_b)(x_c - x_d).
\end{split}
\end{equation}

\subsection{ABS quad equations}

The following are multilinear polynomials of the form \eqref{afflinquadpoly} for equations in the ABS list \cite{ABS,ABS2}, where the expression for $Q4$ is due to Nijhoff \cite{NijhoffQ4Lax}.
\begin{equation}
\begin{alignedat}{2}
Q4 &&\qquad \wp'(\al)\Bigl(\bigl(x_a-\wp(\bt)\bigr)\bigl(x_c-\wp(\bt)\bigr)-\bigl(\wp(\al)-\wp(\bt)\bigr)\bigl(\wp(\al-\bt)-\wp(\bt)\bigr)\Bigr) \\ &&\qquad \phantom{\wp'(\al)}\times
\Bigl(\bigl(x_d-\wp(\bt)\bigr)\bigl(x_b-\wp(\bt)\bigr)-\bigl(\wp(\al)-\wp(\bt)\bigr)\bigl(\wp(\al-\bt)-\wp(\bt)\bigr)\Bigr)
 \\ &&\qquad +
 \wp'(\bt)\Bigl(\bigl(x_a-\wp(\al)\bigr)\bigl(x_b-\wp(\al)\bigr)-\bigl(\wp(\bt)-\wp(\al)\bigr)\bigl(\wp(\al-\bt)-\wp(\al)\bigr)\Bigr) \\ &&\qquad \phantom{\wp'(\bt)}\times
\Bigl(\bigl(x_c-\wp(\al)\bigr)\bigl(x_d-\wp(\al)\bigr)-\bigl(\wp(\bt)-\wp(\al)\bigr)\bigl(\wp(\al-\bt)-\wp(\al)\bigr)\Bigr) \\ &&\qquad +
\wp'(\al)\wp'(\bt)\wp'(\al-\bt)\bigl(\wp(\al)-\wp(\bt)\bigr) \\
Q3_{(\delta)}&&\quad \sh{\al}(x_ax_b+x_cx_d)-\sh{\bt}(x_ax_c+x_bx_d)-\sh{\al-\bt}(x_ax_d+x_bx_c) -\delta\sh{\al}\sh{\bt}\sh{\al-\bt} \\
Q2&&\qquad Q1_{(0)}+\alpha\beta(\alpha-\beta)(x_a+x_b+x_c+x_d+\alpha\beta-\alpha^2-\beta^2) \\
Q1_{(\delta)}&&\qquad \alpha(x_a-x_c)(x_b-x_d)-\beta(x_a-x_b)(x_c-x_d)+\delta\alpha\beta(\alpha-\beta) \\
H3_{(\delta;\,\varepsilon)}&&\qquad \zz{\alpha}(x_ax_c+x_bx_d)-\zz{\beta}(x_ax_b+x_cx_d)-\sh{\al-\bt}\bigl(\delta(2+\varepsilon)+\varepsilon(\zz{\alpha_\beta} x_ax_d)\bigr) \\
H2_{(\varepsilon)}&&\qquad (x_a-x_d)(x_b-x_c)-(\alpha-\beta)(x_a+x_b+x_c+x_d)-(\alpha^2-\beta^2) \\ &&\qquad +\tfrac{\varepsilon}{2}(\alpha-\beta)\bigl((2x_a+\alpha+\beta+1)(2x_d+\alpha+\beta+1)+(\alpha-\beta)^2-1\bigr) \\
H1_{(\varepsilon)}&&\qquad (x_a-x_d)(x_b-x_c)+2(\alpha-\beta)-\varepsilon(\alpha-\beta)(2+x_a+x_d)
\end{alignedat}
\end{equation}

Apart from $Q4$, each of the above are equivalent to $P_1(x_a,-x_d,x_c,x_b;\al,\al-\bt,\al-\bt)$ in \eqref{afflinsum} for one of the polynomials in \eqref{p1polys}, where the correspondence is given in Table \ref{tab:CAFCCtoABS}.

\subsection{Three-leg and four-leg equations}

In the following, $\sigma(z)$ denotes the Weierstrass sigma function, and $\overline{x}$ is used to denote
\begin{equation}
\overline{x}=x+\sqrt{x^2-1}.
\end{equation}
The definitions in \eqref{def1}, \eqref{def2}, and \eqref{def3}, are also used.

The functions used for the three-leg and four-leg equations may be organised according to which face-centered quad equations they are associated to.  The four-leg equations associated to type-A and type-C face-centered quad equations are respectively given by \eqref{4lega} and \eqref{4legc} (or in additive form \eqref{4legadd} and \eqref{4legcdd}). First, for $A4$ one may use $a(x;y;\alpha)$ given by
\begin{equation}\label{a4q4leg}
    \ds
    a(x;y;\alpha)=
\frac{\Bigl(4\bigl(x+y+\wp(\al)\bigr)\wes(x,\al)+\bigl(\dxv^{\frac{1}{2}}+\wp'(\al)\bigr)^2\Bigr)}
{\Bigl(4\bigl(x+y+\wp(\al)\bigr)\wes(x,\al)+\bigl(\dxv^{\frac{1}{2}}-\wp'(\al)\bigr)^2\Bigr)}
\frac{\sigma\bigl(\wp(x)-\wp(\al)\bigr)}{\sigma\bigl(\wp(x)+\wp(\al)\bigr)}.
\end{equation}
The functions for all other face-centered quad equations are given in Table \ref{tab:4legquads}, where the abbreviation $add.$ is used to indicate it is for the additive forms of the equation.

\begin{table}[htb!]
\centering

\begin{tabular}{c|ccc|c}%




 


Type-A & $a(x;y;\alpha)$ & & Type-C & $c(x;y;\alpha)$ %
 
 \\
 
\cline{1-2}\cline{4-5}
 
 \\[-0.35cm]

$A3_{(1)}$ & $\ds
\frac{1+Z(\alpha)^2\overline{x}^2-2Z(\alpha)\overline{x}y}
     {Z(\alpha)^2+\overline{x}^2-2Z(\alpha) \overline{x}y}$
& &     
$C3_{(\frac{1}{2};\,\frac{1}{2};\,0)}$  
&
$\ds
\frac{1-Z(\alpha)\overline{x}y}{\overline{x}-Z(\alpha)y}$

\\[0.3cm]
\cline{1-2}\cline{4-5}
\\[-0.4cm]

\multirow{3}{*}[-0.3cm]{$A3_{(0)}$} & \multirow{3}{*}[-0.3cm]{$\ds
\frac{Z(\alpha) x-y}{x-Z(\alpha) y}$}
& &
$C3_{(\frac{1}{2};\,0;\,\frac{1}{2})}$
&
$Z(-\alpha)+Z(\alpha)x^2-2xy$

\\[0.1cm]
\cline{4-5}
\\[-0.4cm]

& & &
$C3_{(1;\,0;\,0)}$      
&
$xy-Z(-\alpha)$

\\[0.1cm]
\cline{4-5}
\\[-0.4cm]

& & &
$C3_{(0;\,0;\,0)}$ 
&
$y$

\\[0.1cm]
\cline{1-2}\cline{4-5}
\\[-0.35cm]

$A2_{(1;\,1)}$ & $\ds
\frac{(\sqrt{x}-\alpha)^2-y}{(\sqrt{x}+\alpha)^2-y}$
& &
$C2_{(1;\,1;\,0)}$
&
$\ds
\frac{y-\sqrt{x}+\alpha}{y+\sqrt{x}+\alpha}$

\\[0.3cm]
\cline{1-2}\cline{4-5}
\\[-0.4cm]

\multirow{3}{*}[-0.3cm]{$A2_{(1;\,0)}$} & \multirow{3}{*}[-0.3cm]{$\ds
\frac{x-y+\alpha}{x-y-\alpha}$}
& &
$C2_{(1;\,0;\,1)}$ 
&
$(x+\alpha)^2-y$

\\[0.1cm]
\cline{4-5}
\\[-0.4cm]

& & &
$C2_{(1;\,0;\,0)}$ 
&
$\ds
x+y+\alpha$

\\[0.1cm]
\cline{4-5}
\\[-0.4cm]

& & &
$C1_{(1)}$ 
& 
$y$

\\[0.1cm]
\cline{1-2}\cline{4-5}
\\[-0.3cm]

\multirow{2}{*}[-0.15cm]{$A2_{(0;\,0)}$} & \multirow{2}{*}[-0.15cm]{$\ds
\phantom{(add.) }\quad  \frac{\alpha}{x-y} \quad (add.)$}
& &
$C2_{(0;\,0;\,0)}$ 
&
$\phantom{(add.) }$ \,  $\ds\frac{y+\al}{2x}  \quad (add.)$

\\[0.2cm]
\cline{4-5}
\\[-0.3cm]

& & &
$C1_{(0)}$ 
& 
$\phantom{(add.) }\ds\quad  -\frac{y}{2} \quad (add.)$

\\[0.2cm]

\cline{1-2}\cline{4-5}
\end{tabular}
\caption{Left: 
A list of $a(x;y;\alpha)$ in \eqref{4lega} for type-A equations \eqref{a3d} and \eqref{a2dd}. Right: A list of $c(x;y;\alpha)$ in \eqref{4legc} for type-C equations \eqref{c3ddd}, \eqref{c2ddd}, and \eqref{c1d}.  For $C3_{(\delta_1;\,\delta_2;\,\delta_3)}$, $C2_{(\delta_1;\,\delta_2;\,\delta_3)}$, and  $C1_{(\delta_1)}$, the $a(x;y;\alpha,\beta)$ are respectively given by the functions for $A3_{(2\delta_2)}$, $A2_{(\delta_1;\,\delta_2)}$, and $A2_{(\delta_1;\,\delta_2=0)}$.}%
\label{tab:4legquads}
\end{table}

The three-leg equations for type-Q and type-H ABS quad equations are respectively given by \eqref{43legsQ} and \eqref{43legsH} (or for the additive form with \eqref{3legadd}), where the function for $Q4$ is given by \eqref{a4q4leg}, and the functions for all other equations are given in Table \ref{tab:3legquadsH} in terms of the functions from Table \ref{tab:4legquads}.\vspace{2cm}  

\begin{table}[htb!]
\centering
\begin{tabular}{c|ccc|c|c|c}%

Type-Q & $a(x;y;\alpha)$ & & Type-H & $a^\ast(x;y;\alpha)$ & $c(x;y;\alpha)$ & $c^\ast(x;y;\alpha)$
 
 \\
 
\cline{1-2}\cline{4-7}
 
 \\[-0.3cm]

$Q3_{(1)}$ & $A3_{(1)}$ & &
$H3_{(1;\,1)}$ & \multirow{3}{*}[-0.3cm]{$\ds A3_{(0)}$}
& 
$\ds C3_{(\frac{1}{2};\,\frac{1}{2};\,0)}$  
& $\ds
C3_{(\frac{1}{2};\,0\;\,\frac{1}{2})}$

\\[0.1cm]
 
\cline{1-2} \cline{4-4}\cline{6-7}
 
\\[-0.3cm]

\multirow{2}{*}[-0.15cm]{$Q3_{(0)}$} & \multirow{2}{*}[-0.15cm]{$A3_{(0)}$} & &
$H3_{(1;\,0)}$ & & \multicolumn{2}{|c}{$\ds
C3_{(1;\,;0;\,0)}$  }

\\[0.1cm]
 
 \cline{4-4}\cline{6-7}
 
\\[-0.3cm]

 & & &
$H3_{(0;\,0)}$ & & \multicolumn{2}{|c}{$\ds C3_{(0;\,;0;\,0)}$}

\\[0.1cm]
 
\cline{1-2}\cline{4-7}
 
\\[-0.3cm]

$Q2$ & $A2_{(1;\,1)}$ & &
$H2_{(1)}$ & \multirow{2}{*}[-0.15cm]{$\ds A2_{(1;\,0)}$}
& $\ds
C2_{(1;\,;1;\,0)}$
& $\ds
C2_{(1;\,;0;\,1)}$

\\[0.1cm]
 
\cline{1-2}\cline{4-4}\cline{6-7}
 
\\[-0.3cm]

\multirow{2}{*}[-0.15cm]{$Q1_{(1)}$} & \multirow{2}{*}[-0.15cm]{$A2_{(1;\,0)}$} & &
$H2_{(0)}$  & & \multicolumn{2}{|c}{$\ds
C2_{(1;\,;0;\,0)}$}

\\[0.1cm]
 
 \cline{4-7}
 
\\[-0.3cm]

 & & &
$H1_{(1)}$ & \multirow{2}{*}[-0.15cm]{$\ds A2_{(0;\,0)}$}
& $\ds
C1_{(1)}$
& $ \ds
C2_{(0;\,;0;\,0)}$

\\[0.1cm]
 
\cline{1-2}\cline{4-4}\cline{6-7}
 
\\[-0.3cm]

$Q1_{(0)}$ & $A2_{(0;\,0)}$ & &
$H1_{(0)}$ & & \multicolumn{2}{|c}{$ \ds
C1_{(0)}$}

\\[0.0cm]

\cline{1-2}\cline{4-7}
\end{tabular}
\caption{Left: 
A list of $a(x;y;\alpha)$ in \eqref{43legsQ} for type-Q ABS equations. Right: A list of $a^\ast(x;y;\alpha),c(x;y;\alpha),c^\ast(x;y;\alpha)$ in \eqref{43legsH} for type-H ABS equations.  For $H3_{(\delta;\,\varepsilon)}$ the $a(x;y;\alpha)$ are given by the functions for $Q3_{(\varepsilon)}$, for $H2_{(\varepsilon)}$ the $a(x;y;\alpha)$ is given by the function for $Q2$ if $\varepsilon=1$, and the function for $Q1_{(1)}$ if $\varepsilon=0$, and for $H1_{(\varepsilon)}$ the $a(x;y;\alpha)$ are given by the functions for $Q1_{(\varepsilon)}$.}
%
\label{tab:3legquadsH}
\end{table}


\end{appendices}


{\small

\bibliography{MComp}
\bibliographystyle{utphys}

}

\end{document}